\RequirePackage{fix-cm}
\documentclass[final,nomarks]{dmtcs-episciences}

\usepackage{mathptmx}
\usepackage{stmaryrd}
\usepackage{amsmath,amssymb,amsthm}
\usepackage{xspace}
\usepackage[shortlabels]{enumitem}
\usepackage{empheq}
\usepackage{hyperref}
\usepackage[capitalize]{cleveref}
\usepackage{doi}

\usepackage{xcolor}
\definecolor{LinkColor}{rgb}{.3,.3,.6}
\hypersetup{colorlinks=true,linkcolor=LinkColor,citecolor=LinkColor,filecolor=LinkColor,menucolor=LinkColor,urlcolor=LinkColor}

\newcommand{\M}{\ensuremath{\mathbf{M}}} % model
\newcommand{\U}{\ensuremath{\mathbf{U}}} % UCA
\newcommand{\I}{\ensuremath{\mathbf{I}}} % IG
\newcommand{\A}{\ensuremath{\mathbf{A}}} % Arcs
\newcommand{\F}{\ensuremath{\mathbf{F}}} % Full system
\renewcommand{\S}{\ensuremath{\mathbf{S}}} % synthetic
\newcommand{\R}{\ensuremath{\mathbf{R}}} % reduced graph

\newcommand{\N}{\ensuremath{\mathbf{N}}} % near neigh
\newcommand{\W}{\ensuremath{\mathbf{W}}} % walks
\newcommand{\T}{\ensuremath{\mathbf{T}}} % copy of walks
\newcommand{\Q}{\ensuremath{\mathbf{Q}}} % copy of walks
\renewcommand{\Q}{\ensuremath{\mathbf{K}}} % copy of walks
\newcommand{\G}{\ensuremath{\mathbf{G}}} % greedy walks

\renewcommand{\L}{\ensuremath{\mathbf{L}}} % rows

\DeclareMathOperator{\Ext}{ext}
\DeclareMathOperator{\Sep}{sep}
\DeclareMathOperator{\Jmp}{jmp}
\DeclareMathOperator{\Bal}{bal}
\DeclareMathOperator{\Ratio}{ratio}
\DeclareMathOperator{\RATIO}{RATIO}
\DeclareMathOperator{\Length}{len}
\DeclareMathOperator{\Lex}{lex}

\newcommand{\LL}{L}
\newcommand{\RR}{R}
\newcommand{\LN}{\N^-}
\newcommand{\RN}{\N^+}
\newcommand{\Fl}{F_l}
\newcommand{\Fr}{F_r}
\newcommand{\Hl}{H_l}
\newcommand{\Hr}{H_r}

\DeclareMathOperator{\Sign}{sgn}
\DeclareMathOperator{\Row}{row}
\DeclareMathOperator{\Rows}{rows}
\DeclareMathOperator{\Col}{col}
\DeclareMathOperator{\Cols}{cols}
\DeclareMathOperator{\Pos}{pos}
\DeclareMathOperator{\Gr}{Gr}
\DeclareMathOperator{\mult}{\times}

\newcommand{\Rep}{\textsc{Rep}\xspace}
\newcommand{\RepUIG}{\textsc{RepUIG}\xspace}
\newcommand{\nMult}[1]{\textsc{$#1$-Mult}\xspace}
\newcommand{\kMult}{\nMult{k}}
\newcommand{\klcMult}{\nMult{(k,\Circ,\Len)}}

\newcommand{\wrap}{\omega}

\newcommand{\Unit}{\U}
\newcommand{\arclength}[1]{|#1|}

\newcommand{\Circ}{c}
\newcommand{\Len}{\ell}

\newcommand{\Extra}{e}
\newcommand{\Syn}{\S}
\newcommand{\Red}{\R}
\newcommand{\ARow}{\L}
\newcommand{\Copies}{\lambda}
\newcommand{\Spiral}{\gamma}
\newcommand{\Unroll}{\cdot}

\newcommand{\nose}{\mu}
\newcommand{\hollow}{\eta}
\newcommand{\Bool}{\beta}
\newcommand{\Boolean}[1]{\Bool\{#1\}}

\newcommand{\CRange}[1]{\llbracket #1\rrbracket}
\newcommand{\Range}[1]{\llbracket #1\rrparenthesis}
\newcommand{\ORange}[1]{\llparenthesis #1\rrparenthesis}
\newcommand{\Dist}[1]{\ensuremath{{\rm\bf d}{#1}}}
\newcommand{\IDist}[1]{\ensuremath{{\rm\bf d^*}{#1}}}

\newtheorem{theorem}{Theorem}{\bf}{\it}
\newtheorem{lemma}{Lemma}{\bf}{\it}
\newtheorem{corollary}{Corollary}{\bf}{\it}
{\bf}{\it}
{\itshape}{\rmfamily}
\newtheorem{factproof}{Proof of Facts}{\itshape}{\rmfamily}

\newlist{discription}{enumerate}{2}
\setlist[discription]{align=left,leftmargin=\parindent,labelsep=*,itemindent=!,labelindent=0pt,parsep=0pt,topsep=2pt,itemsep=2pt}

\setlist[enumerate]{align=left,leftmargin=\parindent,labelsep=*,itemindent=!,labelindent=0pt,parsep=0pt,topsep=2pt,itemsep=2pt}

\title[Loop unrolling of UCA models]{Loop unrolling of UCA models: distance labeling}        % if too long for running head

\author[F. J. Soulignac and P. Terlisky]{Francisco J.\ Soulignac \and Pablo Terlisky}

\affiliation{Departamento de Computación, Facultad de Ciencias Exactas y Naturales and CONICET-Instituto de Investigación en Ciencias de la Computación (ICC), Universidad de Buenos Aires, Buenos Aires, Argentina.}

\begin{document}

\publicationdata{vol. 28:3}{2026}{19}{10.46298/dmtcs.9131}{2022-02-24; 2022-02-24; 2026-04-02}{2026-10-15}

%Hack to make the heading publication data appear in black (suggested by an LLM)
\makeatletter
\expandafter\let\expandafter\dmtcs@oldfirsthead\csname @firsthead\endcsname
\gdef\@firsthead{\begingroup\hypersetup{hidelinks}\dmtcs@oldfirsthead\endgroup}
\makeatother

\maketitle

\begin{abstract}
  \vspace{0.2cm} 
  A \emph{proper circular-arc (PCA) model} is a pair $\mathbf{M} = (C, \mathbf{A})$ where $C$ is a circle and $\mathbf{A}$ is a family of inclusion-free arcs on $C$ whose extremes are pairwise different.  The model $\mathbf{M}$ \emph{represents} a digraph $D$ that has one vertex $v(A)$ for each $A \in \mathbf{A}$ and one edge $v(A) \to v(B)$ for each pair of arcs $A,B \in \mathbf{A}(\mathbf{M})$ such that the beginning point of $B$ belongs to $A$.  For $k \geq 0$, the \emph{$k$-th power} $D^k$ of $D$ has the same vertices as $D$ and $v(A) \to v(B)$ is an edge of $D^k$ when $A\neq B$ and the distance from $v(A)$ to $v(B)$ in $D$ is at most $k$.  A \emph{unit circular-arc (UCA) model} is a PCA model $\mathbf{U} = (C,\mathbf{A})$ in which all the arcs have the same length $\ell+1$. If $\ell$, the length $c$ of $C$, and the extremes of the arcs of $\mathbf{A}$ are integer, then $\mathbf{U}$ is a \emph{$(c,\ell+1)$-CA model}.  For $i \geq 0$, the model $i \times \mathbf{U}$ of $\mathbf{U}$ is obtained by replacing each arc $(s,s+\ell+1)$ with the arc $(s,s+i\ell+1)$.  If $\mathbf{U}$ represents a digraph $D$, then $\mathbf{U}$ is \emph{$k$-multiplicative} when $i \times\mathbf{U}$ represents $D^i$ for every $0 \leq i \leq k$.  In this article we design a linear time algorithm to decide if a PCA model $\mathbf{M}$ is equivalent to a $k$-multiplicative UCA model when $k$ is given as input.  The algorithm either outputs a $k$-multiplicative UCA model $\mathbf{U}$ equivalent to $\mathbf{M}$ or a negative certificate that can be authenticated in linear time.  
  Our main technical tool is a new characterization of those PCA models that are equivalent to $k$-multiplicative UCA models.  For $k=1$, this characterization yields a new algorithm for the classical representation problem that is simpler than the previously known algorithms.
\end{abstract}

\section{Introduction}
\label{sec:introduction}

The last decade saw an increasing amount of research on numerical representation problems for unit circular-arc (UCA) models and some of its subclasses~\cite{CostaDantasSankoffXuJBCS2012,DuranFernandezSlezakGrippoSouzaOliveiraSzwarcfiterDAM2017,KlavikKratochvilOtachiRutterSaitohSaumellVyskocilA2017,LinSoulignacSzwarcfiter2009,LinSzwarcfiterSJDM2008,SoulignacJGAA2017,SoulignacJGAA2017a}.  In these problems we are given a proper circular-arc (PCA) model $\M$ and we have to find a UCA model $\Unit$, related to $\M$, that satisfies certain numerical constraints.  The paradigmatic example is the \emph{classical representation problem} in which a UCA model $\Unit$ equivalent to an input PCA model $\M$ has to be computed.  The equivalence of $\M$ and $\Unit$ means that the endpoints of $\Unit$ must appear in the same circular order as those of $\M$.

In this article we consider a generalization of the classical representation problem.  In a nutshell, given a PCA model $\M$ and $k \geq 0$, the goal is to find a UCA model $\Unit$ whose ``multiplication'' $i \times \Unit$ is equivalent to the ``power'' $\M^i$ of $\M$ for every $0 \leq i \leq k$.  Here, $i \mult \Unit$ is obtained from $\Unit$ by lengthening each arc to have length $i\times\Len+1$, where $\Len+1$ is the length of the arcs in $\Unit$.  On the other hand, $\M^i$ is a PCA model whose intersection graph is the $i$-th power of the intersection graph of $\M$.  To formally state the problem we require some terminology that will be used throughout the article.

\subsection{Statement of the problem}

In this work, the term \emph{arc} refers to open circular arcs.  For points $s \neq t$ of a circle $C$, we write $(s,t)$ to denote the arc of $C$ that goes from $s$ to $t$ in a clockwise traversal of $C$.  Each arc $A = (s,t)$ of $C$ with \emph{extremes} $s$ and $t$ is described by its \emph{beginning point} $s(A) = s$ and its \emph{ending point} $t(A) = t$.  We write $\arclength{A} = \arclength{s,t}$ and $\arclength{C}$ to denote the \emph{lengths} of $A$ and $C$, respectively.  We assume that every circle $C$ has a special point $0$ such that $p = \arclength{0,p}$ for every point $p \in C$.  Thus, $p < q$ if and only if $p$ appears before $q$ in a clockwise traversal of $C$ from $0$.  For arcs $A_1$ and $A_2$ of $C$, we write $A_1 < A_2$ to mean that $s(A_1) < s(A_2)$.  We classify the arcs of $C$ as being \emph{external} or \emph{internal} according to whether $A \cup \{t(A)\}$ contains $0$ or not, respectively.  In other words, $A$ is external when $t(A) < s(A)$.

A \emph{proper circular-arc (PCA) model} (\cref{fig:example-uig-pca}) is a pair $\M = (C, \A)$ where $C$ is a circle and $\A$ is a family of inclusion-free arcs on $C$, no two of which share an extreme.  We write $C(\M) = C$ and $\A(\M) = \A$ to denote the circle and the family of arcs of $\M$, respectively.  The \emph{extremes} of $\M$ are those extremes of the arcs in $\A$.  Say that $\M$ and a PCA model $\M'$ are \emph{equivalent} if there exists a bijection $f\colon\A(\M) \to \A(\M')$ such that $e(A) < e'(B)$ if and only if $e(f(A)) < e'(f(B))$, for $e, e' \in \{s,t\}$.  Colloquially, $\M$ and $\M'$ are equivalent if their extremes appear in the same order, regarding $f$, when their circles are traversed clockwise from their respective $0$ points.

\begin{figure}[b]
  \mbox{}\hfill \includegraphics{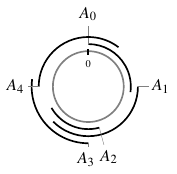} \hfill \includegraphics{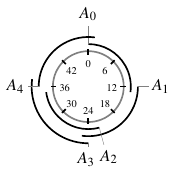} \hfill \includegraphics{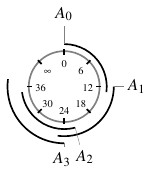} \hfill\mbox{} \raisebox{2mm}{\includegraphics{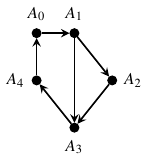}} 
  
  \caption[]{From left to right: a PCA model $\M$ with arcs $A_0 < \ldots < A_4$; a $(48,13)$-CA model $\Unit$ equivalent to $\M$; the $13$-IG model obtained by removing the external arc of $\Unit$ (the circle represents $\mathbb{R}$); the digraph $D(\M)$.}\label{fig:example-uig-pca}
\end{figure}

A \emph{unit circular-arc} (UCA) model is a PCA model $\M$ whose arcs all have the same length $\Len$.  If every extreme of $\M$ is integer, then we refer to $\M$ as being a \emph{$(\arclength{C}, \Len)$-CA} model (\cref{fig:example-uig-pca}).  Note that if $\M$ is a PCA model with no external arcs, then we can remove a segment $(|C(\M)|-\varepsilon, 0)$ from $C(\M)$ to obtain a line $L$, without removing points of the arcs of $\M$.  Replacing $L$ with the real line (of infinite length), we obtain a new representation of $\M$ where each arc corresponds to an interval of the real line.  Conversely, any family of intervals on the real line can be transformed into arcs of a circle by pasting together two points of the line that bound all the intervals.   To keep a uniform terminology for both PCA and proper intervals models, in this work we say that $\I$ is a \emph{proper interval (PIG)} or \emph{unit interval (UIG)} model to mean that $(\mathbb{R}, \I)$ is a PCA or UCA model with no external arcs, respectively, where the real line $\mathbb{R}$ is thought of as a circle with infinite length (\cref{fig:example-uig-pca}).  Moreover, instead of stating that $(\mathbb{R},\I)$ is an $(\infty, \Len)$-CA model, we simply state that $\I$ is an $\Len$-IG model.

Every PCA model $\M$ defines a digraph $D(\M)$ that has a vertex $v(A)$ for each $A \in \A(\M)$ where $v(A) \to v(B)$ is a directed edge ($A,B \in \A(\M)$) if and only if $s(B) \in A$ (\cref{fig:example-uig-pca}).  In the underlying graph $G(\M)$ of $D(\M)$, $v(A)$ and $v(B)$ are adjacent if and only if $A \cap B \neq \emptyset$.  A (di)graph $G$ is a \emph{proper circular-arc (PCA)} (di)graph \emph{represented} by $\M$ when $G$ is isomorphic to $G(\M)$.  \emph{Unit circular-arc} (UCA), \emph{proper interval} (PIG) and \emph{unit interval} (UIG) (di)graphs are defined analogously.  Because of the circular nature of $\M$, the distance between $v(A)$ and $v(B)$ in $G(\M)$ is the minimum of the distances in $D(\M)$ from $v(A)$ to $v(B)$ and from $v(B)$ to $v(A)$ (e.g.~\cite[Lemmas 5 and 6]{GavoillePaulSJDM2008}).  Thus, to determine the distance between two vertices of a graph represented by a PCA model $\M$, it suffices to find the distances of their respective vertices in $D(\M)$.  And, as $D(\M)$ is implicitly encoded by $\M$, we can work directly with $\M$.

Let $\M$ be a PCA model with arcs $A_0 < \ldots < A_{n-1}$.  The arc $A_0$ is called the \emph{initial} arc of $\M$.  Any sequence $\L = A_i, A_{i+1}, \ldots, A_{i+k}$, with subindices modulo $n$, is said to be \emph{contiguous}. The arcs $A_i$ and $A_{i+k}$ are the \emph{leftmost} and \emph{rightmost} arcs of $\L$, respectively.  For $0 \leq i < n$, define: 
\begin{itemize}
 \item $\LN[A_i]$ as the contiguous sequence of arcs with ending point in $A_i \cup \{t(A_i)\}$ that has $A_i$ as its rightmost arc,
 \item $\RN[A_i]$ as the contiguous sequence of arcs with beginning point in $A_i \cup \{s(A_i)\}$ that has $A_i$ as its leftmost arc,
 \item $\Fl(A_i)$ as the leftmost arc in $\LN[A_i]$ and $\Fr(A_i)$ as the rightmost arc in $\RN[A_i]$,
 \item $\LL(A_i) = A_{i-1}$ and $\RR(A_i) = A_{i+1}$ (modulo $n$),
 \item $\Hl(A_i)$ as the unique arc $A$ such that $\Fr(A) = A_i$ and $\Fr \circ \RR(A) \neq A_i$; if $A$ does not exist, then $\Hl(A_i) = \bot$, and
 \item $\Hr(A_i)$ as the unique arc $A$ such that $\Fl(A) = A_i$ and $\Fl \circ \LL(A) \neq A_i$; if $A$ does not exist, then $\Hr(A_i) = \bot$.
\end{itemize}

In \cref{fig:example-uig-pca}, $\RN[A_1] = \LN[A_3] = A_1,A_2,A_3$, $\RR(A_1) = \LL(A_3) = A_2$, $\Fr(A_1) = \Fl(A_4) = A_3$, $\Hr(A_1) = \Hl(A_3) = A_2$, and $\Hl(A_{2}) = \Hr(A_{2}) = \bot$.  Note that $v(A_i) \to v(A_j)$ is a directed edge of $D(\M)$ if and only if $i \neq j$ and $A_i \in \LN[A_j]$, which happens if and only if $i \neq j$ and $A_j \in \RN[A_i]$.  Therefore, $\LN[A_i]$ and $\RN[A_i]$ represent the in and out closed neighborhoods of $v(A_i)$ in $D(\M)$, respectively.  

% The following observation provides alternative definitions for $\Hr$ and $\Hl$.

% \begin{observation}
%  $\M$ is a PIG model if and only if $\Fl(A) = A$ for the initial arc $A$.
% \end{observation}

% \begin{observation}\label{obs:f vs h}
%  $\Hl(A)$ is the rightmost arc of the contiguous sequence\/ $\{B \in \A(\M) \mid \Fr(B) = A\}$, whereas $\Hr(A)$ is the leftmost arc of the contiguous sequence\/ $\{B \in \A(\M) \mid \Fl(B) = A\}$.
 % if and only if $\Fr(\Fl(A)) = A$, while $\Hr(A) \neq \bot$ if and only if $\Fl(\Fr(A)) = A$.  Moreover, if $\Hl(A) \neq \bot$, then $\Fr(\Hl(A)) = A$, whereas if $\Hr(A) \neq \bot$, then $\Fl(\Hr(A)) = A$.
% \end{observation}

% \begin{proof}
%  The sequence $\L = \{B \in \A(\M) \mid \Fr(B) = A\}$ is contiguous because $\M$ is PCA.  If $\Hl(A) \neq \bot$, then $\Fr(\Hl(A)) = A$ and $\Fr(\RR(\Hl(A)) \neq A$, thus $\Hl(A)$ is the rightmost arc of $\L$.  Conversely, if $\Hl(A) = \bot$, then $\L = \emptyset$, thus its rightmost arc is undefined.
 %Certainly, $\L = \{B \in \A(\M) \mid \Fr(B) = A\}$ is a contiguous sequence.  If $\Hl(A) \neq \bot$, then $\Fr(\Hl(A)) = A$, thus $\L \neq \emptyset$ and, moreover, $\Fl(A)$ is the leftmost arc in $\L$.  Conversely, if $\Fr(\Fl(A)) = A$, then $\L \neq \emptyset$, thus $\Hl(A)$ is the rightmost arc in $\L$.  The proof that $\Fl(\Hr(A)) = \Fl(\Fr(A)) = A$ if and only if $\Hr(A) \neq \bot$ is analogous.
% 6\end{proof}

For a (di)graph $D$ with vertex set $V(D)$, its \emph{$k$-th power} $D^k$ is the (di)graph with vertex set $V(D)$ such that $v \to w$ is a (directed) edge of $D^k$ if and only if $v \neq w$ and the distance from $v$ to $w$ in $D$ is at most $k$.  Let $D^k(\M) = (D(\M))^k$. The known fact that $D^k(\M)$ is a PCA digraph can be proved with the following construction.  For $k \geq 0$, $A \in \A(\M)$, and $f \in \{\LL, \RR, \Fl, \Fr, \Hl, \Hr\}$, let $f^0(A) = A$ and $f^{k+1}(A) = f \circ f^{k}(A)$, where $f^{k+1}(A) = \bot$ if $f^k(A) = \bot$.  The \emph{$k$-th power} of $A$ is the arc $A^k = (s(A), s(\Fr^{k}(A)) + \epsilon(x+1))$, where $x$ is the number of arcs $B$ with $s(A) \in B$ and $\Fr^k(B) = \Fr^k(A)$, and $\varepsilon < n^{-1}$ is small enough so that $t(A^i) \not\in \RR \circ \Fr^i(A)$.  Note that $A^k$ and $B^k$ share no extremes for $A, B\in \A$. The \emph{$k$-th power} of $\M$ is the pair $\M^k = (C(\M), \{A^k \mid A \in \A(\M)\})$; see \cref{fig:example-mult-power}.

\begin{figure}[t]
 \mbox{}\hfill \includegraphics{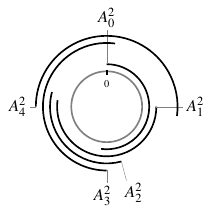} \hfill \includegraphics{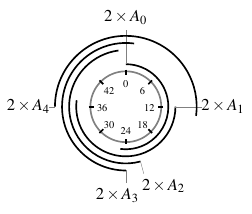} \hfill \includegraphics{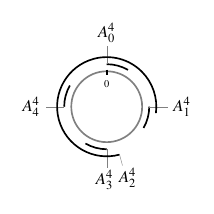} \hfill\mbox{}
 \caption[]{From left to right: $\M^2$; $2 \times \Unit$; and $\M^4$ for $\M$ and $\Unit$ in \cref{fig:example-uig-pca}.  Model $\Unit$ is $2$-multiplicative and $\wrap(\U)=4$.}\label{fig:example-mult-power}
\end{figure}

Define the \emph{wraparound value} $\wrap(\M)$ of $\M$ to be the minimum $\wrap > 1$ for which there is an arc $A \in \A(\M)$ such that $A^\wrap \subset A$. For the sake of notation, we usually omit the parameter $\M$ of $\wrap$. Note that $\wrap$ is well defined unless $\M$ is a PIG model, in which case we let $\wrap = n$.  Although $\M^k$ is defined for proving that $D^k(\M)$ is PCA, the statement $D(\M^k) = D^k(\M)$ is false when $k \geq \omega$ and $\M$ is not PIG.  Indeed, as $A^{\wrap} \subset A$ for some arc $A$, $A^{\wrap}$ intersects fewer arcs than $A$, whereas $v(A)$ has more neighbors in $D^\wrap(\M)$ than in $D(\M)$.  The reason why this happens is that $A^{\wrap-1} \cup \Fr^\wrap(A)$ covers the circle.  To fix this issue it can be observed that $D^{\wrap}(\M)$ is a complete digraph, thus it suffices to define $A^k = (s(A), s(A) -\varepsilon)$ when $k \geq \wrap$.  In this article we are concerned with the model, thus we avoid this approach.  Nevertheless, the following well-known theorem holds.

\begin{theorem}[\cite{FlotowDAM1996}]\label{thm:power models}
  Let\/ $\M$ be a PCA model.  If\/ $0 \leq k < \wrap$, then\/ $\M^k$ is a PCA model that represents $D^k(\M)$; otherwise, $D^k(\M)$ is a complete digraph.
\end{theorem}

If we store $\M^k$ for every $k < \wrap$, then we can efficiently answer any distance query in $G(\M)$.  Our goal, however, is to define one UCA model $\M$ to answer these queries efficiently.  Let $\M$ be a $(\Circ,\Len+1)$-CA model.  For $i \geq 0$ and $A \in \A(\M)$, define the \emph{$i$-multiple} of $A$ as the arc $i \mult A = (s(A), s(A) + i\Len+1 \bmod \Circ)$.  The \emph{$i$-multiple} of $\M$ is the pair $i\mult \M = (C(\M), \{i \mult A \mid A \in \A(\M)\})$; see %refchange-R2
\cref{fig:example-mult-power}.  For $k \geq 0$, we say that $\M$ is \emph{$k$-multiplicative} when $i\mult\M$ is a UCA model equivalent to $\M^i$ for every $0 \leq i \leq k$ (and thus it represents $D^i(\M)$ for $i$ up to $k$).  We remark that $i\mult\M$ is a UCA model unless two arcs have a common extreme.  To avoid this possibility, say that a $(\Circ,\Len+1)$-CA model $\M$ is \emph{even} when $\Circ$, $\Len$, and all the beginning points of the arcs in $\A(\M)$ are even.  It is not hard to see that $i \mult \M$ is an even UCA model when $\M$ is even.  There is no loss of generality in restricting the study to even models, as every $(\Circ,\Len)$-CA model $\M$ can be transformed into an equivalent $(2\Circ, 2\Len+1)$-CA model by replacing every arc $A$ by the arc $(2s(A), 2t(A)+1)$.  Note that $O(1)$ time is enough to decide if the distance from $v(A)$ to $v(B)$ in $D(\M)$ is $i \leq k$ when a $k$-multiplicative $(\Circ,\Len+1)$-CA model $\M$ is given, as it suffices to check that $s(A), s(B), t(i \mult A) = s(A) + i\Len+1 \bmod \Circ$ appear in this order in a clockwise traversal of $C(\M)$.

% In this article we study two computational problems: \kMult and \MultNumber.  Say that a PCA model $\M$ is \emph{normal} when no two arcs of $\A(\M)$ cover the circle.  Given a normal PCA model $\M$ and $k \geq 0$, the goal of \kMult is to determine if $\M$ is equivalent to a $k$-multiplicative model.  If affirmative, a \emph{certifying} algorithm outputs a $k$-multiplicative model $\Unit$; otherwise, it outputs a negative certificate.  Given a normal PCA model $\M$, the goal of \MultNumber is to determine the maximum $k$ such that $\M$ is equivalent to a $k$-multiplicative model.  We remark that such a value $k$ is well defined for every normal PCA model $\M$ where, perhaps, $k = \infty$.  The reason why those PCA models that are not normal are excluded as input of both problems has to do with the fact that if a non-normal PCA model $\M$ is equivalent to $1$-multiplicative model (i.e., UCA), then $G(\M)$ is a complete graph.  Hence, instead of using $\M$ as input, we can use a PIG model $\M'$ representing the complete graph $G(\M)$.  Interestingly, $\M'$ is $\infty$-multiplicative.  See also the discussion in~\cite{SoulignacJGAA2017}.

In this article we study the \kMult problem. Given a PCA model $\M$ and $0 \leq k < \wrap$, the goal of \kMult is to determine if $\M$ is equivalent to a $k$-multiplicative model.  If affirmative, a \emph{certifying} algorithm outputs a $k$-multiplicative model $\Unit$ equivalent to $\M$; otherwise, it outputs a negative certificate. The problem is trivial when $k = 0$ because $\M$ is $0$-multiplicative.  For this reason, we restrict our attention to the case $k > 0$ in which $\U$ must be UCA.

\subsection{Motivation for the problem}

Every PIG model $\I$ yields a metric $d_{\I}$ on $V(G(\I))$ where $d_{\I}(v(A),v(B)) = |s(A) - s(B)|$ for every $A, B \in \I$.  Among all the PIG representations of $G = G(\I)$, those that are UIG provide a better notion of \emph{nearness}, as the vertices adjacent in $G$ are nearer than those non-adjacent.  Indeed, if $\Unit$ is an $\Len$-IG model and $d_G(v(A),v(B)) \leq 1 < d_G(v(X),v(Y))$, then $d_{\Unit}(v(A),v(B)) < \Len < d_{\Unit}(v(X),v(Y))$.  This feature is one of the main reasons why (some notion equivalent to) UIG models are introduced in many different theoretical frameworks, including uniform arrays~\cite{GalanterPR1956,Goodman1977a}, semiorders~\cite{LuceE1956,ScottSuppesJSL1958} and indifference graphs~\cite{Roberts1969}.  As argued by %refchange-R4
Goodman in~\cite{Goodman1977a}, $d_{\Unit}$ reflects the natural idea that among all the pairs %refchange-R3
of adjacent vertices of $G$, some are nearer than others.  This is important in Goodman's work about the topology of quality, as large gaps in $d_{\Unit}$ may suggest that some qualia are yet undiscovered.  However, when greater distances on $G$ are considered, the main feature of UIG models is lost: there are UIG models $\Unit$ with $d_G(v(A), v(B)) \leq k < d_G(v(X), v(Y))$ and $d_{\Unit}(v(A), v(B)) > d_{\Unit}(v(X), v(Y))$.  Instead, if $\Unit$ is $(\Len+1)$-IG and $\infty$-multiplicative, then $d_{\Unit}(v(A),v(B)) < k\Len+1 < d_{\Unit}(v(X),v(Y))$.  \Cref{fig:proportionality} depicts the situation for general PCA models.

\begin{figure}
 \mbox{}\hfill\includegraphics{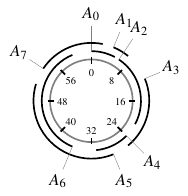}\hfill\includegraphics{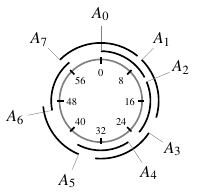}\hfill\includegraphics{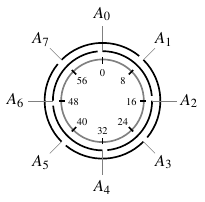}\hfill\mbox{}
	\caption{From left to right: a PCA model $\M$, a UCA model $\Unit_1$ equivalent to $\M$, and a $6$-multiplicative UCA model $\Unit_6$ equivalent to $\M$.  In $\M$, $s(A_2) - s(A_0) = 7 < 18 = s(A_7) - s(A_6)$ and $d(v(A_0), v(A_2)) = 2 > 1 = d(v(A_7), v(A_6))$.  In $\Unit_1$, $s(A_4)-s(A_0) = 25 < 30 = s(A_7) - s(A_4)$ and $d(v(A_0), v(A_4)) = 4 > 3 = d(v(A_4), v(A_7))$. In turn, $\Unit_6$ preserves the proportions for every $k < \wrap = 7$.}\label{fig:proportionality}
\end{figure}

Robert's ``PIG=UIG'' theorem~\cite{Roberts1969} %refchange-R4
states that every PIG model is equivalent to a UIG model. The classical representation problem \RepUIG asks to find a UIG model $\Unit$ equivalent to an input PIG model $\I$.  By definition, $\Unit$ is a UIG model if and only if $\Unit$ is $1$-multiplicative.  Thus, \RepUIG is simply the restriction of \nMult{1} to PIG inputs.  There are many algorithms to solve \RepUIG, at least three of which run in linear time~\cite{CorneilKimNatarajanOlariuSpragueIPL1995,LinSoulignacSzwarcfiter2009,Mitas1994}.  It is not hard to prove that the UIG models produced by the linear time algorithms in~\cite{CorneilKimNatarajanOlariuSpragueIPL1995} %refchange-R4
and~\cite{LinSoulignacSzwarcfiter2009} are $\infty$-multiplicative. (An implicit proof for the algorithm in~\cite{CorneilKimNatarajanOlariuSpragueIPL1995} follows from Gavoille and Paul~\cite{GavoillePaulSJDM2008}; %refchange-R4
see also Theorem~\ref{thm:equivalence}.)  The following generalization of Roberts' ``PIG=UIG'' theorem is obtained.

\begin{theorem}\label{thm:multiplicative PIG}%solved for pig
 Every PIG model\/ $\I$ is equivalent to some $\infty$-multiplicative UIG model.  Furthermore, an $\infty$-multiplicative UIG model equivalent\/ $\I$ can be computed in linear time.
\end{theorem}

The problem \kMult shares a strong relation to the \emph{distance labeling} problem for circular-arc graphs.  The latter problem asks to assign a label $L(v)$ to each vertex $v$ of a graph $G$ in such a way that the adjacency between $v$ and $w$ in $G^k$ can be determined from $L(v)$ and $L(w)$ alone, i.e., $d_G(v,w) = f(L(v),L(w))$ for some function $f$.  The primary goal is to minimize the number of bits required by each label $L(v)$, the secondary goal is to minimize the time required by $f$, and the third goal is to minimize the time required to compute $L$ from $G$.  If $\Unit$ is an $\infty$-multiplicative $(\Len+1)$-IG model representing $G = G(\Unit)$, then we can assign the label $L(v(I)) = s(I)$ for every $I \in \Unit$ because $d_G(v,w) = \lceil|L(v) - L(w)|/\Len\rceil$.  When $\Unit$ is produced using the algorithm in~\cite{CorneilKimNatarajanOlariuSpragueIPL1995}, %refchange-R4
each interval has a length $n$ whereas each beginning point is a number in $[1,n^2]$, thus each label requires at most $2\lceil\log(n)\rceil$ bits which is asymptotically optimal.  Moreover, $d_G(v,w)$ can be computed in $O(1)$ time, whereas $L$ can be computed in linear time.  Essentially, this is the labeling scheme proposed in~\cite{GavoillePaulSJDM2008} %refchange-R4
for PIG graphs, even though they do not mention that the generated labels are the possible beginning points of a UIG model.  In \cite{GavoillePaulSJDM2008}, Gavoille andd Paul also %refchange-R4
show that this scheme can be applied to solve the labeling problem for the general class of circular-arc graphs.  However, contrary to our goal in this article, the labels generated for UCA graphs have little to do with the UCA models representing them.

Theorem~\ref{thm:multiplicative PIG} yields an $O(n)$-time %refchange-R5
algorithm to find a UIG model $\Unit$, equivalent to an input PIG model $\M$, that implicitly encodes a UIG model $i \times \Unit$ of $G^i(\M)$ for every $i \geq 0$.  There is no hope in finding a similar algorithm when the input $\M$ is UCA because $G^i(\M)$ need not be UCA.  By Theorem~\ref{thm:power models}, this implies the well known fact that PCA and UCA are different classes of graphs \cite{TuckerDM1974}. The classical representation problem \Rep asks to determine if an input PCA model $\M$ is equivalent to some UCA model.  A UCA model equivalent to $\M$ or a negative certificate should be given as well.  Observe that a PCA model is UCA if and only if it is $1$-multiplicative.  Thus, \kMult is a natural generalization of \Rep = \nMult{1} that asks for a UCA model $\Unit$, if existing, to implicitly encode the UCA model $i \times \Unit$ of $G^i(\M)$ for every $0 \leq i \leq k$.  The problem \nMult{1} can be solved in linear time using any of the algorithms for \Rep~\cite{KaplanNussbaumDAM2009,LinSzwarcfiterSJDM2008,SoulignacJGAA2017a}.  As far as our knowledge extends, no efficient algorithms are known to solve \kMult for $k > 1$.  

\subsection{Brief history of the problems}

The problem \kMult is a generalization of \Rep that, in turn, is a generalization of \RepUIG.  One of the earliest references to \RepUIG was given by Goodman in the 1940s~\cite{Goodman1977a}, %refchange-R4
predating the current definition of UIG graphs.  Since then, several algorithms to solve \RepUIG were developed, many of which run in linear time (e.g.~\cite{CorneilKimNatarajanOlariuSpragueIPL1995,LinSoulignacSzwarcfiter2009,Mitas1994}).  Regarding \Rep, Goodman states that no adequate rules to transform a PCA model into an equivalent UCA model are known.  Of course, such general rules do not exist because some PCA graphs are not UCA.  Tucker characterized those PCA graphs that are not UCA by showing a family of forbidden induced subgraphs~\cite{TuckerDM1974}.  %refchange-R4
His proof yields an effective method to transform a PCA model $\M$ into an equivalent UCA model $\Unit$.  The first linear time algorithm to solve $\Rep$ was given by Lin and Szwarcfiter~\cite{LinSzwarcfiterSJDM2008}.  %refchange-R4
Their algorithm outputs a UCA model $\Unit$ equivalent to the input PCA model $\M$ when the output is yes, but it fails to provide a negative certificate when the output is no.  A different algorithm to find such a negative certificate was developed by Kaplan and Nussbaum~\cite{KaplanNussbaumDAM2009}, %refchange-R3
who left open the problem of finding a \emph{unified} certifying algorithm for \Rep; such an algorithm was given by Soulignac~\cite{SoulignacJGAA2017,SoulignacJGAA2017a}.%refchange-R4

From a technical point of view, our manuscript can be though of as the sixth on a series of articles that deal with $\RepUIG$ and $\Rep$.  The series started when Pirlot proved that every PIG model $\I$ is equivalent to a minimal UIG model~\cite{PirlotTaD1990}.  %refchange-R4
Although Pirlot's work is not of an algorithmic nature, his results yield an $O(n^2\log n)$-time %refchange-R5
algorithm to solve the minimal representation problem. %, in which a minimal UIG model equivalent to an input PIG model $\M$ is to be found.  
As part of his work, Pirlot shows that the problem of computing an $\Len$-IG model equivalent to $\I$, when $\Len$ is given, can be modeled with a system $S_\Len$ having $O(n)$ difference constraints.  A solution to $S_\Len$, if existing, can be found in $O(n^2)$ time by running a shortest path algorithm on its weighted constraint graph $\Syn(\Len)$ (see Theorem~\ref{thm:difference constraints}).  As every PIG graph is equivalent to an $n$-IG model, an $O(n^2)$-time %refchange-R5
algorithm to solve \Rep is obtained.  

The unweighted version $\Syn$ of $\Syn(\Len)$ is a succinct representation of $\M$ and, for this reason, Pirlot refers to $\Syn$ as the \emph{synthetic graph} of $\M$.  Mitas continued the series by arguing that the minimal representation problem can be solved in $O(n)$ time~\cite{Mitas1994}.  %refchange-R4
Although her algorithm has a flaw and the correct version runs in $O(n^2)$ time~\cite{SoulignacJGAA2017a}, it correctly solves \RepUIG in $O(n)$ time.  Her algorithm follows by observing that $\Syn$ admits a peculiar plane drawing in which the vertices occupy the entries of an imaginary matrix.  

Klavik et al. rediscovered and extended Pirlot's system $S_\Len$ to solve the bounded representation problem for UIG models in nearly quadratic time~\cite{KlavikKratochvilOtachiRutterSaitohSaumellVyskocilA2017}.  %refchange-R4
Later, Soulignac generalized $S_\Len$ to a new system $S_{\Circ,\Len}$ to solve the problem of deciding if $\M$ is equivalent to a $(\Circ,\Len)$-CA model when $\M$, $\Circ$, and $\Len$ are given as input~\cite{SoulignacJGAA2017,SoulignacJGAA2017a}.  %refchange-R4
The algorithm runs in $O(n^2)$ time and it can be adapted to solve the bounded representation problem for UCA models in $O(n^2)$ time as well.  Furthermore, Soulignac adapted Mitas' drawings to UCA models to design a certifying algorithm for \Rep that runs in $O(n)$ time or logspace.   Moreover, he proved that every UCA model is equivalent to some minimal UCA model, though he left open the problem of computing such a minimal model in polynomial time.  

Besides the previous five works, other articles apply systems of difference constraints to solve numerical representation problems related to intervals and circular-arc graphs (e.g.~\cite{BalofDoignonFioriniO2013,BoyadzhiyskaIsaakTrenkAe2017a,HamburgerMcConnellPorSpinradXu2018}).

\subsection{Our contributions}

In this article we follow the path described above.  In Section~\ref{sec:synthetic graph}, we define a system $S^k_{\Circ,\Len}$ with $O(n)$ difference constraints to solve \kMult for the particular case in which the output $\Unit$ is required to be a $(\Circ,\Len+1)$-CA model.  The algorithm obtained runs in $O(n^2)$ time.  In Section~\ref{sec:mitas drawing}, we study the structure of the unweighted graph $\Syn^k$ that represents $S^k_{\Circ,\Len}$.  As part of this section we provide an analogous of Mitas' drawings for $\Syn^k$.  In Section~\ref{sec:recognition} we exploit these drawings to devise a simple $O(n)$-time algorithm %refchange-R5
to solve \kMult.  The algorithm in this section outputs a negative certificate when the answer is no.  Theorems \ref{thm:equivalence}~and~\ref{thm:fast equivalence} are the main theoretical contributions in Section~\ref{sec:recognition}, as they provide characterizations of those PCA models that have equivalent $k$-multiplicative UCA models.  As far as our knowledge extends, these theorems are new even for $k=1$, and they yield the simplest algorithm currently known to solve \Rep.  Finally, in Section~\ref{sec:representation} we show how to build a $k$-multiplicative UCA model equivalent to $\M$ when the answer to \kMult is yes.  Theorem~\ref{thm:tucker} is the main theoretical contribution in this section, as it gives us an alternative characterization of those PCA models that are equivalent to $k$-multiplicative UCA models.  For $k=1$, Theorem~\ref{thm:tucker} is a restatement of a theorem by Soulignac~\cite{SoulignacJGAA2017} %refchange-R4
that, in turn, is a generalization Tucker's characterization.

The algorithm to transform a PCA model $\M$ into an equivalent $k$-multiplicative UCA model in Section~\ref{sec:representation} is rather similar to the one given in \cite{SoulignacJGAA2017a}.  %refchange-R4
However, the theoretical framework developed to prove that the algorithm is correct is new.  For instance, Theorems \ref{thm:equivalence}~and~\ref{thm:fast equivalence} in Section~\ref{sec:recognition} are new and can be applied to solve other open problems, such as the minimal representation problem, in polynomial time.  These applications, preliminarily described in~\cite{SoulignacTerliskyC2017}, will be discussed in forthcoming articles.  The major difference with respect to previous contributions is that we exploit a powerful geometric framework arising from the combination of Mitas' drawings and the loop unrolling technique.  When $\M$ is a PIG model, the Mitas' drawing of its synthetic graph $\Syn^k$ is a plane drawing for every $k \geq 0$; this property is lost when $\M$ is a PCA model.  The problem is that the edges of $\Syn^k$ corresponding to external arcs of $\M^k$ cross other edges.  This is hard to deal with when $k$ is large, as many arcs in $\M^k$ are external.  To apply the loop unrolling technique (\cref{fig:unrolling}), the idea is to replicate $\Copies$ times the arcs of a PCA model $\M$, for a sufficiently large $\Copies$.  This yields a new model $\Copies\Unroll\M$ in which every arc of $\M^k$ has several internal copies.  Interestingly, the PIG model $\M'$ obtained after removing all the arcs of $\Copies\Unroll\M$ that are external in $(\Copies\Unroll\M)^k$ has enough information to solve \kMult.  The idea, then, is to study the Mitas' drawing of the synthetic graph of $\M'$ as if it were a planar representation of $\Syn^k$.

\begin{figure}
   \mbox{}\hfill\includegraphics{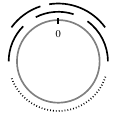}\hfill\includegraphics{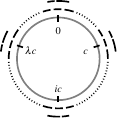}\hfill\includegraphics{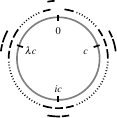}\hfill\mbox{}
	\caption{From left to right: a PCA model $\M$, its loop unrolling $\Copies \Unroll \M$ with $\Copies $ copies, and the model $\M$ obtained by removing the external arcs of $\Copies \Unroll \M$. Every external arc of $\M$ has many internal copies in $\Copies \Unroll\M$.}	
	\label{fig:unrolling}
\end{figure}

It is important to remark that none of the previous tools is required to formally state our characterizations; all of our results can be easily translated to the idiom of PCA models.  Clearly, synthetic graphs and Mitas' drawings are not new concepts, while loop unrolling is a natural and old technique that, unsurprisingly, has already been applied to circular-arc models (e.g.~\cite{WerraEisenbeisLelaitMarmolDAM1999}).  Yet, we are not aware of any work that combines them together to obtain structural results about PCA and UCA models.

\section{Preliminaries}

This section recalls how to solve a system of difference constraints and it introduces the remaining non-standard definitions that we use throughout the article.  For $k \in \mathbb{N}$, we write $\ORange{k} = (0,k) \cap \mathbb{Z}$, $\Range{k} = \ORange{k} \cup \{0\}$ and $\CRange{k} = \Range{k} \cup \{k\}$.  For a logical predicate $b$, we write $\Boolean{b} \in \{0,1\}$ to denote $1$ if and only if $b$ is true. For partial functions $f, g\colon \mathbb{R} \to \mathbb{R}$, we say that $f$ is \emph{bounded below} by $g$, and that $g$ is \emph{bounded above} by $f$, when $f(x) \geq g(x)$ for every $x \in \mathbb{R}$ that belongs to the domain of both $f$ and $g$.  We sometimes write $vw$ or $v \to w$ to denote an ordered pair $(v,w)$.  As usual, we reserve $n$ to denote the number of arcs of an input PCA model $\M$.

A walk $W$ in a digraph $D$ is a sequence of vertices $v_0, \ldots, v_k$ such that $v_iv_{i+1}$ is an edge of $D$ for $i \in \Range{k}$.  Walk $W$ goes \emph{from} (or \emph{begins at}) $v_0$ \emph{to} (or \emph{ends at}) $v_k$.  If $v_k = v_0$, then $W$ is a \emph{circuit}, if $v_i \neq v_j$ for every $0 \leq i < j \leq k$, then $W$ is a \emph{path}. If %refchange-R6
$W$ is a circuit and $v_0, \ldots, v_{k-1}$ is a path, then $W$ is a \emph{cycle}.  If $W' = v_k, \ldots, v_j$ is a walk, then $W + W' = v_1, \ldots, v_j$ is also walk. If %refchange-R7
$W$ is a circuit, then $j \Unroll W = \sum_{i=1}^j W$ is also a circuit for every $j \geq 1$, and if $D$ has no cycles, then $D$ is \emph{acyclic}.  For the sake of notation, we say that $W$ is a \emph{circuit} when $v_0 \neq v_k$ to mean that $W, v_0$ is a circuit.  

An \emph{edge weighting}, or simply a \emph{weighting}, of a digraph $D$ is a function $w\colon E(D) \to \mathbb{G}$ where $\mathbb{G}$ is a totally ordered additive group.  The value $w(e)$ is referred to as the \emph{weight} of $e$ (with respect to $w$).  For any multiset of edges $E$, the \emph{weight} of $E$ (with respect to an edge weighting $w$) is $w(E) = \sum_{e \in E}w(e)$.  We use two distance measures on a digraph $D$ with a weighting $w$.  For vertices $u, v$, we denote by $\Dist{w}(D, u, v)$ the maximum $w(W)$ among the walks $W$ from $u$ to $v$, while $\IDist{w}(D, u,v)$ denotes the maximum $w(W)$ among the paths $W$ starting at $u$ and ending at $v$.  Note that $\IDist{w}(D,u,v) < \infty$ for every $u,v$, while $\Dist{w}(D, u, v) = \IDist{w}(D, u, v)$ when $D$ contains no cycle of positive weight \cite[Section~24.1]{CormenLeisersonRivestStein2009}.  For the sake of notation, we omit the parameter $D$ when no ambiguities are possible.

A \emph{system of difference constraints} is a system $S$ with $m$ linear inequalities and one equation over a set $x_0, \ldots, x_{n-1}$ of indeterminates. The unique equation of $S$ is $x_0 = 0$ while each of the \emph{difference constraints} is an inequality of the form $x_j \geq x_i + c_{ij}$ for $i,j\in\Range{n}$, where $c_{ij}$ is a constant.  For each $i \in \Range{n}$, $i\neq0$, one of the inequalities is the \emph{non-negativity constraint} $x_i \geq x_0 + 0$.  The system $S$ defines a \emph{constraint digraph} $D$ with $n$ vertices and $m$ edges that has a weighting $\Sep$.  The digraph $D$ has a vertex $v_i$ \emph{corresponding} to $x_i$, $i \in \Range{n}$, and an edge $v_iv_j$ with weight $\Sep(v_iv_j) = c_{ij}$ \emph{corresponding} to each inequality $x_j \geq x_i + c_{ij}$ of $S$.  Vertex $v_0$ is the \emph{initial vertex} of $D$.  Clearly, $S$ is fully \emph{determined} by $D$, $\Sep$, and $v_0$.  The following well-known theorem gives a method to solve $S$.

\begin{theorem}[e.g.~{\cite[Theorem 24.9]{CormenLeisersonRivestStein2009}}]\label{thm:difference constraints}
 Let $D$ be the constraint digraph of a system of difference constraints $S$ with indeterminates $x_0, \ldots, x_{n-1}$.  Then, $S$ has a feasible solution if and only if $\Sep(W) \leq 0$ for every cycle $W$ of $D$.  Moreover, if $S$ has a feasible solution, then $x_i = \Dist{\Sep}(v_0, v_i)$ is a feasible solution to $S$.
\end{theorem}

If $S$ has $m$ constrains, then the Bellman-Ford algorithm applied to $D$ outputs in $O(nm)$ time a set of values for $x_0, \ldots, x_{n-1}$ or a cycle $W$ of $D$ with $\Sep(W)>0$. In the former case, we refer to $x_i = \Dist{\Sep}(v_0,v_i)$ as the \emph{canonical solution} to $S$.  Say that $S$ and a system $S'$ are \emph{equivalent} when they have the same canonical solution.  An edge $v_iv_j$ of $D$ is \emph{implied} by a path $W$ from $v_i$ to $v_j$ when $\Sep(W) \geq \Sep(v_iv_j)$; if the inequality is strict, then $v_iv_j$ is \emph{strongly implied} by $W$.  By definition, the digraph $D'$ obtained by removing all the strongly implied edges of $D$ defines a system equivalent to $S$.  Moreover, if every edge of $D$ is implied by a path of a spanning subgraph $D'$ of $D$, then $D'$ defines a system equivalent to $S$.

In the above description, there is at most one inequality $x_j \geq x_i + c_{ij}$ in $S$ for each ordered pair $x_ix_j$, while $D$ is a digraph.  Of course, there is no need for another inequality on $x_ix_j$ as one of these would be strongly implied.  Yet, for the sake of simplicity, it is sometimes convenient to describe a system with more than one constraint for each ordered pair $x_ix_j$.  In these situations, the corresponding constraint digraph $D$ is a multidigraph.  For the sake of notation we ignore this fact and we regard the edge $v_iv_j$ as representing both inequalities.  %Similarly, when $n \leq 3$, some of the systems that we describe have inequalities of the form $x_i \geq x_i + c_{ii}$; we ignore this fact and treat the corresponding loop of the constraint digraph as if it were a regular edge.

\section{The synthetic graph of a model}
\label{sec:synthetic graph}

The goal of \kMult is to decide if a PCA model $\M$ is equivalent to a $k$-multiplicative model.  In this section we define a compact system of difference constraints to solve the simpler problem \klcMult: given $\M$, $k \in \ORange{\wrap}$, and even values $\Circ$ and $\Len$, determine if $\M$ is equivalent to a $k$-multiplicative $(\Circ,\Len+1)$-CA model.

By definition, if $\M$ is equivalent to an even $k$-multiplicative $(\Circ,\Len+1)$-CA model $\Unit$, then $i \mult \Unit$ is equivalent to $\M^i$ for every $i \in \CRange{k}$.  So, if $A_0 < \ldots < A_{n-1}$ are the arcs of $\M$ and $U_0 < \ldots < U_{n-1}$ are the arcs of $\U$, then $\Boolean{s(i\mult U_y) \in i \mult U_x} = \Boolean{s(A_y^i) \in A_x^i}$ for every $x,y \in \Range{n}$.  Moreover, $U_x$ and $U_y$ satisfy the following inequalities because $s(i \mult U_y) = s(U_y)$ is even whereas $i\Len+1$ is odd (e.g.~\cref{fig:example-uig-pca,fig:example-mult-power}):
\begin{empheq}[left=\empheqlbrace]{align*}    
   s(U_y) &\leq s(U_x) + i\Len - \Circ\Boolean{x \geq y} & \text{for $i \in \CRange{k}$ if $s(A_y^i) \in A_x^i$}\\
   s(U_y) &\geq s(U_x) + i\Len+2 - \Circ\Boolean{x \geq y} & \text{for $i \in \CRange{k}$ if $s(A_y^i) \not\in A_x^i$} 
\end{empheq}

It is not hard to see that the converse of the previous reasoning is also true, regardless of whether $\Circ$ and $\Len$ are even or odd.  That is, if $\Unit$ is a $(\Circ,\Len+1)$-CA model with arcs $U_0 < \ldots < U_{n-1}$ that satisfy the above system, then $\Unit$ is $k$-multiplicative and equivalent to $\M$.  Moreover, as the position of $0$ in $C(\Unit)$ is irrelevant in the above inequalities, we can take $s(U_0) = 0$.  Therefore, a PCA model $\M$ is equivalent to a $k$-multiplicative UCA model if and only if the \emph{full system} $\F^k_{\Circ,\Len}(\M)$ below has a solution; see \cref{fig:example-model-constraints}. Moreover, any solution to $\F^k_{\Circ,\Len}(\M)$ yields a $k$-multiplicative $(\Circ,\Len+1)$-CA model equivalent to $\M$.  The system $\F^k_{\Circ,\Len}(\M)$ has an indeterminate $s(A)$ for each $A \in \A(\M)$; the overloaded notation is intentional. The equations of $\F^k_{\Circ,\Len}(\M)$ are defined as  
\begin{empheq}{align}    
   s(A) &= 0 && \text{if $A$ is the initial arc}\tag{initial}\label{eq:initial}\\
   s(A) &\geq s(B) - i\Len + \Circ\Boolean{A \geq B} && \text{for $i \in \CRange{k}$ if $s(B^i) \in A^i$}  \tag{$i$-attract}\label{eq:general hollow}\\
   s(B) &\geq s(A) + i\Len+2 - \Circ\Boolean{A \geq B} && \text{for $i \in \CRange{k}$ if $s(B^i) \not\in A^i$} \tag{$i$-repel} \label{eq:general nose}
\end{empheq}

Note that $\F^k_{\Circ,\Len}(\M)$ is a system of difference constraints where the non-negativity constraints follow by \ref{eq:initial} and the $0$-repel constraints (with $A$ as the initial arc).  Thus, it can conveniently be described with its constraint digraph $\F^k(\M)$ and its weighting $\Sep_{\Circ,\Len}(\M)$.  From now on, we drop $\M$ from $\F^k_{\Circ,\Len}$, $\F^k$, and $\Sep_{\Circ,\Len}$ when no confusions are possible.  For the sake of notation, we write $A$ to denote the vertex of $\F^k$ corresponding to $s(A)$ for every $A \in \A(\M)$.  Moreover, we interchangeably treat $A$ as an arc of $\M$ and as a vertex of $\F^k$. Note that $\F^k$ has $\Theta(kn^2)$ edges. The edge $B \to A$ corresponding to \eqref{eq:general hollow} is said to be an \emph{$i$-attract}, while the edge $A \to B$ corresponding to \eqref{eq:general nose} is called an \emph{$i$-repel}, for $i\in \CRange{k}$ and $A,B \in \A(\M)$.  Those $i$-attracts $B \to A$ with $A < B$ and those $i$-repels $A \to B$ with $A < B$ are called \emph{internal}; non-internal edges are said to be \emph{external}.  Intuitively, $A \to B$ is internal when $(s(A),s(B)) \cup \{s(B)\}$ does not contain $0$. 

\begin{figure}
  \mbox{}\hfill\parbox{.28\linewidth}{\includegraphics{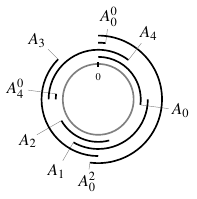}}\hfill%
  \parbox{.62\linewidth}{\begin{empheq}{align*}    
    s(A_1) &\geq s(A_0) + 2 && \text{$0$-repel $A_0 \to A_1$}\\
    s(A_0) &\geq s(A_1) - \Len && \text{$1$-attract $A_1 \to A_0$}\\
    s(A_2) &\geq s(A_0) + \Len + 2 && \text{$1$-repel $A_0 \to A_2$}\\
    s(A_0) &\geq s(A_2) - 2\Len && \text{$2$-attract $A_2 \to A_0$}\\    
    s(A_0) &\geq s(A_4) + 2 - \Circ && \text{$0$-repel external $A_4 \to A_0$}\\
    s(A_4) &\geq s(A_0) - \Len + \Circ && \text{$1$-attract external $A_0 \to A_4$}
  \end{empheq}}\hfill\mbox{}
  \caption{Some constraints of $\F_{\Circ,\Len}^2(\M)$ for the model $\M$ with arcs $A_0, \ldots, A_4$.}\label{fig:example-model-constraints}
\end{figure}

\paragraph{Some strongly implied edges.}

By Theorem~\ref{thm:difference constraints}, a solution to $\F^k_{\Circ,\Len}$, if existing, can be obtained in $O(kn^3)$ time by invoking Bellman-Ford's shortest path algorithm on $\F^k$ and $\Sep_{\Circ,\Len}$.  This algorithm can be improved by observing that most of the edges in $\F^k$ are strongly implied and can be safely removed.  Say that an $i$-attract $B \to A$ is an \emph{$i$-hollow} when $B = \Fr^{i}(A)$ and $A = \Fl^i(B)$.  Intuitively, $A^i$ is the leftmost arc of $\M^i$ reaching $s(B)$ and vice versa, thus $B \to A$ imposes the tightest $i$-attract on $A$ and $B$.

\begin{lemma}\label{lem:attracts are implied}
If\/ $\M$ is a PCA model, $k \in \ORange{\wrap}$, and $i \in \CRange{k}$, then every $i$-attract of\/ $\F^k$ that is not a hollow is strongly implied. 
\end{lemma}

\begin{proof}
Suppose that an $i$-attract $B \to A$ is not a hollow, thus $i > 0$ and either $B \neq \Fr^i(A)$ or $A \neq \Fl^i(B)$.  We prove only the former case, as the latter case is analogous. Hence, $s(B)$, $s(\RR(B))$, and $t(\Fr^{i}(A))$ appear in this order in $C(\M)$ (\cref{fig:constraints}(a)). By definition, $\RR(B) \to A$ is an $i$-attract and $B \to \RR(B)$ is a $0$-repel.  Moreover, either $A \leq B \leq R(B)$ or $R(B) \leq A \leq B$ or $B \leq R(B) \leq A$, thus $\Boolean{A \geq R(B)} - \Boolean{B \geq R(B)} = \Boolean{A \geq B}$.  Altogether, it follows that $B \to A$ is strongly implied by the path $\W = B, \RR(B), A$ of $\F^k$ because
\begin{align*}
 \Sep_{\Circ,\Len}(\W) &= 2 - \Circ\Boolean{B \geq R(B)} -i\Len + \Circ\Boolean{A \geq R(B)}  \\&= 2 -i\Len  + \Circ\Boolean{A \geq B} = 2+\Sep_{\Circ,\Len}(B \to A).\tag*{}
\end{align*}
\end{proof}

The fact that most $i$-repel edges are strongly implied can be proven with similar arguments.  Say that an $i$-repel $A \to B$ is an \emph{$i$-nose} when $A = \Hl^{i} \circ \LL(B)$. The next lemma provides a symmetric definition for $i$-noses: $A \to B$ is an $i$-nose when $B = \Hr^i\circ\RR(A)$.  Colloquially, $A \to B$ is an $i$-nose when $A^i$ is the rightmost arc not reaching $B^i$ and $B^i$ is the leftmost arc not reached by $A^i$.

\begin{lemma}\label{lem:nose equivalences}
Let\/ $\M$ be a PCA model, $k \in \ORange{\wrap}$, and $i \in \CRange{k}$.  The following statements are equivalent for $A, B \in \A(\M)$, and each of them implies that $A \to B$ is an $i$-repel of\/ $\F^k$.
\begin{discription}[label={\textbf{S\arabic*:}},ref={S\arabic*}]
 \item $A = \Hl^i \circ \LL(B)$.\label{lem:nose equivalences:def}
 \item $\Fr^w(A)=\Hl^{i-w} \circ \LL(B)$ for every $w \in \CRange{i}$.\label{lem:nose equivalences:fr}
 \item \ref{lem:nose equivalences:fr} and $\RR \circ \Fr^x(A) = \Hr^x \circ \RR(A)$ for every $x \in \CRange{i}$.\label{lem:nose equivalences:hr r}
 \item $\Fl^z(B)=\Hr^{i-z} \circ \RR(A)$ for every $z \in \CRange{i}$.\label{lem:nose equivalences:fl}
 \item \ref{lem:nose equivalences:fl} and $\LL \circ \Fl^y(B) = \Hl^y \circ \LL(B)$ for every $y \in \CRange{i}$.\label{lem:nose equivalences:hl l}
 \item $B = \Hr^i \circ \RR(A)$.\label{lem:nose equivalences:sym}
\end{discription}
\end{lemma}

\begin{proof}
 In this proof we use the following facts, and we write IH to reference the active inductive hypothesis, if any.  
 \begin{discription}[label={\textbf{F\arabic*:}},ref={F\arabic*}]
  \item If $X \in \A(\M)$ and $\Hr(X) \neq \bot$, then $\Fl \circ \Hr(X) = X$.\label{lem:nose equivalences:fl hr}
  \item If $X \in \A(\M)$ and $\Hl(X) \neq \bot$, then $\Fr \circ \Hl(X) = X$. \label{lem:nose equivalences:fr hl}
  \item If $X, Y \in \A(\M)$ and $X = \Hl(Y)$, then $\RR(Y) = \Hr \circ \RR(X)$. \label{lem:nose equivalence:one step r}
  \item If $X, Y \in \A(\M)$ and $X = \Hr(Y)$, then $\LL(Y) = \Hl \circ \LL(X)$. \label{lem:nose equivalence:one step l}
 \end{discription}
 
 \begin{factproof}
  \ref{lem:nose equivalences:fl hr} and \ref{lem:nose equivalences:fr hl} follow directly from the definition of $\Hr$ and $\Hl$, respectively.  For \ref{lem:nose equivalence:one step r}, observe that $s(\RR(Y))$ is the extreme immediately after $t(X)$ in a clockwise traversal of $C(\M)$ (\cref{fig:constraints}(b)).  Then, $\Fl\circ\RR(Y) = \RR(X)$, thus $\RR(Y) = \Hr \circ \RR(X)$ because $\Fl(Y) \neq \RR(X)$.  The proof of \ref{lem:nose equivalence:one step l} is omitted as it is analogous to that of \ref{lem:nose equivalence:one step r}.\hfill$\triangle$
 \end{factproof}

 $\ref{lem:nose equivalences:def} \Rightarrow \ref{lem:nose equivalences:fr}$ is proven by induction.  The base case $w = 0$ is trivial.  For $w+1 \leq i$, note that $\Hl^{i-w}\circ\LL(B) \neq \bot$ because $\Hl^i\circ\LL(B) = A \neq \bot$.  Then, $\Hl^{i-(w+1)}\circ\LL(B) \stackrel{\ref{lem:nose equivalences:fr hl}}{=} \Fr\circ\Hl^{i-w}\circ\LL(B) \stackrel{IH}{=} \Fr\circ\Fr^w(A) = \Fr^{w+1}(A)$.

  $\ref{lem:nose equivalences:fr} \Rightarrow \ref{lem:nose equivalences:hr r}$ is proven by induction.  The base case $x = 0$ is trivial.  For $x+1 \leq i$, let $X = \Hl^{i-x} \circ \LL(B)$ and $Y = \Hl^{i-(x+1)} \circ \LL(B)$.  By \ref{lem:nose equivalences:fr}, $X \neq \bot$ and $Y \neq \bot$, thus $X, Y \in \A(\M)$.  Then, $\Hr^{x+1}\circ\RR(A) = \Hr\circ\Hr^x\circ\RR(A) \stackrel{IH}{=} \Hr\circ\RR\circ\Fr^x(A) \stackrel{\ref{lem:nose equivalences:fr}}{=} \Hr\circ\RR\circ\Hl^{i-x}\circ\LL(B) = \Hr\circ\RR(X) \stackrel{\ref{lem:nose equivalence:one step r}}{=} \RR(Y) = 
  \RR\circ\Hl^{i-(x+1)}\circ \LL(B) \stackrel{\ref{lem:nose equivalences:fr}}{=} \RR \circ \Fr^{x+1}(A)$.
 
 $\ref{lem:nose equivalences:hr r} \Rightarrow \ref{lem:nose equivalences:fl}$ is proven by induction.  If $z = 0$, then $B = \RR\circ\Hl^0\circ\LL(B) \stackrel{\ref{lem:nose equivalences:fr}}{=} \RR\circ\Fr^i(A) \stackrel{\ref{lem:nose equivalences:hr r}}{=} \Hr^i\circ\RR(A)$, whereas if $z+1\leq i$, then $\Fl^{z+1}(B) = \Fl\circ\Fl^z(B) \stackrel{IH}{=} \Fl\circ\Hr^{i-z}\circ\RR(A) \stackrel{\ref{lem:nose equivalences:fl hr}}= \Hr^{i-(z+1)}\circ\RR(A)$.

 Implication $\ref{lem:nose equivalences:fl} \Rightarrow \ref{lem:nose equivalences:hl l}$ is analogous to $\ref{lem:nose equivalences:fr} \Rightarrow \ref{lem:nose equivalences:hr r}$; $\ref{lem:nose equivalences:hl l} \Rightarrow \ref{lem:nose equivalences:sym}$ is trivial; and the chain $\ref{lem:nose equivalences:sym} \Rightarrow \ref{lem:nose equivalences:fl} \Rightarrow \ref{lem:nose equivalences:hl l} \Rightarrow \ref{lem:nose equivalences:fr} \Rightarrow \ref{lem:nose equivalences:hr r} \Rightarrow \ref{lem:nose equivalences:def}$ is analogous to the chain $\ref{lem:nose equivalences:def} \Rightarrow \ref{lem:nose equivalences:fr} \Rightarrow \ref{lem:nose equivalences:hr r}\Rightarrow \ref{lem:nose equivalences:fl} \Rightarrow \ref{lem:nose equivalences:hl l} \Rightarrow \ref{lem:nose equivalences:sym}$.  Finally, note that  if any of the statements is true, then  $\Fr^i(A) \stackrel{\ref{lem:nose equivalences:def}}{=} \Fr^i\circ \Hl^i \circ \LL(B) \stackrel{\ref{lem:nose equivalences:fl hr}}{=} \LL(B)$, thus $A \to B$ is an $i$-repel of $\F^k$ by $\ref{lem:nose equivalences:fr}$.
\end{proof}

\begin{lemma}\label{lem:repels are implied}
If\/ $\M$ is a PCA model, $k \in \ORange{\wrap}$, and $i \in \CRange{k}$, then every $i$-repel of\/ $\F^k$ that is not a nose is strongly implied.
\end{lemma}

\begin{proof}
Suppose that an $i$-repel $A \to B$ is not a nose.  By Lemma~\ref{lem:nose equivalences}, $\Fr^j(A) \neq \Hl^{i-j}\circ\LL(B)$ for some $j \in \CRange{i}$.  Among all the possible choices, take the one maximizing $j$.  Note that either $j = i$ or $\Fr^{j+1}(A) = \Hl^{i-j-1}\circ\LL(B)$.  In the latter case $\Hl^{i-j} \circ \LL(B) \neq \bot$, while in the former case $\Hl^{i-j} \circ \LL(B) = \LL(B) \neq \bot$.  So, regardless of whether $i = j$, $\Hl^{i-j}\circ \LL(B) = X \neq \bot$ for some $X \in \A(\M)$ (\cref{fig:constraints}(c)).  Moreover, $s(\RR(X)) \not\in A^j$ because otherwise either $i = j$ and $A \to B$ is an $i$-attract or $i < j$ and $\Fr^{j+1}(A) \neq \Fr(X) = \Hl^{i-j-1}\circ\LL(B)$ (\cref{fig:constraints}(c)).  Then, $A \to X$ is a $j$-repel because $\Fr^j(A) \neq X$, i.e., $s(X) \not\in A^j$ (\cref{fig:constraints}(c)).  By Lemma~\ref{lem:nose equivalences}, $X \to B$ is an $(i-j)$-repel.  Altogether, $A \to B$ is implied by the path $\W = A, X, B$ of $\F^k$ because either $A \leq X \leq B$ or $X \leq B \leq A$ or $B \leq A \leq X$ and
 \begin{align*}
  \Sep_{\Circ,\Len}(\W) &= j\Len + 2 - \Circ\Boolean{A \geq X} + (i-j)\Len + 2 - \Circ\Boolean{X \geq B} 
  \\&= i\Len + 4 - \Circ\Boolean{A \geq B} = 4+\Sep_{\Circ,\Len}(A \to B).\tag*{}
 \end{align*}
\end{proof}

\begin{figure}
 \begin{tabular}{ccc}
   \includegraphics{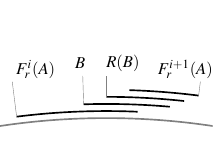} & \includegraphics{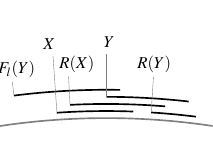} &  \includegraphics{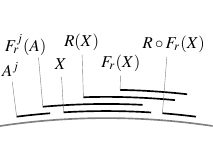}\\
  (a) & (b) & (c)
 \end{tabular}
 \caption[]{Removal of strongly implied constraints: (a) Lemma~\ref{lem:attracts are implied}; (b) Lemma~\ref{lem:nose equivalences} (\ref{lem:nose equivalence:one step r}); (c) Lemma~\ref{lem:repels are implied} for $j < i$.  In (c), $\Fr^{j+1}(A) = \Fr(X) = \Hl^{i-j-1} \circ \LL(B)$.}\label{fig:constraints}
\end{figure}

\paragraph{Some (weakly) implied edges.}

By Lemmas~\ref{lem:attracts are implied} and~\ref{lem:repels are implied}, the spanning subgraph $\F'$ of $\F^k$ formed by the hollows and noses, together with $\Sep_{\Circ,\Len}$, describes a system equivalent to $\F^k_{\Circ,\Len}$.  The digraph $\F'$ has $O(kn)$ edges and it can be further simplified.  For $k \in \ORange{\omega}$, let $\Syn_*^k(\M)$ be the spanning subgraph of $\F^k$ whose edges are the $1$-hollows and $i$-noses of $\M$, for $i \in \CRange{k}$.
 
\begin{lemma}\label{lem:only 1 hollows}
If\/ $\M$ is a PCA model and $k \in \ORange{\wrap}$, then\/ $\Syn^k_*(\M)$, together with\/ $\Sep_{\Circ,\Len}$, describes a system equivalent to\/ $\F^k_{\Circ,\Len}$.
\end{lemma}
 
\begin{proof}
For $i > 1$, consider an $i$-hollow $A_i \to A_0$ of $\F^k$ and let $A_{j} = F_r^j(A_0)$ for $j \in \Range{i}$.  By definition $A_{j+1} \to A_j$ is a $1$-attract for $j \in \Range{i}$. Since $k < \wrap$, it follows that $A_{j} \leq \ldots \leq A_i \leq A_0 \leq \ldots \leq A_{j-1}$ for the unique $j \in \CRange{i}$ such that $0$ belongs to $(s(A_{j-1}), s(A_j))$ (with indices modulo $i$).  Thus, $A_0 \to A_i$ is implied by the path $\W = A_i, \ldots, A_0$ of $\F^k$ because
\begin{align*}
 \Sep_{\Circ,\Len}(\W) &= \sum_{j=1}^i (\Circ\Boolean{A_{j-1} \geq A_j} - \Len) = \Circ\Boolean{A_0 \geq A_i}-i\Len = \Sep_{\Circ,\Len}(A_i \to A_0).
\end{align*}
If $A_{j+1} \to A_j$ is not a $1$-hollow for some $j \in \CRange{i}$, then $A_{j+1} \to A_j$ is strongly implied (Lemma~\ref{lem:attracts are implied}) and, consequently, $A_i \to A_0$ is strongly implied.  Otherwise, $\W$ is a path of $\Syn_*^k$ and the result follows by Lemma~\ref{lem:attracts are implied}.
\end{proof}

An argument similar to that of Lemma~\ref{lem:only 1 hollows} can be used to remove most of the noses.  Roughly speaking, the $i$-noses that must be kept have the largest possible $i$ (note that $i$ need not be equal to $k$ because one cannot assure that a $k$-nose from $A$ exists for every $A \in \A(\M)$).  Say that an $i$-nose $A \to B$ is \emph{short} when $i < k$ and either $\Hl(A) \neq \bot$ or $\Hr(B) \neq \bot$. Those $i$-noses that are not short are said to be \emph{long}.  For $k \in \ORange{\wrap}$, define $\Syn^k(\M)$ as the spanning subgraph of $\Syn_*^k(\M)$ obtained by removing all the short noses, and $\Syn^0(\M)$ as the spanning subgraph of $\Syn_*^\wrap(\M)$ having the $0$-noses and $1$-hollows.  Note that $\Syn^k_* = \bigcup\{\Syn^i \mid i\in\CRange{k}\}$; as usual, we omit the parameter $\M$.  The digraph $\Syn^k$ is called the \emph{$k$-order synthetic graph} of $\M$.

\begin{lemma}\label{lem:only long noses}
If\/ $\M$ is a PCA model and $k \in \ORange{\wrap}$, then\/ $\Syn^k(\M)$, together with\/ $\Sep_{\Circ,\Len}$, describes a system equivalent to\/ $\F^k_{\Circ,\Len}$.
\end{lemma}

\begin{proof}
By Lemma~\ref{lem:only 1 hollows}, it suffices to prove that every $i$-nose $A \to B$ is implied by a path of $\Syn^k$.  The proof is by induction on $k-i$.  The base case $i = k$ is trivial because $A \to B$ is long.  For the inductive step, in which $i < k$ and $A \to B$ is short, we have that either $\Hl(A) \neq \bot$ or $\Hr(B) \neq \bot$. We prove only the former case, as the latter case is analogous.  Thus, $A \to \Hl(A)$ is a $1$-attract, whereas $\Hl(A) \to B$ is an $(i+1)$-repel by Lemma~\ref{lem:nose equivalences}. Then, since $i \in \Range{\wrap}$, we obtain that either $\Hl(A) \leq A \leq B$ or $A \leq B \leq \Hl(A)$ or $B \leq \Hl(A) \leq A$.  Consequently, $\W = A, \Hl(A), B$ is a path of $\Syn_*^k$ that implies $A \to B$ because
\begin{align*}
 \Sep_{\Circ,\Len}(\W) &= -\Len + \Circ\Boolean{\Hl(A) \geq A} + (i+1)\Len + 2 - \Circ\Boolean{\Hl(A) \geq B} \\
 &= i\Len + 2 - \Circ\Boolean{A \geq B} = \Sep_{\Circ,\Len}(A \to B). 
\end{align*}
If $A \to \Hl(A)$ is not a hollow or $\Hl(A) \to B$ is not a nose, then $A \to B$ is strongly implied by Lemmas~\ref{lem:attracts are implied}~and~\ref{lem:repels are implied}.  Otherwise, by induction, either $A, \Hl(A), B$ is a path of $\Syn^k$ or $\Hl(A) \to B$ is a short nose and $A \to B$ is implied by a path of $\Syn^k$.
\end{proof}

Finally, the following corollary of Theorem~\ref{thm:difference constraints} sums up this section.

\begin{theorem}\label{thm:no_positive_cycles}
  Let $A_0$ be the initial arc of a PCA model\/ $\M$, $k \in \ORange{\wrap}$, and $\Circ, \Len \in\mathbb{N}$.  If\/ $\M$ is equivalent to an even $k$-multiplicative $(\Circ,\Len+1)$-CA model, then $\Sep_{\Circ,\Len}(\W) \leq 0$ for every cycle\/ $\W$ of\/ $\F^k$.  Conversely, if $\Sep_{\Circ,\Len}(\W) \leq 0$ for every cycle\/ $\W$ of $\Syn^k$, then\/ $\M$ is equivalent to the $k$-multiplicative $(\Circ,\Len+1)$-CA model\/ $\Unit$ that has an arc with beginning point $\Dist{\Sep_{\Circ,\Len}}(A_0, A)$ for every $A \in \A(\M)$.
\end{theorem}

\paragraph{The weighting of a walk.} By Theorem~\ref{thm:no_positive_cycles}, the weighting of each cycle of $\Syn_*^k$ plays a fundamental role in deciding if a PCA model $\M$ is equivalent to a $k$-multiplicative $(\Circ,\Len+1)$-model; we find it useful to define $\Sep_{\Circ,\Len}$ as a linear function on $\Circ$ and $\Len$.  For a walk $\W$ of $\Syn_*^k$, let:
\begin{itemize}
 \item $\hollow(\W)$ and $\nose(\W)$ be the number of hollows and noses of $\W$, respectively,
 \item $\hollow_{\Ext}(\W)$ and $\nose_{\Ext}(\W)$ be the number of external hollows and noses of $\W$, respectively,
 \item $\nose(i,\W)$ be the number of $i$-noses of $\W$,
 \item $\Bal(\W) = \sum_{i=0}^{k} i\nose(i, \W) - \hollow(\W)$, and $\Ext(\W) = \hollow_{\Ext}(\W)-\nose_{\Ext}(\W)$.
\end{itemize}
By definition, $\Sep_{\Circ,\Len}(B \to A) = -\Len + \Circ\Boolean{A \geq B}$ for every $1$-hollow $B \to A$, and also $\Sep_{\Circ,\Len}(A \to B) = i\Len + 2 - \Circ\Boolean{A \geq B}$ for every $i$-nose $A \to B$.  With the above terminology,
\begin{equation}
 \Sep_{\Circ,\Len}(\W) = \Len\Bal(\W) + \Circ\Ext(\W) + 2\nose(\W). \label{eq:sep}
\end{equation}
Intuitively, a walk $\W = A_1, \ldots, A_k$ can be seen as a traversal of the beginning points $s(A_1)$, \ldots, $s(A_k)$ of $\M$ in this order.  The weight $\Sep_{\Circ,\Len}(A_i \to A_{i+1})$ is a lower bound on how far $s(A_i)$ and $s(A_{i+1})$ must be in $C(\M)$, thus $\Sep_{\Circ,\Len}(\W)$ is a lower bound for the separation of $s(A_1)$ and $s(A_k)$.   In this sense, $\Bal(\W)$ (for ``balance'') accumulates the separation according to $\Len$, while $\Ext(\W)$ denotes the number of times that $0$ is crossed in $C(\M)$, taking into account if the cross is in a clockwise (nose) or anticlockwise (hollow) sense. 

\subsection{Building the synthetic graph of a model}
\label{sec:building}

Lemma~\ref{lem:nose equivalences} implies that noses have a well defined structure, that is depicted by the gray arrows of Figure~\ref{fig:diagonal}. To obtain this picture, suppose $A_0 \to B_i$ is an $i$-nose of $\Syn^k_*$.  By definition (i.e., \ref{lem:nose equivalences:def}), $B_i = \Hr^i \circ \RR(A_0)$, thus $B_j = \Hr^j \circ \RR(A_0) \neq \bot$ for every $j \in \CRange{i}$.  Similarly, $A_0 = \Hl^i \circ \LL(B_i)$ by~\ref{lem:nose equivalences:sym}, thus $A_j = \Hl^{i-j} \circ \LL(B_i) \neq \bot$.  That $A_j = \Hl(A_{j+1})$ and $B_{j+1} = \Hr(B_j)$ when $j < i$ follow by definition, while $A_{j+1} = \Fr(A_j)$ and $B_j = \Fl(B_{j+1})$ follow by \ref{lem:nose equivalences:fr}~and~\ref{lem:nose equivalences:fl}, respectively.  Also, \ref{lem:nose equivalences:hr r} (or, symmetrically, \ref{lem:nose equivalences:hl l}) implies that $B_j = \RR(A_j)$.  Finally, \ref{lem:nose equivalences:fr}~and~\ref{lem:nose equivalences:fl} imply that $\Fr^h(A_0) \to \RR \circ \Fr^j(A_0)$ is a $(j-h)$-nose for every $h \in \Range{j}$.  Summing up, the existence of an $i$-nose $A \to \RR \circ \Fr^i(A)$ implies the existence of many other noses in $\Syn^k_*$.  Conversely, the $0$-nose $A \to \RR(A)$ generates the $i$-nose $A \to \RR \circ \Fr^i(A)$ when $\Hr^i \circ \RR(A) \neq \bot$.

\begin{figure}
	\centering
	\includegraphics{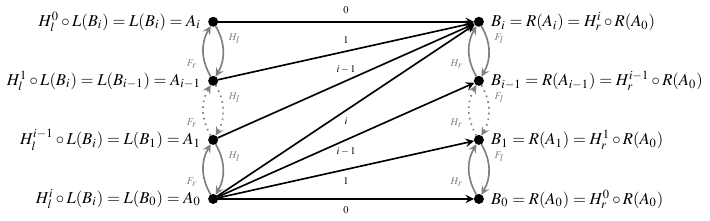}
	\caption[]{Noses of $\Syn_*^k$ from $A_0$ and to $\RR(A_{i})$ for an $i$-nose $A_0 \to \RR(A_i)$ as implied by Lemma~\ref{lem:nose equivalences}.  A label $x$ is attached to each $x$-nose.} \label{fig:diagonal}
\end{figure}

By definition, any long $i$-nose $A \to B$ of $\Syn^k$ with $i < k$ is an $i$-nose of $\Syn^{k+1}$ as well.  Therefore, to transform $\Syn^k$ into $\Syn^{k+1}$ it suffices to insert the $(k+1)$-noses of $\Syn_*^{k+1}$ and to remove the $k$-noses of $\Syn^k$ that are short in $\Syn_*^{k+1}$.  In contrast to Theorem~\ref{thm:no_positive_cycles}, the next result holds for $k = 0$.  This highlights the importance of $\Syn^0$: it is the base case for building $\Syn^k$.

\begin{theorem}\label{thm:Syn^k to Syn^k+1 remove}
 Let\/ $\M$ be a PCA model and $k \in \Range{\wrap-1}$.  Then, $\Syn^{k+1}$ can be computed from\/ $\Syn^k$ in two phases: first, each $k$-nose $A \to B$ such that $\Hr(B) \neq \bot$ is iteratively replaced by the $(k+1)$-nose $A \to \Hr(B)$; then, each remaining $k$-nose $A \to B$ with $\Hl(A) \neq \bot$ is removed.
\end{theorem}

\begin{proof}
 By Lemma~\ref{lem:nose equivalences}, the graph obtained after the first phase has all the $(k+1)$-noses and every $k$-nose $A \to B$ satisfies $\Hr(B) = \bot$.  The second step removes the remaining short $k$-noses.
\end{proof}

Theorem~\ref{thm:Syn^k to Syn^k+1 remove} yields a simple method to compute $\Syn^k$ in $O(nk)$ time when $\M$ is given (\cref{fig:example-syngraph}).  Starting with $\Syn^0$, execute $k$ steps to transform $\Syn^i$ into $\Syn^{i+1}$ for each $0 \leq i < k$.  However, Theorem~\ref{thm:Syn^k to Syn^k+1 remove} can be reinterpreted to design a faster algorithm.  Just note that $A \to B$ is an $i$-nose of $\Syn^k$ if and only if $B = \Hr^i(\RR(A))$ and either $i = k$ or both $\Hr^{i+1}(\RR(A)) = \bot$ and $\Hl(A) = \bot$.  For the next theorem, extend the weighting $\Bal$ of $\Syn_*^k$ to work with edges: $\Bal(A \to B) = -1$ if $A \to B$ is a $1$-hollow and $\Bal(A \to B) = i$ if $A \to B$ is an $i$-nose, for $i \in \CRange{k}$.

\begin{theorem}\label{thm:Syn^k algorithm}
 The problem of computing\/ $\Syn^k$ and $\Bal$, when a PCA model\/ $\M$ and a value $k \in \Range{\wrap}$ are given, can be solved in $O(n)$ time.  With this information, $\Sep_{\Circ,\Len}(A \to B)$ can be obtained in $O(1)$ time when $\Circ,\Len \in \mathbb{N}$ are given, for $k \in \Range{\wrap}$.
\end{theorem}

\begin{proof}
 We assume that $\LL(A)$, $\RR(A)$, $\Fl(A)$, $\Fr(A)$, $\Hl(A)$, and $\Hr(A)$ can be obtained in $O(1)$ time for a given $A \in \A(\M)$.  There is no loss of generality, as they can be computed in $O(n)$ time from any reasonable representation of $\M$ (e.g.~\cite[Algorithm~5.2]{Soulignac2010}).  To find the $1$-hollows of $\Syn^k$ with $\Bal = -1$, it suffices to traverse each arc $A$ of $\M$ to determine if $\Fr \circ \Fl(A) = A$.  This step requires $O(n)$ time; we now discuss how to find all the $i$-noses.
 
 Let $\nose$ be the function such that $\nose(A) = \max\{i \mid i \leq k \text{ and } \Hr^i(A) \neq \bot\}$ and $H$ be the function such that $H(A) = \Hr^j(A)$, $j=\nose(A)$, for every $A \in \A(\M)$.  By Lemma~\ref{lem:nose equivalences}, there is a long nose $\LL(A) \to H(A)$ if and only if either $\nose(A) = k$ or $\Hl \circ \LL(A) = \bot$.  Moreover, if existing, then $\LL(A) \to H(A)$ is the unique $\nose(A)$-nose of $\Syn^k$ starting at $\LL(A)$.  Thus, once $\nose$ and $H$ are known, the problem of finding each nose $A \to B$ of $\Syn^k$, together with $\Bal(A\to B)$, can be accomplished in $O(n)$ time.
 
 Let $D$ be the digraph that has one vertex $v(A)$ for each $A \in \A(\M)$ and one edge $v(A) \to v(B)$ when $\Hr(A) = B$.  Clearly, every vertex of $D$ has at most one out neighbor.  Moreover, since there is at most one arc $A$ such that $B = \Hr(A)$ for every $B \in \A(\M)$, it follows that all the vertices in $D$ have at most one in neighbor as well.  Therefore, every component $D'$ of $D$ is either a path or a cycle.  If $D'$ is a path $v(A_1), \ldots, v(A_j)$, then $\nose(A_i) = \min\{k, j-i\}$ and $H(A_i) = A_{i + \nose(A_i)}$.  Similarly, if $D'$ is a cycle $v(A_1), \ldots, v(A_j)$, then $\nose(A_i) = k$ and $H(A_i) = A_{i+k \bmod j}$.  Then, by keeping two pointers, $\nose$ and $H$ can be computed for all the arcs corresponding to vertices in $D'$ in $O(j)$ time.  Hence, the total time required to compute $H$ and $\nose$ is $O(n)$ and, therefore, $O(n)$ time suffices to compute $\Syn^k$ and $\Bal$.
\end{proof}

\begin{corollary}\label{cor:klcMult}
 The problem \klcMult can be solved in $O(n^2)$ time for any PCA model\/ $\M$, $k \in \ORange{\wrap}$, and $\Circ,\Len \in \mathbb{N}$.  The algorithm outputs either a $k$-multiplicative $(\Circ,\Len+1)$-CA model equivalent to\/ $\M$ or a minimal family of difference constraints of\/ $\Syn^k$ with no feasible solution.
\end{corollary}

\begin{proof}
 By Theorem~\ref{thm:no_positive_cycles}, \klcMult is solved with an execution of Bellman-Ford's algorithm on $\Syn^k$ weighted by $\Sep_{\Circ,\Len}$ from the initial arc of $\M$.  The algorithm requires $O(n^2)$ time because $\Syn^k$ has $O(n)$ edges and can be computed in $O(n)$ time.  If a cycle of positive weight $\W$ is found, the corresponding family $\F$ of constraints is given as output.  By Theorem~\ref{thm:difference constraints}, $\F$ has no solution because its constraint digraph is isomorphic to $\W$ (with weight $\Sep_{\Circ, \Len}$), while $\F$ is minimal because the constraint digraph of $\F - \F'$ has no cycles for every nonempty subsystem $\F'$.
\end{proof}

If $\Syn^k$ has no feasible solution, then the family of difference constraints $\F$ given by Corollary~\ref{cor:klcMult} defines a submodel $\M'$ of $\M$ that contains an arc $A$ of $\M$ if and only if $A$ is referred by a constraint of $\F$. Note that $\M'$ is equivalent to no $k$-multiplicative $(\Circ,\Len+1)$-CA model because $\Syn^k(\M')$ contains $\F$.

\section{Mitas' drawing of a synthetic graph}
\label{sec:mitas drawing}

Mitas observed that $\Syn^1$ admits a peculiar drawing in the plane when $\M$ is PIG~\cite{Mitas1994}. %refchange-R4
These drawings were later adapted to PCA models by Soulignac~\cite{SoulignacJGAA2017}, %refchange-R4
and provide a powerful tool for solving numerical representation problems \cite{Mitas1994,PirlotVincke1997,SoulignacJGAA2017,SoulignacJGAA2017a}.  In this section we define an analogous of Mitas' drawings for $\Syn^k_*$.  Although drawings are defined for general PCA models, the results are restricted to connected models for simplicity.

In Mitas' drawings, each arc $A$ of a PCA model $\M$ occupies an entry of an imaginary matrix.  The \emph{row} of $A \in \A(\M)$ is (\cref{fig:example-syngraph}):
\begin{align}
\Row(A) & = \begin{cases}
           0 & \text{if $A$ is the initial arc}\\
           \Row(B)+1 & \text{if $A = \Hr(B)$ for some $B < A$} \\
           \Row(\LL(A)) & \text{otherwise}
           \end{cases}\tag{row} \label{eq:row}
\end{align}
The \emph{number of rows} of $\M$ is defined as $\Rows(\M) = 1+\max\{\Row(A) \mid A \in \A(\M)\} = 1+\Row(\LL(A_0))$ for the initial arc $A_0$.  The family $\ARow$ of arcs with $\Row = r$ is referred to as the \emph{row $r$} of $\M$.  It is not hard to see that $\ARow$ forms a contiguous sequence.  We refer to those arcs that are the leftmost and rightmost of its row as being \emph{leftmost} and \emph{rightmost}, respectively.  A nose $A \to B$ is said to be \emph{backward} when it is internal and $A$ is rightmost. All the hollows and non-backward noses that are internal are called \emph{forward}.  A walk of $\Syn^k_*$ is \emph{internal} if all its edges are internal and it is \emph{forward} if all its edges are forward.  The \emph{interior} and \emph{backbone} of $\Syn^k_*$ are the spanning subgraphs of $\Syn^k_*$ formed by the interior and forward edges, respectively.

\begin{figure}
	\centering
    \mbox{}\hfill\includegraphics{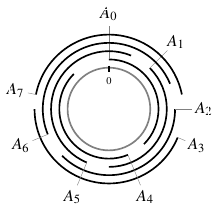}\hfill\includegraphics{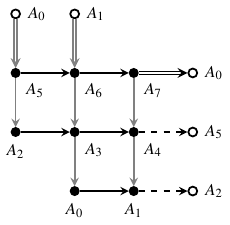}\hfill\includegraphics{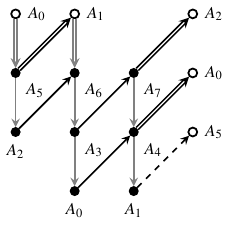}\hfill\mbox{}
	\caption{%The $\LinEnd_c$ and $\Boolean{b}$ operators encapsulate the modulo arithmetic.
	  A PCA model $\M$, $\Syn^0$, and $\Syn^1$ are shown from left to right.  The vertices of $\Syn^0$ and $\Syn^1$ are displayed black in rows according to \eqref{eq:row}, while white vertices correspond to the black vertex with the same label.  Noses are black, hollows are gray, double arrows are for external edges, and dashed arrows are for backward noses.  Note that $\Syn^1(\M)$ can be obtained from $\Syn^0(\M)$ as in \cref{thm:Syn^k to Syn^k+1 remove}.}\label{fig:example-syngraph}
\end{figure}

Besides helping with the drawing of $\Syn^k$, the assignment of rows to arcs allows us to classify the hollows and noses according to how many rows are skipped when a hollow or nose is traversed.
Specifically, the \emph{jump} of an internal\footnote{Although this definition can be properly applied to any edge of $\Syn_*^k$, as it is on \cite{SoulignacJGAA2017}, in this work we will be concerned just with internal edges.} edge $A \to B$ of $\Syn_*^k$ is the number $\Jmp(A \to B) = \Row(B) - \Row(A)$ of rows crossed by $A \to B$.  

\begin{lemma}\label{lem:jmp of noses and hollows}
 If\/ $\M$ is a connected PCA model, $k \in \Range{\wrap}$, and $i \in \CRange{k}$, then internal hollows have $\Jmp = -1$, forward $i$-noses have $\Jmp = i$, and backward $i$-noses have $\Jmp = i+1$.
\end{lemma}

\begin{proof}
 If $A \to B$ is an internal $1$-hollow, then $B = \Fl(A)$ and $A = \Fr(B)$, thus $\Hr(B) \neq \bot$ is such that $\Hr(B) \leq A$.  Moreover, since $\Fl(X) = \Fl \circ \Hr(B) = \Fl(A) = B$ for every arc $X$ with $\Hr(B) < X \leq A$, it follows that no arc $Z$ with $\Hr(Z) = X$ exists.  Consequently, $\Row(A) = \Row(\Hr(B)) = \Row(B) + 1$ by \eqref{eq:row} because $A \to B$ is internal.
 
 If $A \to B$ is an internal $i$-nose, then $B = \Hr^i \circ \RR(A)$ by Lemma~\ref{lem:nose equivalences}.  Observe that $\Hr^{j+1} \circ \RR(A) >  \Hr^j \circ \RR(A)$ for every $j \in \Range{i}$ because $\M$ is connected, $\RR(A)$ is not the initial arc, and $A \to B$ is internal.  Hence,  by~\eqref{eq:row}, $\Row(B) = \Row(\Hr^i  \circ \RR(A)) = \Row(\Hr^{i-1}  \circ \RR(A)) + 1 = \ldots = \Row(\RR(A)) + i$.  Recall that $\Jmp(A \to \RR(A)) \in \{0,1\}$ and it equals $0$ if and only if $A \to \RR(A)$ is forward.  Therefore, $\Jmp(A \to B) \in \{i,i+1\}$ and it equals $i$ if and only if $A \to B$ is forward.
\end{proof}

Lemma~\ref{lem:jmp of noses and hollows} can be extended to general walks.  For this purpose, define the jump of an internal walk $\W = A_1,\ldots,A_k$ as $\Jmp(\W)=\sum_{i=1}^{k-1}\Jmp(A_i \to A_{i+1})$.

\begin{corollary}\label{cor:jmp and bal}
 Let\/ $\M$ be a connected PCA model and $k \in \Range{\wrap}$.  If\/ $\W$ is an internal walk of\/ $\Syn_*^k$, then\/ $\Jmp(\W) = \Bal(\W) + \nose_b$ where $\nose_b$ is the number of backward noses of\/ $\W$.
\end{corollary}

\begin{proof}
 If $f_i$ and $b_i$ are the number of forward and backward $i$-noses of $\W$, respectively, then 
 \begin{align*}
  \Jmp(\W) &= \sum_{i=1}^k \left(if_i + (i+1)b_i\right) - \hollow(\W) 
            = \sum_{i=1}^k i\nose(i,\W) - \hollow(\W) + \sum_{i=1}^k b_i = \Bal(\W) + \nose_b
 \end{align*}
 by Lemma~\ref{lem:jmp of noses and hollows}.
\end{proof}

The \emph{column} $\Col(A)$ of the arc $A$ is defined according to a ``topological ordering'' of the backbone of $\Syn^k_*$.  The fact that such an ordering exists follows from the next lemma.

\begin{lemma}\label{lem:acyclic}
 If\/ $\M$ is a connected PCA model and $k \in \Range{\wrap}$, then the backbone of\/ $\Syn^k_*$ is acyclic.
\end{lemma}

\begin{proof}
 It suffices to show that $A_0 < A_j$ for any forward walk $\W = A_0, \ldots, A_j$ of $\Syn^k_*$ with $j > 0$ such that $\Row(A_0)=\Row(A_j)$.  The proof is by induction on $\hollow(\W)$ and $|\W|$.  By Lemma~\ref{lem:jmp of noses and hollows}, every edge of $\W$ is a $0$-nose when $\hollow(\W) = 0$, thus $A_i < A_{i+1} = \RR(A_i)$ for every $i \in \Range{j}$ and, hence, $A_0 < A_j$.  For the inductive step, consider the following alternatives.
 \begin{discription}[label={\textbf{Case~$\mathbf{\arabic*}$:}}]
  \item $\Row(A_i) = \Row(A_0)$ for some $i \in \ORange{j}$. Let $\W_0$ and $\W_1$ be the subpaths of $\W$ from $A_0$ to $A_i$ and from $A_i$ to $A_j$, respectively.  Clearly, $\max\{\hollow(\W_0), \hollow(\W_1)\} \leq \hollow(\W)$ and $\max\{|\W_0|, |\W_1|\} < j$, thus $A_0 < A_i$ and $A_i < A_j$ by induction.

  \item $\Row(A_i) > \Row(A_0)$ for some $i \in \ORange{j}$.  By Lemma~\ref{lem:jmp of noses and hollows}, there exists $x,y \in \ORange{j}$, $x \leq y$, such that $A_{x-1} \to A_x$ is a $p$-nose, $p > 0$, and $A_{y} \to A_{y+1}$ is a $1$-hollow.  Among all the possible combinations for $x$ and $y$, take one minimizing $y-x$.  In this configuration, $\Row(A_x) = \Row(A_y)$ as every edge in the subpath $A_x, \ldots, A_y$ of $\W$ is a $0$-nose.  By Lemma~\ref{lem:nose equivalences} (\cref{fig:diagonal}), $\Fl(A_x) = \Hr^{p-1}\circ\RR(A_{x-1})$ and $A_{x-1} \to \Fl(A_x)$ is a $(p-1)$-nose.  The former condition implies that $s(A_{x-1}), s(\Fl(A_x)), s(A_{x})$ appear in this order in a clockwise traversal of $C(\M)$.  Then, taking into account that $A_{x-1} < A_x$ because $\W$ is internal, we conclude that $A_{x-1} < \Fl(A_x) < A_x$ and, thus, $A_{x-1} \to \Fl(A_x)$ is internal as well.  Moreover, $A_{x-1} \to \Fl(A_x)$ is forward because it starts at the same vertex as the forward edge $A_{x-1} \to A_x$.  Then, by Lemma~\ref{lem:jmp of noses and hollows} and recalling that $A_y \to A_{y+1}$ is a $1$-hollow, we obtain that $\Row(\Fl(A_x)) = \Row(A_x) - 1 = \Row(A_y) - 1 = \Row(A_{y+1})$.  And, since $\Fl(A_{x}) \leq \Fl(A_y) = A_{y+1}$ because $A_x \leq A_y$, the walk $\W_0$ of $0$-noses of $\Syn_*^k$ from $\Fl(A_x)$ to $A_{y+1}$ is forward as well.  Summing up, $\W'$ $=$ $A_0, \ldots, A_{x-1}, \Fl(A_x)$ $+$ $\W_0$ + $A_{y+1}, \ldots, A_j$ is a forward walk of $\Syn_*^k$.  Then $A_0 < A_j$ follows by induction because $\hollow(\W') = \hollow(\W) - 1$.
  
  \item $\Row(A_i) < \Row(A_0)$ for every $i \in \ORange{j}$.  By Lemma~\ref{lem:jmp of noses and hollows}, $A_0 \to A_1$ is a $1$-hollow and $A_{j-1} \to A_j$ is a $p$-nose, $p > 0$.  By Lemma~\ref{lem:nose equivalences} (\cref{fig:diagonal}), $A_j = \Hr \circ \Fl(A_j)$ and $A_{j-1} \to \Fl(A_j)$ is a $(p-1)$-nose that is forward because it starts at the same vertex as the forward edge $A_{j-1} \to A_j$.  Then, $\W' = A_1, \ldots, A_{j-1}, \Fl(A_j)$ is a forward walk of $\Syn^k_*$.  Clearly, $|\W'| > 1$, $\hollow(\W') = \hollow(\W)-1$ and, by Lemma~\ref{lem:jmp of noses and hollows}, $\Row(A_1) = \Row(A_0) - 1 = \Row(A_j) - 1 = \Row(\Fl(A_j))$.  Then, $A_1 \leq \LL \circ \Fl(A_j) < \Fl(A_j)$ follows by induction and, consequently, $A_0 = \Fr(A_1) \leq \Fr\circ\LL \circ \Fl(A_j) < \Hr \circ \Fl(A_j) = A_j$.
 \end{discription}
\end{proof}

The \emph{column} of an arc $A$ of $\M$ is defined as:%
\begin{displaymath}
 \Col_0(A) = \max(\{0\} \cup \{1+\Col_0(B) \mid B \to A \text{ is a forward edge of $\Syn_*^{\wrap-1}$}\}).
\end{displaymath}%
Note that $\Col_0(A) \geq 0$ is well defined since the backbone of $\Syn_*^{\wrap-1}$ is acyclic (Lemma~\ref{lem:acyclic}).  The \emph{number of columns} of $\M$ is $\Cols(\M) = 1+\max\{\Col_0(A) \mid A \in \A(\M)\}$. 

\begin{figure}
	\centering
    \includegraphics{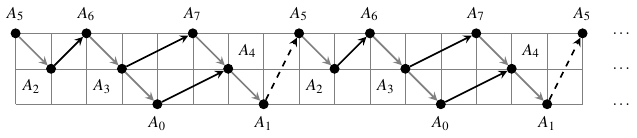}
    
	\caption{Mitas' drawing of $\Syn^1$ for the model $\M$ in \cref{fig:example-syngraph}, where only $0$- and $1$-arrows are drawn.  Each point in the drawing represents one of infinitely many positions of the vertex, while each dashed arrow corresponds to a backward nose.  Note that $\Col_0(A_7) = 5$ because $\Syn_*^2$ contains the $2$-nose $A_0 \to A_7$.}\label{fig:syngraph-columns}
\end{figure}

Let $\Col_i(A) = i\Cols(\M) + \Col_0(A)$ and $\Pos_i(A) = (\Col_i(A), \Row(A))$ for every $i > 0$ and $A \in \A(\M)$.  The \emph{(Mitas') drawing} of each subdigraph $D$ of $\Syn_*^{\wrap-1}$ is obtained by placing, in $\mathbb{R}^2$ and for every $i \geq 0$, a straight \emph{$i$-arrow} from $\Pos_i(A)$ to $\Pos_i(B)$ for each forward edge $A \to B$ of $D$ and a straight \emph{$i$-arrow} from $\Pos_i(A)$ to $\Pos_{i+1}(B)$ for each backward nose $A \to B$ of $D$ (\cref{fig:syngraph-columns}).  We write $v \to w$ to denote the arrow from $v$ to $w$ for $v,w \in \mathbb{R}^2$ and, for simplicity, we say that $v \to w$ is an \emph{arrow} of $D$ to mean that $v \to w$ is an arrow corresponding to an edge of $D$.  For every $p \geq0$, every internal walk $\W = A_0, \ldots, A_j$ of $D$ defines a curve $\Gr_p(\W)$ in $\mathbb{R}^2$ that starts at $q_0 = \Pos_p(A_0)$ and, for $x \in \Range{j}$, it takes the $i$-arrow of $A_x \to A_{x+1}$ to move from $q_x$ to $q_{x+1}$ for the unique $i$ such that $q_x = \Pos_i(A_x)$.

\begin{corollary}\label{cor:walk drawing}
  Let\/ $\M$ be a connected PCA model.  If\/ $\W$ is an internal walk of\/ $\Syn^{\wrap-1}_*$ from $A$ to $B$ with $b$ backward noses, then\/ $\Gr_p(\W)$, $p \geq 0$, is the graph of a continuous function with domain $[\Col_p(A_0), \Col_{p+b}(A_j)]$.
\end{corollary}

The drawing of $\Syn^k$ is so attractive because it is ``plane'', thus it provides a geometric framework to reason about PCA models.

\begin{theorem}\label{thm:planarity}
 Let\/ $\M$ be a connected PCA model.  For every $k \in \Range{\wrap}$, two internal walks\/ $\W$ and\/ $\W'$ of\/ $\Syn^k$ have a common vertex if and only if\/ $\Gr_p(\W)$ and $\Gr_q(\W')$ share a point for some $p,q \geq 0$. Furthermore, a vertex $A$ belongs to\/ $\W$ and\/ $\W'$ if and only if\/ $\Pos_i(A)$ belongs to both\/ $\Gr_p(\W)$ and\/ $\Gr_q(\W')$ for some $i,p,q \geq 0$.
% Let\/ $\M$ be a PCA model and $k \in \Range{\wrap}$.  Two internal walks\/ $\W$ and\/ $\W'$ of\/ $\Syn^k$ have a vertex $A$ in common if and only if\/ $\Pos_i(A)$ belongs to both\/ $\Gr_p(\W)$ and\/ $\Gr_q(\W')$ for some $i,p,q \geq 0$.
\end{theorem}

\begin{proof}
 Recall that for $k \in \ORange{\wrap}$, $\Syn^k$ is the subgraph of $\Syn^k_*$ obtained by removing every short nose.  That is, an $i$-nose $A \to B$ is removed from $\Syn^k_*$ when $i < k$ and either $\Hl(A) \neq \bot$ or $\Hr(B) \neq \bot$.  Let $\Syn^k_+$ be the subgraph of $\Syn^k_*$ obtained by removing each $i$-nose $A \to B$ such that $i < k$, $\Hl(A) \neq \bot$, and $\Hr(B) \neq \bot$.  Clearly, $\Syn^k_+$ is a supergraph of $\Syn^k$.  For technical reasons, we prove that the theorem holds even when $\Syn^k$ is replaced by $\Syn_+^k$, where $\Syn_+^0 = \Syn^0$.  This stronger version of the theorem can be of interest in other applications.  In this article we hide the definition of $\Syn^k_+$ here, to avoid distractions in the main text.

 Bending $\mathbb{R}^2$, we can identify the vertical lines passing through $i\Cols(\M)$ at the $x$ axis to obtain a cylinder $\mathbb{Y}$ in which $\Pos_i(A)$ is mapped to $\Pos_0(A)$ for every $A \in \A(\M)$ and $i \geq 0$.  This mapping transforms the drawing of $\Syn^k_+$ from $\mathbb{R}^2$ to $\mathbb{Y}$, where each $i$-arrow for $A \to B$ is mapped to the $0$-arrow of $A \to B$.  Clearly, all the curves $\Gr_p(\W)$, $p \geq 0$, of an internal walk $\W$ are mapped to the same curve of $\mathbb{Y}$.  Thus, it suffices to show that $\Syn^k_+$ has no crossing arrows when drawn in $\mathbb{Y}$ or, equivalently, that every edge has a corresponding $b$-arrow ($b \geq 0$) that is crossed by no other $b$-arrow. We prove the latter by induction on $k$.
 
 For the base case $k = 0$ we observe four facts that together imply the no $0$-arrow of $\Syn^0_+ = \Syn^0$ is crossed by another $0$-arrow.  For $r \in \Range{\Rows(\M)}$, let $A_0^r < \ldots < A_{q}^r$ be the vertices of $\Syn^0_+$ in row $r$ ($q$ depends on $r$).  Recall that $A_{i+1}^r = \RR(A^r_i)$ for $i \in \Range{q}$ by \eqref{eq:row}.  Hence, $A^r_i \to A^r_{i+1}$ is a forward $0$-nose and $\W_r = A^r_0, \ldots, A_{q}^r$ is a path of the backbone of $\Syn^0_+$.  Fact~1 then follows by Corollary~\ref{cor:walk drawing}: $\Gr_0(\W_r)$ is the graph of the constant function $x \mapsto r$ in the domain $[\Col_0(A_1^r), \Col_0(A_{q}^r)]$.  Fact 2 follows by definition: if $r < \Rows(\M)-1$, then the $0$-arrow corresponding to the backward $0$-nose $A_{q}^r \to A_0^{r+1}$ goes from $(\Col_0(A_{q}^r), r)$ to $(\Col_1(A_0^{r+1}), r+1)$.  Fact 3 follows by Lemma~\ref{lem:jmp of noses and hollows}: if $r > 0$, then the $0$-arrow corresponding to a $1$-hollow $A_i^r \to B$ starts at $(\Col_0(A_i^r), r)$ and ends at $(\Col_0(B), r-1)$.  Finally, since $B \leq B'$ when $A \to B$ and $A' \to B'$ are $1$-hollows with $A \leq A'$, Facts 1~and~3 imply Fact~4: if $A_i^r \to B$ and $A_j^r \to B'$ are hollows for $i<j$ (i.e., $\Col_0(A_i^r) < \Col_0(A_j^r)$), then $\Row(B) = \Row(B') = r-1$ and $\Col_0(B) < \Col_0(B')$.  Facts~1--4 imply that no $0$-arrow of $\Syn^0_+$ is crossed by another $0$-arrow; in order to have a clearer vision of this we can observe an example Mita's drawing on a cylinder (\cref{fig:syngraph-cylinder-0}).

\begin{figure}
	\centering
    \includegraphics{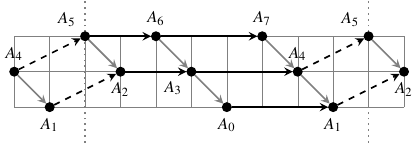}
    
	\caption{Mitas' drawing of $\Syn_+^0=\Syn^0$ for the model $\M$ in \cref{fig:example-syngraph}, rendered in the cylinder $\mathbb{Y}$.  In this drawing, the vertical dotted lines are identified with each other. This provides an example for Facts 1--4: Fact 1 says that any forward $0$-noses connecting two vertices at row $r$ are drawn as horizontal arrows at height $r$; Fact 2 says that backward $0$-noses connect the rightmost vertex at some row with the leftmost vertex at the next row; Fact 3 shows that all $0$-arrows for hollows connect vertices from some row to the previous row (and they cannot cross the arrow of a backward nose), and Fact 4 says that arrows derived from hollows that connect the same rows cannot intersect each other. In conclusion, $\Syn^0_+$ has no crossing arrows.}\label{fig:syngraph-cylinder-0}
\end{figure}%refchange-R10
  
 For the inductive step we apply an algorithm to transform $\Syn^k_+$ into $\Syn^{k+1}_+$.  In a first phase, the algorithm inserts a nose $A \to \Hr(B)$ for every $k$-nose $A \to B$ of $\Syn^k_+$ with $\Hr(B) \neq \bot$.  In a second phase, the algorithm removes those $k$-noses that do not belong to $\Syn_+^{k+1}$.  This algorithm is correct by Lemma~\ref{lem:nose equivalences}.  Clearly, the removal of edges and the insertion of external noses create no crossing arrows in the drawing of $\Syn_+^{k+1}$.  Then, to prove the inductive step, it suffices to consider only the first phase of the algorithm for the case in which the inserted nose is internal.  Let $\Syn_i$ be the graph obtained immediately after the $i$-th $(k+1)$-nose was inserted.  We show by induction on $i$ that some $b$-arrow corresponding to the inserted edge crosses no other $b$-arrows ($b \geq 0$).  The case $i = 0$ follows by induction on $k$ because $\Syn_0 = \Syn^k_+$.  For $i > 0$, suppose the $(k+1)$-nose $A \to \Hr(B)$ is inserted to obtain $\Syn_{i}$ from $\Syn_{i-1}$ (\cref{fig:planarity}).  

 For $j \in \CRange{k+1}$, let $A_j = \Fr^j(A_0)$, $B_j = \RR(A_j)$, $X_j = \Fl(A_{j+1})$, and $Y_{j+1} = \Fr(B_j)$.  If $j \leq k$, then $A_{j+1} \to X_{j}$ and $Y_{j+1} \to B_j$ are a $1$-hollows because $\Fr(X_j) = \Fr(A_j) = A_{j+1}$ and $\Fl(Y_{j+1}) = \Fl(B_{j+1}) = B_j$.  Thus, $A_{j+1} \to X_j$ and $Y_{j+1} \to B_j$ belong to $\Syn_i$ (\cref{fig:planarity}).  Moreover, if $Z$ is a vertex with $X_j < Z < A_j$, then $\Fr(Z) = \Fr \circ \RR(Z) = \Fr(A_j)$, thus $\Hl(Z) = \bot$ and, consequently, $Z \to \RR(Z)$ is a $0$-nose of $\Syn_+^k$ and of $\Syn_{i-1}$. That is, $\Syn_{i-1}$ has a path of $0$-noses from $X_j$ to $A_j$.  Similarly, $\Syn_{i-1}$ has a path of $0$-noses from $B_{j+1}$ to $Y_{j+1}$ (\cref{fig:planarity}).  By Lemma~\ref{lem:nose equivalences}, $\Hr(B) = B_{k+1}$, thus $\Syn_i$ has a path $\W$ from $A_{k+1}$ to $B_0$ that contains the $(k+1)$-nose $A_0 \to B_{k+1}$ together with every hollow $A_{j+1}$ to $X_j$, every hollow $Y_{j+1} \to B_j$, every path of $0$-noses from $X_j$ to $A_j$, and every path of $0$-noses from $B_{j+1}$ to $Y_{j+1}$, for $j \in \CRange{k}$.  By Corollary~\ref{cor:walk drawing}, $\Gr_0(\W)$ is the graph of a partial function from $\Pos_0(A_{k+1})$ to $\Pos_b(B_0)$, where $b$ is the number of backward noses of $\W$ (\cref{fig:planarity} depicts the case $b = 0$).
 
%  Let $A_j = \Fr^j(A)$ for every $j \in \CRange{k+1}$, $X_j = \Fl(A_{j+1})$ for $j \in \CRange{k}$.   Note that $A_{j+1} \to X_{j}$ is a $1$-hollow for every $j \in \CRange{k}$ because $\Fr(X_j) = \Fr(A_j) = A_{j+1}$.  Thus, $A_{j+1} \to X_j$ belongs to $\Syn_i$.  Moreover, Corollary~\ref{cor:S_0 to S_k} implies that the nose of $\Syn^k$ starting at $Z$ is a $0$-nose for every $X_{j} < Z \leq A_j$ because $\Hr(Z) = \bot$.  Consequently, $\Syn_i$ has a path of $0$-noses from $X_j$ to $A_j$ (\cref{fig:planarity}).  Now, let $B_0 = \RR(A_0)$. By Lemma~\ref{lem:alternative nose definition}, $B' = \Hr^{k+1}(B_0)$, thus $B_j = \Hr^j(B_0)$ is well defined for every $j \in \CRange{k+1}$.  Let $Y_{j+1} = \Fr(B_j)$ for every $j \in \CRange{k}$.  By Observation~\ref{obs:f vs h}, $\Fl(Y_{i+1}) =  B_j$, thus $Y_{j+1} \to B_j$ is a $1$-hollow that belongs to $\Syn_i$.  Similarly as before, Corollary~\ref{cor:S_0 to S_k} implies that $\Syn_i$ has a path of $0$-noses that joins $Y_j$ to $B_j$ for every $j \in \CRange{k+1} - \{0\}$ (\cref{fig:planarity}).  Summing up, $\Syn_i$ has a path $\W$ from $A_{k+1}$ to $B_0$ that contains the $(k+1)$-nose $A_0 \to B_{k+1}$ together with every hollow $A_{j+1}$ to $X_j$, every hollow $Y_{j+1} \to B_j$, every path of $0$-noses from $X_j$ to $A_j$, and every path of $0$-noses from $B_{j+1}$ to $Y_{j+1}$, for $j \in \CRange{k}$.  By Observation~\ref{cor:walk drawing}, $\Gr_0(\W)$ is the graph of a partial function from $\Pos_0(A_{k+1})$ to $\Pos_b(B_0)$, where $b$ is the number of backward noses of $\W$ (\cref{fig:planarity} depicts the case $b = 0$).

 \begin{figure}
	\centering
	\includegraphics{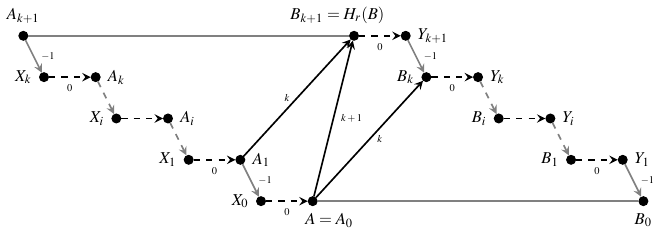} 
    
	\caption[]{Proof of the inductive step of Theorem~\ref{thm:planarity} for the case $b = 0$.  All the $0$-arrows in the picture belong to the drawing of $\Syn_i$, gray arrows correspond to $1$-hollows, black arrows correspond to noses, and dashed arrows correspond to paths. Labels in the arrows denote the jump of the corresponding arrow. The gray lines, that are not part of the drawing of $\Syn_i$, are used to enclose the polygon $P$ defined in Theorem~\ref{thm:planarity}.  Clearly, the polygon $Q$ of Theorem~\ref{thm:planarity} is contained in $P$, whereas the $(k+1)$-nose is inside $Q$.\label{fig:planarity}}
 \end{figure}

 Let $b(j)$ be the number of backward noses of the subpath of $\W$ from $A_{k+1}$ to $A_j$ for $j \in \CRange{k}$.  Also, let $\pi = 1$ if $A_0 \to B_{k+1}$ is backward and $\pi = 0$ otherwise.  By definition, the arrow corresponding to $A_0 \to B_{k+1}$ in $\Gr_0(\W)$ is a $b(0)$-arrow that goes from $\Pos_{b(0)}(A_0)$ to $\Pos_{b(0)+\pi}(B_{k+1})$.  By Lemma~\ref{lem:nose equivalences}, $B_{j+1} = \Hr(B_j)$, thus $\Row(B_{j+1}) = \Row(B_{j})+1$ by~\eqref{eq:row}.  Then, the subpath of $\W$ from $B_{j+1}$ to $B_j$ has exactly one $1$-hollow and $\Jmp = -1$ and, therefore, it has no backward noses by Corollary~\ref{cor:jmp and bal}.  Then, by induction, it follows that the arrow corresponding to $Y_1 \to B_0$ in $\Gr_0(\W)$ is a $(b(0)+\pi)$-arrow, thus $b = b(0)+\pi$. 
  
 By definition, $B_j = \RR(A_j)$ for every $j \in \CRange{k+1}$.  Then, $\Row(A_j)$ is equal to either $\Row(B_j)$ or $\Row(B_{j})-1$, the latter being true if and only if $B_j$ is leftmost.  In other words, the subpath of $\W$ from $A_{j+1}$ to $A_j$ has $\Jmp \in \{-1,0\}$ and exactly one hollow.  Thus, by Corollary~\ref{cor:jmp and bal}, this subpath has either $0$ or $1$ backward noses and it has $1$ backward nose only if $\Row(A_{j+1}) = \Row(A_{j}) = \Row(B_j)$. Consequently, $b(j) + \pi \leq 1$ follows by induction for every $j \in \CRange{k}$.  Altogether, this means that $\Syn^0$ has a $b(j)$-arrow from $\Pos_{b(j)}(A_j)$ to $\Pos_{b}(B_j)$ that corresponds to the $0$-nose $A_j \to B_j$.  As discussed in the base case, this implies that no vertex $A$ in $\Syn_{i-1}$ has $\Pos_p(A)$, $p \geq 0$, in the interior of the polygon $P$ whose borders are determined by the curve of $\Gr_0(\W)$ from $\Pos_0(A_{k+1})$ to $\Pos_{b(0)}(A_0)$, the curve of $\Gr_0(\W)$ from $\Pos_b(B_{k+1})$ to $\Pos_b(B_0)$, the line from $\Pos_{b(0)}(A_0)$ to $\Pos_b(B_0)$, and the line from $\Pos_0(A_{k+1})$ to $\Pos_b(B_{k+1})$ (\cref{fig:planarity}). 
 
 By Lemma~\ref{lem:nose equivalences}, $A_1 \to B_{k+1}$ is a $k$-nose of $\Syn^k_+$ that also belongs to $\Syn_{i-1}$ because the second phase of the algorithm was not executed yet.  Thus, there is a $b(1)$-arrow $A_1 \to B_{k+1}$ in the drawing of $\Syn_{i-1}$ that, by construction, is completely inside $P$.  Similarly, the $b(0)$-arrow corresponding to the $k$-nose $A_0 \to B_k$ is also inside in $P$.  Therefore, the sides of $Q$ are arrows of the drawing of $\Syn_{i-1}$ for the polygon $Q$ whose boundary is determined by the $b(0)$-arrow $A_0 \to B_k$, the $b(1)$-arrow $A_1 \to B_{k+1}$, the curve of $\Gr_0(\W)$ between $\Pos_{b(1)}(A_1)$ and $\Pos_{b(0)}(A_0)$ and the curve of $\Gr_0(\W)$ between $\Pos_b(B_{k+1})$ and $\Pos_b(B_k)$.  By construction, $Q \subseteq P$, thus no point in the interior of $Q$ corresponds to $\Pos_p(A)$ for every arc $A$ of $\M$ and every $p \geq 0$ (\cref{fig:planarity}). Then, by the inductive hypothesis, no arrow of $\Syn_{i-1}$ crosses the borders of $Q$. Moreover, by Lemma~\ref{lem:nose equivalences}, $A_1 \to B_k$ is the only edge of $\Syn_*^{k+1}$ with an arrow inside $Q$, but it does not belong to $\Syn_+^k$ because $\Hl(A_1) = A_0 \neq \bot$ and $\Hr(B_k) = B_{k+1} \neq \bot$.  Consequently, no arrow of $\Syn_{i-1}$ intersects the interior of $Q$ and, therefore, the $b(0)$-arrow $A_0 \to B_{k+1}$, that belongs to the interior of $Q$ as it goes from $\Pos_{b(0)}(A_0)$ to $\Pos_b(B_{k+1})$  (\cref{fig:planarity}), is crossed by no arrows of $\Syn_i$.  Summing up, $\Syn_i$ has no pair of crossing arrows.
\end{proof}

%%%%%%%%%%%%%%%%%%%%%%%%%%%%%%%%%%%%%

\section{A simple and efficient algorithm for the multiplicative problem}
\label{sec:recognition}

In this section we devise a simple and efficient algorithm to solve \kMult for an input PCA model $\M$.  By Theorem~\ref{thm:no_positive_cycles}, it suffices to determine if there exist $\Circ$ and $\Len$ such that $\Sep_{\Circ,\Len}(\W) \leq 0$ for every cycle $\W$ of $\Syn^k$.  One of the advantages of arranging the arcs of $\M$ into rows is that we can immediately conclude that $\Sep_{\Circ,\Len}(\W) \leq 0$ when $\W$ is internal and $\Len \geq 2n$.  Just note that: $\Ext(\W) = 0$ because $\W$ is internal, $\Jmp(\W) = 0$ because $\W$ starts and ends at the same row, and $\W$ has at least one backward edge by Lemma~\ref{lem:acyclic}.  Then, $\Sep_{\Circ,\Len}(\W) \leq -\Len + 2\nose(\W) \leq 0$ by Corollary~\ref{cor:jmp and bal} and \eqref{eq:sep}.  Indeed, if all cycles of $\Syn^k$ were internal, then $\M$ would essentially be a PIG model, and thus $k$-multiplicative by Theorem~\ref{thm:multiplicative PIG}.%refchange-R11

The reason there are PCA models which are not $k$-multiplicative for some $k$ is because their synthetic graphs have some cycles with $\Ext\neq 0$ that leave no feasible value for $\Circ$ when viewed as inequalities of the full system $\F_{\Circ,\Len}^k$. To give a general intuition, consider a circuit $\W_N$ of $\Syn^k$ with $\Ext(\W_N)<0$ and search for the values of $\Circ$ so that $\Sep_{\Circ,\Len}(\W_N)\leq 0$, we get a lower bound for $\Circ$, while if we do the same for a circuit $\W_H$ with $\Ext(\W_H)>0$, we get an upper bound for $\Circ$. Certainly, when those two bounds are incompatible, $\M$ is not $k$-multiplicative. As we will see, we can find a pair of cycles of $\Syn^k$ such that they are compatible if and only if $\M$ is $k$-multiplicative.%refchange-R11

The properties we will study are for internal cycles of synthetic graphs, so we will use a technique that will allow us to create an alternative model in which we have internal circuits that correspond to circuits of $\Syn^k$ which are not internal; this is the loop unrolling technique.%refchange-R11

Let $\Circ$ be the circumference of the circle of a PCA model $\M$ and $\Copies \in \mathbb{N}$.  The \emph{$\Copies$-unrolling} of $\M$ (\cref{fig:example-unroll}) is the PCA model $\Copies \Unroll \M$ whose circle has circumference $\Copies\Circ$ that has $\Copies$ arcs $A_0$, $A_1$, $\ldots$, $A_{\Copies-1}$ for every $A \in \A$ such that, for $i \in \Range{\Copies}$:
\begin{displaymath}
  s(A_i)=s(A) + i\Circ\mbox{ and } t(A_i)=t(A) + \Circ(i+\Boolean{s(A) > t(A)}) \bmod \Copies\Circ.
\end{displaymath}
We refer to $A_i$ as being a \emph{copy} of both $A$ and $A_j$ (for $j \in \Range{\Copies} - \{i\}$); see~\cref{fig:example-unroll}.  Also, we say that a row $\ARow$ of $\Copies\Unroll\M$ is a \emph{copy} of another row $\ARow'$ of $\Copies\Unroll\M$ when $|\ARow| = |\ARow'|$ and the $i$-th arc of $\ARow$ is a copy of the $i$-th arc of $\ARow'$ for every $0 \leq i < |\ARow|$.  An important feature of $\Copies\Unroll\M$ is that it is highly repetitive when $\Copies$ is large enough.

\begin{figure}
    \mbox{}\hfill\includegraphics{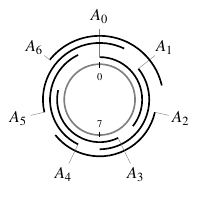}\hfill\includegraphics{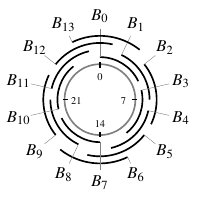}\hfill\includegraphics{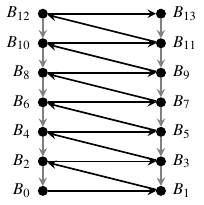}\hfill\mbox{}
	\caption{From left to right: a model $\M$ with arcs $A_0 < \ldots < A_6$, the model $2\Unroll\M$ with arcs $B_0, \ldots, B_{13}$, and the interior of $2\Unroll\Syn^0$.  The copies of $A_i$ are $B_{i}$ and $B_{i+7}$.  Note that $B_i$ is leftmost if and only if $B_{i+7}$ is rightmost, thus each backward $0$-nose $A_i \to A_{i+1}$ of $\Syn^0$ has a forward copy $B_{i+7} \to B_{i+8}$ in $2\Unroll\Syn$.  Also note that the external $0$-nose $A_6 \to A_0$ has an internal copy $B_6 \to B_7$, while the external $1$-hollows $A_0 \to A_5$ and $A_1 \to A_6$ have internal copies $B_7 \to B_5$ and $B_8 \to B_6$.}\label{fig:example-unroll}
\end{figure}

\begin{lemma}\label{lem:repetitive}
 Let\/ $\M$ be a PCA model and $r = \Rows(\Copies\Unroll\M)$ for $\Copies > 0$.  If $r \geq n$, then there exist $x, z\in\CRange{n} - \{0\}$ such that the row $i+jz$ of\/ $\M$ is a copy of the row $i$ of\/ $\M$ for every $i \geq x$, $j \geq 0$, and $i+jz < r$.
\end{lemma}

\begin{proof}
 Let $L_i$ be the leftmost arc in the $i$-th row $\ARow_i$ of $\Copies\Unroll\M$ for $i \in \Range{r}$.  Since $r \geq n$, there exist $x,z\in \CRange{n} - \{0\}$ such that $L_{x+z}$ is a copy of $L_{x}$.  By~\eqref{eq:row}, if $L_{i}$ is a copy of $L_{i+j}$, then $\ARow_i$ is a copy of $\ARow_{i+j}$, thus $L_{i+1}$ is a copy of $L_{i+j+1}$ and, by induction, $\ARow_{i+jz}$ is a copy of $\ARow_{i}$ for $i \geq x$, $j \geq 0$ and $i+jz < r$.
\end{proof}

For $k \in \Range{\wrap}$, we write $\Copies \Unroll \Syn^k(\M)$ as a shortcut for $\Syn^k(\Copies \Unroll \M)$ and we drop the parameter $\M$ when no confusions are possible.  Following our naming convention, $\Copies \Unroll \Syn^k$ has a vertex called $A$ for each arc $A \in \Copies\Unroll \M$.  Thus, each vertex of $\Copies\Unroll\Syn^k$ is a copy of an arc of $\M$.  We say that a hollow (resp.\ nose) $A' \to B'$ of $\Copies\Unroll\Syn^k$ is a \emph{copy} of the hollow (resp.\ nose) $A \to B$ of $\Syn^k$ to mean that $A'$ is a copy of $A$ and $B'$ is a copy of $B$.  It is not hard to see that every edge $A' \to B'$ of $\Copies\Unroll\Syn^k$ is a copy of an edge $A \to B$ of $\Syn^k$ and, conversely, $A \to B$ has $\Copies$ copies in $\Copies\Unroll\Syn^k$.  Some of these copies can be internal while others are external and, moreover, some internal copies can be forward while others are backward (\cref{fig:example-unroll}).  Thus, $\Jmp(A' \to B')$ need not be equal to $\Jmp(A \to B)$.  Similarly, a walk $\T$ of $\Copies\Unroll\Syn^k$ is a \emph{copy} of the walk $\W$ of $\Syn^k$ when the $i$-th edge of $\T$ is a copy of the $i$-th edge of $\W$.  Again, $\W$ has $\Copies$ copies in $\Copies\Unroll\Syn^k$, one starting at each copy of its first arc, and these copies may have different $\Jmp$ values.

By definition, each external edge $A \to B$ of $\Syn^k$ has a copy $A' \to B'$ in $\Copies\Unroll\Syn^k$ that is internal when $\Copies > 1$ (\cref{fig:example-unroll}).  Similarly, each walk $\W$ of $\Syn^k$ has an internal copy in $\Copies\Unroll\Syn^k$ when $\Copies > |\W|$. When $\W$ is a circuit, any internal copy $\T$ of $\W$ is a walk between two copies $A_0$ and $A_1$ of the same vertex $A$.   If $A_0 \neq A_1$, then $\Sign(\Ext(\W))=-\Sign(\Jmp(\T))$, as the internal copies in $\T$ of the external noses and hollows of $\W$ have $\Jmp>0$ and $\Jmp<0$, respectively. Moreover, $|\Ext(\W)|$ counts one plus the number of copies of $A$ between $A_0$ and $A_1$.

Considering a large enough $\Copies$, we can manipulate all the circuits of $\Syn^k$ as if they were internal walks of $\Copies\Unroll\Syn^k$. In this work this is the sole purpose of unrolled models.  Thus, even though $\Copies \Unroll \Syn^k$ arises from its own system of difference constraints, we have no interest in solving this system for $\Copies\Unroll\M$.  Moreover, we restrict ourselves to \emph{strict PCA} (SPCA) models, i.e. PCA models that are not PIG, because $\Copies\Unroll\M$ is disconnected when $\M$ is PIG.  For these models we can take advantage of Theorem~\ref{thm:planarity} and Corollary~\ref{cor:walk drawing} to depict the drawing of $\Copies\Unroll\Syn^k$ (\cref{fig:non crossing cycles,fig:crossing cycles,fig:greedy,fig:twisting,fig:dually}).  A box labeled $\Copies\Unroll\Syn^k$ is used to frame a portion of the drawing.  Inside this box, a dot labeled $A$ is used to represent $\Pos_p(A)$ for every $A \in \A(\M)$ and $p \geq 0$.  As $\Pos_p(A)$ is defined for every $p \geq 0$, different dots can share a same label.  Similarly, for a walk $\W$ of $\Copies\Unroll\Syn^k$, $\Gr_p(\W)$ is depicted with an unlabeled curve; the identity of $\Gr(\W)$ can be decoded from the traversed dots.  Dashed horizontal and vertical lines are used to represent rows and columns of $\Copies\Unroll\Syn^k$.  The labels outside the box indicate the number of the corresponding row or column.  Note that figures are out of scale because their purpose is to explain a behavior described in the text.  Thus, these pictures as referred to as \emph{schemes}.

Twisting bound %refchange-R12
values are the key concept to determine if $\Copies$ is large enough.  Say that $\Spiral \in \mathbb{N}$ is a \emph{twisting bound} %refchange-R12
of $\Syn^k$, if for every $\Copies \geq \Spiral$, no walk $\W$ of $\Syn^k$ has a forward copy $\T$ in $\Copies\Unroll\Syn^k$ with $\Jmp(\T) \geq \Spiral$.  In other words, any copy of walk of $\Syn^k$ in $\Copies\Unroll\Syn^k$ that is internal cannot move up $\Spiral$ rows without taking a backward edge of $\Copies\Unroll\Syn^k$.  %refchange-R12
Although we do not prove it explicitly, the techniques in this section can be applied to show that, geometrically, if $\Spiral$ is a twisting bound %refchange-R12
and $\T$ is a large path of $\Copies\Unroll\Syn^k$ moving upward ($\Jmp(\T) \geq \Spiral$), then $\Gr_p(\T)$ resembles a helix when drawn in the cylinder $\mathbb{Y}$ obtained by identifying the vertical lines passing through $i \Cols(\M)$, $i \geq 0$.  Theorem~\ref{thm:equivalence} applies a restricted notion of twisting bounds %refchange-R12
as a tool to characterize those SPCA models that are equivalent to some $k$-multiplicative model.  Specifically, $\Spiral \in \mathbb{N}$ is a \emph{hollow} (resp.\ \emph{nose}) \emph{cycle twisting bound} %refchange-R12
of $\Syn^k$ if for every cycle $\W$ of $\Syn^k$ with $\Ext(\W) > 0$ (resp.\ $\Ext(\W) < 0$) and every $\Copies \geq \Spiral$ it happens that no internal copy of $\Spiral\Unroll\W$ in $\Copies\Unroll\Syn^k$ is forward.  For simplicity, $\Spiral$ is a \emph{cycle twisting bound} of $\Syn^k$ when $\Spiral$ is either a nose or hollow cycle twisting bound.  %refchange-R12
Clearly, every twisting bound is a nose cycle twisting bound, and thus a cycle twisting bound, %refchange-R12
because $\Jmp(\T) \geq -\Spiral\Ext(\W)$ when $\T$ is a copy of $\Spiral\Unroll\W$, for every cycle $\W$ with $\Ext(\W) < 0$.  Before Theorem~\ref{thm:equivalence}, Lemma~\ref{lem:cycle ext < 0} records the fact that the point $0$ of $C(\M)$ is crossed in the clockwise direction when enough noses are traversed.  The analogous fact that $0$ can be crossed in a counterclockwise direction follows as Corollary~\ref{cor:cycle ext > 0}.

\begin{lemma}\label{lem:cycle ext < 0}
 If\/ $\M$ is a connected PCA model and $k \in \ORange{\wrap}$, then\/ $\Syn^k$ has a cycle with\/ $\Ext < 0$.  
\end{lemma}

\begin{proof}
 Suppose $\W$ is a circuit of $\Syn^k$ with $\Ext(\W) \geq 0$.  As stated above, if $\Copies > |\W|$, then $\W$ has an internal copy $\T$ in $\Copies\Unroll\Syn$ with $\Jmp(\T) \leq 0$.  By Corollary~\ref{cor:jmp and bal}, $\Bal(\T) = \Jmp(\T) - b$ where $b$ is the number of backward noses in $\T$.  Clearly, $\T$ is a circuit when $\Jmp(\T) = 0$ because $\T$ joins two copies of a same vertex of $\W$.  Then, $\Bal(\T) <0 $ follows by Lemma~\ref{lem:acyclic} because $b > 0$ if $\T$ is a circuit, while $\Jmp(\T) < 0$ otherwise.  By definition, $\W$ and $\T$ have the same number of noses and hollows; equivalently, $\Bal(\W) = \Bal(\T) < 0$.  Consequently, by~\eqref{eq:sep}, $\Sep_{0,3n}(\W) = 3n\Bal(\W) + 2\nose(\W) < 0$.  Then, by Theorem~\ref{thm:no_positive_cycles}, either $\Syn^k$ has a cycle with $\Ext < 0$ or $\M$ is equivalent to a $k$-multiplicative $(0,3\Len)$-CA model.  The latter is clearly impossible.   
\end{proof}

We are now ready to present the relation between cycle twisting bounds and the $\kMult$ problem. Essentially, the idea is that a model $\M$ is $k$-multiplicative if and only if it has a cycle twisting bound. We show first that if $\M$ does not have a cycle twisting bound, then it is possible to find a pair of cycles of $\Syn^k$ such that there are no values for $\Circ$ and $\Len$ whose $\Sep_{\Circ,\Len}$ inequalities can be satisfied toghether, thus making $\M$ not $k$-multiplicative. Then, we show that the existence of a cycle twisting bound implies that every pair of cycles of $\Syn^k$ with different sign for $\Ext$ must have a vertex in common. The idea for this is geometric: If there is a cycle twisting bound $\Spiral$, then we can see that there is a value $\Copies$ for which we can guarante that there are copies of these cycles in the graph $\Copies\Unroll\Syn^k$ whose drawings intersect each other (and thus the cycles have at least one vertex in common). Finally, we provide values of $\Len$ and $\Circ$ for which every cycle of $\Syn^k$ has $\Sep_{\Circ,\Len}\leq 0$, provided that every pair of cycles with different sign for $\Ext$ have a vertex in common. This shows that $\M$ is $k$-multiplicative.%refchange-R13

\begin{theorem}\label{thm:equivalence}
 The following statements are equivalent for an SPCA model $\M$ and $k \in \ORange{\wrap}$.%
 \begin{enumerate}
  \item $\M$ is equivalent to a $k$-multiplicative model.\label{thm:multip}
  \item Some $\Spiral \in \mathbb{N}$ is a cycle twisting bound %refchange-R12
  of\/ $\Syn^k$.\label{thm:twister}
  \item Any two cycles\/ $\W$ and\/ $\W'$ of\/ $\Syn^{k}$ with\/ $\Ext(\W)\Ext(\W') < 0$ have a vertex in common.\label{thm:crossing}
 \end{enumerate}
\end{theorem}

\begin{proof}
 $\ref{thm:multip} \Rightarrow \ref{thm:twister}$. 
 Let $h = \max\{10, n(k+1)\}$ and suppose $\Spiral = h^4$ is not a nose cycle twisting bound %refchange-R12
 of $\Syn^k$.   Then, there exists a cycle $\W_N$ of $\Syn^k$ with $\Ext(\W_N) < 0$ and a value $\Copies \geq \Spiral$ such that $\Copies\Unroll\Syn^k$ has a forward copy $\T$ of $\Spiral\Unroll\W_N$.  For convenience, we first locate a forward copy of a portion of $\T$ in a controlled location of $\Copies\Unroll\Syn^k$ (\cref{fig:non crossing cycles}(a)).  For this purpose, let $A_0$ be the first vertex of $\T$ and $A_1, \ldots, A_\Spiral$ be the other copies of $A_0$ in the order they are traversed by $\T$.  Note that $\Row(A_{i+1}) > \Row(A_i)$ for $i \in \Range{\Spiral}$ because $\Ext(\W_N) < 0$ and, hence, $\Row(A_{4h}) \geq 4h$ (\cref{fig:non crossing cycles}(a)).  Then, by Lemma~\ref{lem:repetitive}, $A_{4h}$ has a copy $X_0$ such that $\Row(X_0) \in [3h, 4h)$ and the row containing $X_0$ is a copy of the row containing $A_{4h}$ (\cref{fig:non crossing cycles}(a)).  For $i \in \Range{h^3}$, let $\T_i$ be the copy of $i\Unroll\W_N$ that begins at $X_0$ and $X_1, \ldots, X_i$ be the other copies of $X_0$ in the order thy are traversed by $\T_i$.  By Corollary~\ref{cor:jmp and bal}, every internal edge of $\Copies\Unroll\Syn^k$ has $\Jmp \geq -1$.  Thus, if $\Row(X_{i-1}) \geq 3h$, then all the edges traversed by $\T_i$ between $X_{i-1}$ and $X_{i}$ are internal, thus all the vertices of $\T_i$ have $\Row \geq 2h$ because $|\W_N| \leq h$.  Hence, $\Row(X_i) > \Row(X_{i-1})$ because $\Ext(\W_N) < 0$.  By induction, this means that $\T_i$ is internal and traverses vertices with $\Row \geq 2h$.  Altogether, Lemma~\ref{lem:repetitive} implies that the row containing the $j$-th vertex traversed by $\T_i$, $j \in \Range{|\T_i|}$, is a copy of the row containing the $j$-th vertex visited by $\T$ after $A_{4h}$.  Then, every edge of $\T_i$ is a copy of an edge of $\T$ and, therefore, $\T_i$ is forward as well (\cref{fig:non crossing cycles}(a)).
 
 \begin{figure}
 \centering 
  \begin{tabular}{c@{\hspace{3mm}}c}
    \includegraphics{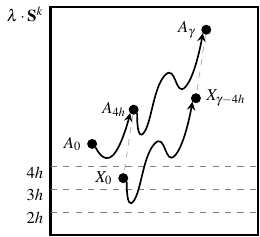} & \includegraphics{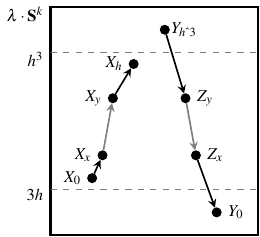} \\
	(a) & (b)
  \end{tabular}

	\caption[]{Schemes for Theorem~\ref{thm:equivalence} ($\ref{thm:multip} \Rightarrow \ref{thm:twister}$). Paths are forward, and (b) depicts $\Q_H$ and $\Q_N$ in gray.}\label{fig:non crossing cycles}
 \end{figure}

 Summing up, the previous paragraph proves that if $\Spiral$ is not a cycle twisting bound %refchange-R12
 of $\Syn^k$, then $\Copies\Unroll\Syn^k$ has a forward copy $\T_N = \T_{h}$ of $h\Unroll\W_N$ whose first vertex $X_0$ has $\Row(X_0) \in [3h, 4h)$.  Similar arguments to those above imply that $\Syn^k$ has a cycle $\W_H$ with $\Ext(\W_H) > 0$ such that $\Copies\Unroll\Syn^k$ has a forward copy $\T_H$ of $(h^3)\Unroll\W_H$ whose last vertex $Y_0$ has $\Row(Y_0) \in [2h, 3h)$.  Let $Y_0, \ldots, Y_{h\hat{\ }3}$ be the copies of $Y_0$ traversed by $\T_H$ in the reverse order.  (For the sake of notation, we write $h\hat{\ }i$ as a replacement of $h^i$ to avoid double subscripting.)  Note that $\Row(Y_{i+1}) > \Row(Y_i)$ for $i \in \Range{h^3}$ because $\Ext(\W_H) > 0$.
 
 By Lemma~\ref{lem:jmp of noses and hollows}, every edge of $\T_H$ has $\Jmp \in [-1, h]$.  Thus, $\T_H$ has an vertex $Z_i$ at $\Row(X_i)$ for every $i \in \Range{h}$ because $\Row(Y_0) < 3h \leq \Row(X_0)$ and $\Row(X_h) \leq \Row(X_0) + h|\W_N| < h^3 \leq \Row(Y_{h\hat{\ }3})$.  Then, as $\Syn^k$ has $n \leq h$ vertices, it follows that $Z_x$ and $Z_y$ are copies of the same vertex of $\W_H$ for some $x,y \in \Range{h}$, $x < y$.  Summing up, $\T_N$ has a forward subpath $\Q_N$ from $X_x$ to $X_y$, whereas $\T_H$ has a forward subpath $\Q_H$ from $Z_y$ to $Z_x$ (\cref{fig:non crossing cycles}(b)).  By definition, $\Row(X_x)=\Row(Z_x)$ and $\Row(X_{y}) = \Row(Z_y)$.  Then, by Corollary~\ref{cor:jmp and bal},
 \begin{displaymath}
  \Bal(\Q_N) = \Jmp(\Q_N) = \Row(X_y) - \Row(X_x) = -\Jmp(\Q_H) = -\Bal(\Q_H). \tag{*}
 \end{displaymath}
 By construction, $\Q_N$ is a copy of $\W_N' = (y-x)\Unroll\W_N$.  Thus, $\Q_N$ and $\W_N'$ have the same number of noses and hollows and, therefore, $\Bal(\Q_N) = \Bal(\W_N')$.  Similarly, $\Q_H$ is a copy of $\W_H' = z\Unroll\W_H$ for some $z \geq 0$, thus $\Bal(\Q_H) = \Bal(\W_H')$.  Moreover, $\Ext(\W_N')$ equals one plus the number of copies of $X_0$ between $X_x$ and $X_y$ that, in turn, equals one plus the number $-\Ext(\W_H')$ of copies of $Z_x$ between $Z_x$ and $Z_y$.  Then, by~\eqref{eq:sep}, (*), and the fact that $\W_N'$ has at least one nose (because $y > x$), we obtain that:
 \begin{align*}
  \Sep_{\Circ,\Len}(\W_N') &= \Len\Bal(\W_N') + \Circ\Ext(\W_N') + 2\nose(\W_N') > \Len\Bal(\Q_N) + \Circ\Ext(\W_N') \text{ and}\\
  \Sep_{\Circ,\Len}(\W_H') &= \Len\Bal(\W_H') + \Circ\Ext(\W_H') + 2\nose(\W_H') \geq -\Len\Bal(\Q_N) - \Circ\Ext(\W_N'),
 \end{align*}
 for every $\Circ, \Len > 0$.  Then, $\Sep_{\Circ,\Len}(\W_N')+\Sep_{\Circ,\Len}(\W_H')>0$ and, by Theorem~\ref{thm:no_positive_cycles}, $\M$ is equivalent to no $k$-multiplicative $(\Circ,\Len+1)$-CA model regardless of the values of $\Circ$ and $\Len$.
  
 $\ref{thm:twister} \Rightarrow \ref{thm:crossing}$.  Suppose  $\Spiral \in \mathbb{N}$ is a nose cycle twisting bound %refchange-R12
 of $\Syn^k$ and let $\W_N$ be a cycle of $\Syn^k$ with $\Ext(\W_N) < 0$.  If $\Copies > h^6$ for $h = \max\{10+\Spiral, n(k+1)\}$, then $h^3 \Unroll \W_N$ has a copy $\T$ whose first vertex has $\Row \in (2h^3, 2h^3+h]$.  By Lemma~\ref{lem:jmp of noses and hollows}, every edge of $\Copies\Unroll\Syn^k$ has $\Jmp \in [-1, h]$, thus $\T$ is internal and all its vertices have $\Row \in (h^3, h^4)$.  Moreover, as $\Spiral$ is a nose cycle twisting bound %refchange-R12
 and $h \geq \Spiral$, it follows that $\T$ has at least $h^2$ backward edges, thus $\T$ traverses $h$ leftmost copies $A_0, \ldots, A_{h}$ of some vertex of $\W_N$ (\cref{fig:crossing cycles}(a)).  If $A = A_0$ and $B \to A_h$ is the edge of $\T$ to $A_h$, then $\Row(B) > \Row(A)$ because $\Row(A_{h}) \geq \Row(A) + h$ as $\Ext(\W_N) < 0$.  Hence, the walk $\T_N$ from the leftmost vertex $A$ to the rightmost vertex $B$ traverses vertices with $\Row \in (h^3, h^4)$, starting from $\Row(A)$ and ending at $\Row(B) > \Row(A)$ (\cref{fig:crossing cycles}(a)).

 \begin{figure}
  \centering 
  \begin{tabular}{c@{\hspace{3mm}}c}
    \includegraphics{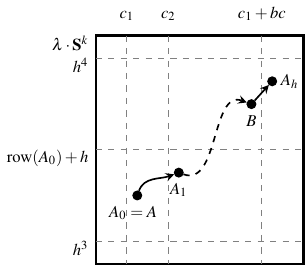} & \includegraphics{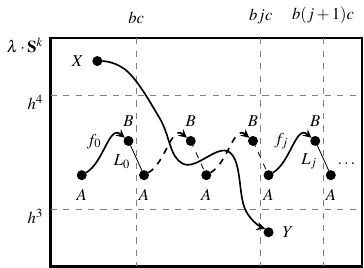} \\
	(a) & (b)
  \end{tabular}

	\caption[]{Schemes for Theorem~\ref{thm:equivalence} ($\ref{thm:twister} \Rightarrow \ref{thm:crossing}$), where $c = \Cols(\M)$.  In (a) $c_1$ and $c_2$ are multiples of $c$.}\label{fig:crossing cycles}
 \end{figure}

 Let $b-1$ be the number of backward edges in $\T_N$ and $j \geq 0$.  By Corollary~\ref{cor:walk drawing}, the curve $\Gr_{bj}(\T_N)$ depicts a  continuous function $f_j$ with domain $[\Col_{bj}(A), \Col_{b(j+1)-1}(B)]$ that is bounded below and above by the constant functions $x \mapsto h^3$ and $x \mapsto h^4$, respectively (\cref{fig:crossing cycles}(b)).  Moreover, if the curve of $f_j$ is extended with the line $L_j$ from $\Pos_{b(j+1)-1}(B)$ to $\Pos_{b(j+1)}(A)$, then the curve of a continuous function $g_j$ with domain $[\Col_{bj}(A),\Col_{b(j+1)}(A)]$ is obtained (\cref{fig:crossing cycles}(b)).  Therefore, $g = \bigcup_{j\geq 0} g_j$ is a continuous function with domain $[\Col_0(A), \infty)$ that is bounded below and above by $x \mapsto h^3$ and $x \mapsto h^4$ (\cref{fig:crossing cycles}(b)).
 
 Let $\W_H$ be any cycle with $\Ext(\W_H) > 0$.  For $i \geq 0$, let $\T_i$ be the copy of $i \Unroll \W_H$ that starts at a vertex $X$ with $\Row(X) \in [h^4, h^4+h)$.  Similarly as above, Lemma~\ref{lem:jmp of noses and hollows} implies that $\T_i$ is internal when its last vertex has $\Row > h^2$.  Then, there exists $i$ such that $\T_H = \T_i$ is internal and ends at a vertex $Y$ with $\Row(Y) < h^2$.  By Corollary~\ref{cor:walk drawing}, $\Gr_1(\T_H)$ is the curve of a continuous function that starts at $\Pos_1(X)$ and ends at a point with ordinate $\Row(Y)$.  Since $\Row(X) \geq h^4$ and $\Row(Y) < h^2$, it follows that $\Gr_1(\T_H)$ crosses the graph of $g$ at some point (\cref{fig:crossing cycles}(b)).  Among such crossing points, let $p = (x,y)$ be the one minimizing $x$.  We claim that $p$ belongs to the graph of $f_j$ for some $j \geq 0$.  Otherwise, $p$ would be a point of $L_j$ for some $j \geq 0$.  As $L_j$ starts at $\Pos_{b(j+1)-1}(B)$ and ends at $\Pos_{b(j+1)}(A)$ and $\Row(B) > \Row(A)$, Lemma~\ref{lem:jmp of noses and hollows} implies that the edge $X' \to Y'$ of $\T_H$ whose arrow contains $p$ satisfies $\Row(B) < \Row(X')$ and $\Row(Y') > \Row(A)$.  As this is impossible, because the first cross between $g$ and $\Gr_1(\T_H)$ happens at a point in which the slope of $\Gr_1(\T_H)$ is smaller (\cref{fig:crossing cycles}(b)), it follows that $p \in \Gr_j(\T_N) \cap \Gr_1(\T_H)$.  Then, $\T_N$ and $\T_H$ have a vertex in common by Theorem~\ref{thm:planarity}, and so do $\W_N$ and $\W_H$.
 
 The proof for the case in which $\Spiral$ is a hollow cycle twisting bound %refchange-R12
 is similar, thus we omit it for the sake of succinctness. We remark, however, that every cycle twisting bound %refchange-R12
 of $\Syn^k$ is a nose cycle twisting bound %refchange-R12
 (this is proven later in Theorem~\ref{thm:fast equivalence}).
 
 $\ref{thm:crossing} \Rightarrow \ref{thm:multip}$. By Lemma~\ref{lem:cycle ext < 0}, $\Syn^k$ has a cycle $\W$ with $\Ext(\W) < 0$.  Let $\Len=4n^2$ and $\Circ$ be the minimum such that $\Sep_{\Circ,\Len}(\W) \leq 0$ for every cycle $\W$ of $\Syn^k$ with $\Ext(\W) < 0$.  Note that $\Circ$ exists because $\Circ$ can be as large as to bound all the other terms of~\eqref{eq:sep}. Moreover, as $\Circ$ is minimum, there exists a cycle $\W_N$ with $\Ext(\W_N) < 0 = \Sep_{\Circ,\Len}(\W_N)$.  We prove that $\Sep_{\Circ,\Len}(\W) \leq 0$ for every cycle $\W$ of $\Syn^k$ and, therefore, $\M$ is equivalent to a $k$-multiplicative $(\Circ,\Len+1)$-CA model by Theorem~\ref{thm:no_positive_cycles}. Consider the following possibilities for $\Ext(\W)$.	
 \begin{discription}[label={\textbf{Case~\arabic*:}},ref={Case~\arabic*}]
  \item $\Ext(\W)<0$, thus $\Sep_{\Circ,\Len}(\W) \leq 0$ by the definition of $\Circ$.
  \item \label{thm:equivalence:ext0} $\Ext(\W)=0$.  If $\Copies > |\W|$, then $\W$ has an internal copy $\T$ in $\Copies \Unroll\Syn^k$.  Clearly, $\Jmp(\T) = 0$ because $\Ext(\W) = 0$.  Hence, $\T$ is a circuit that has $b > 0$ backward edges by Lemma~\ref{lem:acyclic}.  Moreover, $\Bal(\T) = \Bal(\W)$ because they traverse the same number of noses and hollows.  Then, by~\eqref{eq:sep} and Corollary~\ref{cor:jmp and bal},
  \begin{align*}
    \Sep_{\Circ,\Len}(\W) &= \Len\Bal(\W)+\Circ\Ext(\W)+2\nose(\W) \\ &= \Len\Bal(\T) + 2\nose(\W) = -\Len b + 2\nose(\W)  \leq 0.
  \end{align*}

  \item \label{thm:equivalence:repeated}$\Ext(\W)>0$.  By hypothesis, $\W$ and $\W_N$ have a vertex $A$ in common.  Starting at $A$, let $\W_0 = \Ext(\W) \Unroll \W_N + \lvert\Ext(\W_N)\rvert \Unroll \W$, i.e., $\W_0$ is the circuit of $\Syn^k$ that begins at $A$, repeatedly traverses $\Ext(\W)$ times $\W_N$ and then it repeatedly traverses $\lvert\Ext(\W_N)\rvert$ times $\W$.  Observe that $\Ext(\W_0) = 0$ by construction, whereas $\Ext(\W) \leq |\W| \leq n$ and $\lvert\Ext(\W_N)\rvert \leq |\W_N| \leq n$ by definition.  Thus, $2|\W_0| \leq 4n^2 = \Len$.  Then, if $\Copies > \Len^2$, we obtain that $\W_0$ has an internal copy $\T_0$ in $\Copies\Unroll\Syn^k$.  Note that $\T_0$ is a circuit because $\W_0$ is a circuit with $\Ext(\W_0) = 0$ and, thus, $\Jmp(\T_0) = 0$.  Hence, $\T_0$ has $b > 0$ backward noses by Lemma~\ref{lem:acyclic} and, consequently, $\Bal(\T_0) = -b < 0$ by Corollary~\ref{cor:jmp and bal}.  Then, by~\eqref{eq:sep} and the fact that $\Bal(\W_0) = \Bal(\T_0)$ because $\T_0$ is a copy of $\W_0$, we obtain that:
    \begin{align*}
      \lvert\Ext(\W_N)\rvert\Sep_{\Circ, \Len}(\W) &= \lvert\Ext(\W_N)\rvert\Sep_{\Circ, \Len}(\W) + \Ext(\W)\Sep_{\Circ, \Len}(\W_N) =\\
                          &= \Sep_{\Circ, \Len}(\W_0) = \Len\Bal(\W_0) + 2\nose(\W_0) \leq -\Len + 2|\W_0| \leq 0.\tag*{}
		\end{align*}
	\end{discription}
 
\end{proof}

\begin{corollary}\label{cor:cycle ext > 0}
 If\/ $\M$ is an SPCA model and $k \in \ORange{\wrap}$, then $\Syn^k$ has a cycle with\/ $\Ext > 0$.  
\end{corollary}

\begin{proof}
 We refer to the proof of Theorem~\ref{thm:equivalence} ($\ref{thm:crossing} \Rightarrow \ref{thm:multip}$), where $\Len = 4n^2$ and $\Circ$ is defined as the minimum such that $\Sep_{\Circ,\Len}(\W) \leq 0$ for every cycle $\W$ of $\Syn^k$ with $\Ext(\W) < 0$.  Suppose that every cycle of $\Syn^k$ has $\Ext \leq 0$ and let $\W$ be any cycle of $\Syn^k$ and $d = \Circ\Len n^2$.  If $\Ext(\W) < 0$, then $\Sep_{d,\Len}(\W) \leq 0$ because $d \geq \Circ$, while if $\Ext(\W) = 0$, then $\Sep_{d,\Len}(\W) \leq 0$ as in \ref{thm:equivalence:ext0} of Theorem~\ref{thm:equivalence} ($\ref{thm:crossing} \Rightarrow \ref{thm:multip}$).  Then, by Theorem~\ref{thm:no_positive_cycles}, $\M$ is equivalent to a $(d,\Len)$-CA model $\Unit$.  Clearly, some point of $C(\Unit)$ is crossed by no arc of $\Unit$ because $d > n\Len$.  But this is impossible, as it implies that $\Unit$ is a UIG model and so is $\M$ because it is equivalent to $\Unit$.
\end{proof}

By Theorem~\ref{thm:equivalence}, \kMult can be solved by checking that every pair of cycles $\W_N$ and $\W_H$ of $\Syn^k$ with $\Ext(\W_N) < 0 < \Ext(\W_H)$ have a vertex in common.  Theorem~\ref{thm:fast equivalence} yields an efficient method in which only two cycles are traversed (Corollary~\ref{cor:kRep}).  Moreover, Theorem~\ref{thm:fast equivalence} generalizes Theorem~\ref{thm:equivalence} by replacing the restricted notion of cycle twisting bounds %refchange-R12
with the general notion of twisting bounds.  %refchange-R12
Greedy walks play a central role in this theorem, as they do in the characterization by Tucker~\cite{TuckerDM1974} %refchange-R4
(see~\cite{SoulignacJGAA2017}).  A walk $\W = B_0, \ldots, B_j$ of $\Syn^k$ is \emph{greedy hollow} (resp.\ \emph{greedy nose}) when $\Syn^k$ has no hollows (resp.\ noses) from $B_i$ when $B_i \to B_{i+1}$ is a nose (resp.\ hollow), for $i \in \Range{j}$.  By Theorem~\ref{thm:Syn^k to Syn^k+1 remove}, at most two edges of $\Syn^k$ begin at $B_i$, one nose and one hollow.  Thus, in other words, $\W$ is greedy hollow (resp.\ nose) in $\Syn^k$ if hollows (resp.\ noses) are preferred over noses (resp.\ hollows) when choices are possible.  As hollows (resp.\ noses) have $\Jmp < 0$ (resp.\ $\Jmp \geq 0$), greedy hollow (resp.\ nose) cycles usually have $\Ext > 0$ (resp.\ $\Ext \leq 0$).  This could be false, as a greedy hollow cycle is also greedy nose when no choices are possible (e.g.~\cref{fig:example-non-mult-power}).  Lemma~\ref{lem:greedy paths} proves that at least one greedy hollow cycle with $\Ext > 0$ and one greedy nose cycle with $\Ext < 0$ exist.

\begin{lemma}\label{lem:greedy paths}
 If\/ $\M$ is a connected PCA model and $k \in \ORange{\wrap}$, then\/ $\Syn^k$ has a greedy nose cycle with\/ $\Ext < 0$.  Furthermore, if\/ $\M$ is SPCA, then\/ $\Syn^k$ has a greedy hollow cycle with\/ $\Ext > 0$.
\end{lemma}

\begin{proof}
 By Lemma~\ref{lem:cycle ext < 0}, $\Syn^k$ has a cycle $\W_N$ with $\Ext(\W_N) < 0$.  Let $h = \max\{10, n(k+1)\}$.  If $\Copies \geq h^4$, then $\Copies\Unroll\Syn^k$ has an internal copy $\T_N$ of $h^2\Unroll\W_N$ that starts at some copy $A_0$ of a vertex $A$ of $\W_N$.  Note that $\T_N$ has a vertex with $\Row \geq \Row(A_0) + h^2$ because $\Ext(\W_N) < 0$.  Then, by Corollary~\ref{cor:walk drawing}, the co-domain of the function whose graph is $\Gr(\T_N)$ contains $[\Row(A_0), \Row(A_0) + h^2]$ (\cref{fig:greedy}(a)).

 \begin{figure}
 \centering 
  \begin{tabular}{c@{\hspace{3mm}}c}
    \includegraphics{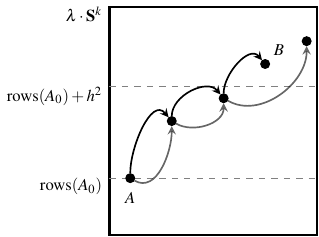} & \includegraphics{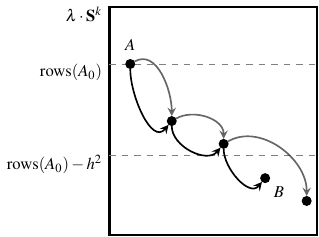} \\
	(a) & (b)
  \end{tabular}

	\caption[]{Schemes for Lemma~\ref{lem:greedy paths}. In (a) the black walk has a copy of a greedy nose cycle, in (b) the black walk has a copy of a greedy hollow cycle.}\label{fig:greedy}
 \end{figure}  

 For $i \geq 0$, let $\G_i$ be the unique greedy nose walk of $\Syn^k$ that starts at $A$ and has $i$ vertices, and $\T_i$ be the copy of $\G_i$ in $\Copies\Unroll\Syn^k$ that starts at $A_0$.  Note that if $\Pos_p(X) \in \Gr_0(\T_i) \cap \Gr_0(\T_N)$ for some vertex $X$ and $p \geq 0$, then the slope of the $p$-arrow leaving $\Pos_p(X)$ in $\Gr_0(\T_i)$ is not smaller than the slope of the $p$-arrow leaving $\Pos_p(X)$ in $\Gr_0(\T_N)$.  Otherwise, by Lemma~\ref{lem:jmp of noses and hollows}, $\T_i$ would take a hollow from $X$, whereas $\T_N$ would take a nose from $X$, contradicting the fact that $\G_i$ is greedy.  Then, as $\Pos_0(A_0) \in \Gr_0(\T_i) \cap \Gr_0(\T_N)$, induction and Theorem~\ref{thm:planarity} imply that $\Gr_0(\T_i)$ is bounded below by $\Gr_0(\T_N)$ for every $i$ such that $\T_i$ is internal (\cref{fig:greedy}(a)).  This means that $\Gr_0(\T_i)$ reaches $\Row(A_0) + h^2$ for some $i$ sufficiently large, thus $\T_i$ has a vertex $B$ with $\Row(B) \geq \Row(A_0) + h^2$ (\cref{fig:greedy}(a)).  Then, as $\Syn^k$ has $n \leq h$ vertices, it follows that $\T_i$ contains a subwalk with $\Jmp > 0$ joining two copies of a same vertex of $\Syn^k$.  The corresponding subpath of $\G_i$ is a greedy nose cycle $\G_N$ that has a copy in $\Copies\Unroll\Syn^k$ with $\Jmp > 0$, i.e., $\Ext(\G_N) < 0$.
 
 Regarding the case in which $\M$ is SPCA, recall that $\Syn^k$ has a cycle $\W_H$ with $\Ext(\W_H) > 0$ by Corollary~\ref{cor:cycle ext > 0}.  Arguments similar to those above, where the role of $\W_N$ is played by $\W_H$, allows us to conclude that $\Syn^k$ has a greedy hollow cycle $\G_H$ with $\Ext(\G_H) > 0$.  We omit the details for the sake of succinctness; see \cref{fig:greedy}(b).
\end{proof}

\begin{theorem}\label{thm:fast equivalence}
 The following statements are equivalent for an SPCA model $\M$ and $k \in \ORange{\wrap}$.%
 \begin{enumerate}
  \item $\M$ is equivalent to a $k$-multiplicative model.\label{thm:fast multip}
  \item Some greedy nose cycle of\/ $\Syn^{k}$ having\/ $\Ext < 0$ shares a vertex with a greedy hollow cycle of\/ $\Syn^k$ having\/ $\Ext > 0$.\label{thm:fast crossing}
  \item Some $\Spiral \in \mathbb{N}$ is a twisting bound %refchange-R12
  of\/ $\Syn^k$.\label{thm:fast twister}
 \end{enumerate}
\end{theorem}

\begin{proof}
 $\ref{thm:fast multip} \Rightarrow \ref{thm:fast crossing}$.  By Lemma~\ref{lem:greedy paths}, $\M$ has a greedy nose cycle $\G_N$ with $\Ext < 0$ and a greedy hollow cycle $\G_H$ with $\Ext > 0$.  By Theorem~\ref{thm:equivalence}, $\G_N$ and $\G_H$ have a vertex in common.

 $\ref{thm:fast crossing} \Rightarrow \ref{thm:fast twister}$.  Let $\Copies \geq \Spiral = h^4$ for $h = \max\{10,(k+1)n\}$, and $A_0$ be a copy of $A$ in $\Copies\Unroll\Syn^k$ with $\Row(A_0) \in (2h^2, 2h^2+h]$.  By hypothesis, a greedy hollow cycle $\G_H$ of $\Syn^k$ having $\Ext(\G_H) > 0$ shares a vertex $A$ with a greedy nose cycle $\G_N$ of $\Syn^k$ having $\Ext(\G_N) < 0$.  Define $\T_N$ as the copy of $2\Ext(\G_H) \Unroll \G_N$ in $\Copies\Unroll\Syn^k$ that starts at $A_0$ and ends at a copy $A_2$ of $A$ (\cref{fig:twisting}).  As usual, Observation~\ref{lem:jmp of noses and hollows} implies that $\T_N$ is an internal walk whose vertices have $\Row \in (h^2, h^3]$.  Similarly, the copy $\T_H$ of $2\lvert\Ext(\G_N)\rvert\Unroll\G_H$ in $\Copies\Unroll\Syn^k$ that starts at $A_2$ is also internal and ends at $A_0$ (\cref{fig:twisting}).  Let $a$ and $b$ be the number of backward noses of $\T_N$ and $\T_H$, respectively.  Clearly, both $\T_N$ and $\T_H$ pass through another copy $A_1$ of $A$ with $\Row(A_0) < \Row(A_1) < \Row(A_2)$ (\cref{fig:twisting}).  Thus, the walk obtained by traversing $\T_N$ from $A_i$ to $A_{i+1}$, $i\in \Range{2}$, and then traversing $\T_H$ from $A_{i+1}$ to $A_i$ is an internal circuit.  By Lemma~\ref{lem:acyclic}, this circuit has at least one backward edge, hence $a + b \geq 2$.  
 
 We claim that $b = 0$, i.e., $\T_H$ has no backward edges.  Contrary to our claim, suppose $X \to Y$ is a backward nose in $\T_H$.  By definition, $X$ is rightmost, thus $\RR(X)$ is leftmost and, by \eqref{eq:row}, $\RR(X) = \Hr(Z) = \Hr \circ \RR \circ \LL(Z)$ for some vertex $Z$.  Then, by Lemma~\ref{lem:nose equivalences}, $\Hl(X) = \LL(Z) \neq \bot$, thus $\Fr \circ \Fl(X) = X$.  In other words, $X \to \Fl(X)$ is a $1$-hollow.  But this is impossible if $X \to Y$ is a backward nose, because $\T_H$ is greedy hollow.  Hence, $b = 0$ and, therefore, $\T_N$ has $a \geq 2$ backward edges.

 \begin{figure}
  \centering 
  \begin{tabular}{c@{\hspace{3mm}}c}
     \includegraphics{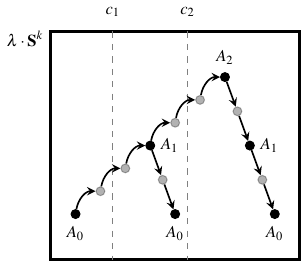} & \includegraphics{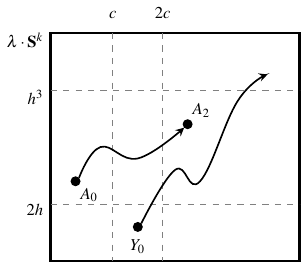}\\
	(a) & (b)
  \end{tabular}

	\caption[]{Scheme for the proof of Theorem~\ref{thm:fast equivalence} ($\ref{thm:fast crossing}%refchange-bugfix
  \Rightarrow \ref{thm:fast twister})$.}\label{fig:twisting}
 \end{figure}

 To prove that $\Spiral$ is a twisting bound %refchange-R12
 of $\Syn^k$, we have to show that any internal copy $\T$ of a walk $\W$ of $\Syn^k$ with $\Jmp(\T) \geq \Spiral$ has a backward edges.  To prove this, we use a copy $\Q$ of $\T$ whose drawing is below $\T_N$.  Let $X_0$ be the last vertex of $\T$ with $\Row(X_0) \leq 3h$.  If $\T$ has no such vertices, then let $X_0$ be a vertex with minimum $\Row$ in $\T$.  By definition, the subwalk $\T'$ of $\T$ from $X_0$ has $\Jmp(\T') \geq \Spiral - 3h \geq h^3$ and visits vertices with $\Row \geq 2h$.  Then, by Lemma~\ref{lem:repetitive}, $X_0$ has a copy $Y_0$ with $\Row(Y_0) \in (h, 2h]$ such that the row containing $Y_0$ is a copy of the row containing $X_0$.  Moreover, if $\Q$ is the copy of $\T'$ from $Y_0$, then the $i$-th vertex of $\Q$ belongs to a row that is a copy of the row containing the $i$-th vertex of $\T'$ that is traversed after $X_0$.  Then, $\T$ has at least as many backward edges as $\Q$.  
 
 By Corollary~\ref{cor:walk drawing}, $\Gr_0(\T_N)$ is a curve that joins $\Pos_0(A_0)$ and $\Pos_{a}(A_2)$.  As in the proof of Lemma~\ref{lem:greedy paths}, Theorem~\ref{thm:planarity} together with the facts that $\T_N$ is greedy and $\Row(Y_0) < \Row(Z)$ for every $Z \in \T_N$, implies that $\Gr_1(\Q)$ is bounded above by $\Gr_0(\T_N)$.  Hence, since $\Q$ ends at a row greater than $h^3$ (because $\Jmp(\Q) \geq h^3$), it follows that $\Gr_1(\Q)$ crosses $2\Cols(\M)$, thus $\Q$ has a backward edge, and so does $\T$ as desired (\cref{fig:twisting}(b)).
 
 $\ref{thm:fast twister} \Rightarrow \ref{thm:fast multip}$ follows by Theorem~\ref{thm:equivalence} because every twisting bound is a cycle twisting bound.%refchange-R12
 \end{proof}

\begin{corollary}\label{cor:kRep}
 The problem $\kMult$ can be solved in $O(n)$ time for every PCA model\/ $\M$ and every $k \in \ORange{\wrap}$.  If the output is no, then a negative certificate that can be authenticated in $O(n)$ time is obtained as a by-product.
\end{corollary}

\begin{proof}
 If $\M$ is PIG, then the algorithm outputs yes because $\M$ is $\infty$-multiplicative~\cite{CorneilKimNatarajanOlariuSpragueIPL1995,GavoillePaulSJDM2008} (see Theorem~\ref{thm:tucker} for an alternative proof).  Otherwise, $\Syn^k$ is built in $O(n)$ time with Theorem~\ref{thm:Syn^k algorithm}.  Then, a subgraph $\Syn_H$ (resp.\ $\Syn_N$) of $\Syn^k$ is obtained in $O(n)$ time by removing each nose (resp.\ hollow) $A \to B$ when a hollow (resp.\ nose) $A \to X$ exists.  By construction, the walks of $\Syn_H$ (resp.\ $\Syn_N$) are precisely the greedy hollow (resp.\ nose) walks of $\Syn^k$.  By Lemma~\ref{lem:nose equivalences} and Theorem~\ref{thm:Syn^k to Syn^k+1 remove}, at most one hollow (resp.\ nose) begins at each vertex $A$, thus the family of greedy hollow (resp.\ nose) cycles is obtained from $\Syn_H$ (resp.\ $\Syn_N$) in $O(n)$ time.  By Lemma~\ref{lem:greedy paths}, at least one of these greedy cycles has $\Ext > 0$ (resp.\ $\Ext < 0$).  Let $\G_N$ and $\G_H$ be greedy cycles with $\Ext(\G_N) < 0 < \Ext(\G_H)$ that are computed in $O(n)$ time.  The algorithm outputs yes if and only if $\G_N$ and $\G_H$ have a vertex in common, a fact that can be checked in $O(n)$ time.  The algorithm is correct by Theorems \ref{thm:equivalence}~and~\ref{thm:fast equivalence}.  When the output is no, the pair of cycles $(\G_H, \G_N)$ is returned.  To authenticate this certificate, $\Syn^k$ is generated in $O(n)$ time to verify that $\G_H$ and $\G_N$ are cycles that have no vertices in common.
\end{proof}

If the model $\M$ is PIG, then $\M$ is $\infty$-multiplicative.  If $G(\M)$ is co-bipartite, then $\wrap = 2$ and $\M$ is $1$-multiplicative~\cite{TuckerDM1974}.  Finally, if $\M$ is SPCA and $G(\M)$ is not co-bipartite, then $\M$ is the unique PCA model representing $G(\M)$, up to equivalence and full reversal~\cite{HuangJCTSB1995}.  In this last case, the pair of cycles $(\G_H,\G_N)$ with $\Ext(\G_H)\Ext(\G_N) < 0$ and no vertices in common defines a submodel $\F = (C(\M), \{A \in \A(\M) \mid A \in \G_H \cup \G_N\})$ of $\M$ that is equivalent to no $k$-multiplicative UCA model.   Therefore, Theorem~\ref{thm:equivalence} implies a characterization by forbidden induced subgraphs for the class of PCA graphs that admit $k$-multiplicative UCA models.  For $k=1$, this is the characterization in \cite{TuckerDM1974}.%refchange-R4

\section{A certifying algorithm for the multiplicative problem}
\label{sec:representation}

Suppose \kMult answers yes for an SPCA model $\M$, and let $\Unit$ be a $k$-multiplicative $(\Circ,2n^2+1)$-CA model equivalent to $\M$, where $\Circ$ is the minimum such that $\Sep_{\Circ,\Len}(\W) \leq 0$ for every cycle $\W$ with $\Ext(\W) < 0$.  The existence of $\Unit$ follows by Theorem~\ref{thm:equivalence}.  Since $|\Ext(\W)| \leq n$ for every cycle $\W$ of $\Syn^k$, Theorem~\ref{thm:no_positive_cycles} and~\eqref{eq:sep} imply that $\Circ$ is polynomial in $n$, thus we can compute $\Unit$ in polynomial time.  In this section we design a certifying algorithm for \kMult that runs in $O(n)$ time.  Although the algorithm is a simple generalization of one by Soulignac~\cite{SoulignacJGAA2017,SoulignacJGAA2017a} %refchange-R4
for \Rep $=$ \nMult{1}, its correctness follows by simpler and shorter arguments that exploit the loop unrolling technique.  We remark that the algorithm works for every PCA model, thus we do not assume that $\M$ is SPCA beyond this point.

Let $\W$ be a cycle of $\Syn^k$. Recall that $\Bal(\W)$ is a total that increases with every nose and decreases with every hollow in $\W$. In a way, $\Bal$ measures how many edges went from some arc to other arc in the clockwise direction in the model (noses), and compares it with the amount of edges that did this movement counter-clockwise (hollows). On other hand, $\Ext$ measures something similar, but this time in terms of the amount of turns of the circle were taken before returning to the initial arc. We will pay attention to the ratio of the amount of edges with respect to the amount of turns. However, notice that clockwise movement implies a positive $\Bal$ but a negative $\Ext$, and vice-versa. Therefore, we define %refchange-R14
$\Ratio(\W) = -\Bal(\W)(\Ext(\W))^{-1}$.  By Theorem~\ref{thm:no_positive_cycles} and~\eqref{eq:sep}, if $\M$ is equivalent to a $k$-multiplicative $(\Circ,\Len+1)$-CA model, then either $\Ext(\W) < 0$ and $\Circ \geq \Len\Ratio(\W) + 2\nose(\W)$ or $\Ext(\W) > 0$ and $\Circ \leq \Len\Ratio(\W)$.  Then, as $\W$ has at least one nose when $\Ext(\W) < 0$, we obtain that $\Circ = \Len\Ratio^k(\M) + \Extra \leq \Len\RATIO^k(\M)$ for some $\Extra > 0$, where
\begin{align*}
 \Ratio^k(\M) &= \max\{\Ratio(\W) \mid \W \text{ is a cycle of $\Syn^k$ with } \Ext(\W) < 0\}\text{, and} \\
 \RATIO^k(\M) &= \min\{\Ratio(\W) \mid \W \text{ is a cycle of $\Syn^k$ with } \Ext(\W) > 0\}.
\end{align*}
We omit $\M$ as usual; note that $\RATIO^k = \infty$ when every cycle of $\Syn^k$ has $\Ext \leq 0$.

The fact that $\Ratio^1 < \RATIO^1$ is a restatement of a well-known result in \cite{TuckerDM1974}.  Specifically, Tucker %refchange-R4
proved that a PCA model $\M$ is equivalent to a $1$-multiplicative model if and only if $a/b > x/y$ for %refchange-R15
every $(a,b)$-independent and $(x,y)$-circuit.  %refchange-R4
We shall not define what an $(a,b)$-independent is or what an $(x,y)$-circuit is.  Instead, we remark that, as discussed in \cite[Theorem~4]{SoulignacJGAA2017}, %refchange-R4
each $(a,b)$-independent corresponds to a circuit $\W_N$ of $\Syn^1$ with $\Ext(\W_N) < 0$ and, similarly, each $(x,y)$-circuit corresponds to a circuit $\W_H$ of $\Syn^1$ with $\Ext(\W_H) > 0$.  Moreover, $a/b = \Ratio(\W_N) + h$ and $x/y = \Ratio(\W_H) + h$ for some constant $h$ \cite[Theorem~4]{SoulignacJGAA2017}.  Therefore, Tucker's characterization not only implies that $\Ratio^1 < \RATIO^1$ when $\M$ is equivalent to a $1$-multiplicative model, it also implies the converse.  Soulignac discusses alternative characterizations of $1$-multiplicative models that are described in terms of $\Ratio^1$ and the parameters $\Length^1$ and $\Lex^1$ that we define next for every $k \geq 1$\cite[Theorem~2]{SoulignacJGAA2017}.  %refchange-R4
In few words, $\Length^k$ and $\Lex^k$ are special weightings of $\Syn^k$ that can be used to discard some edges of $\Syn^k$ that are implied when some specific values of $\Circ$ and $\Len$ are used.  As an acyclic digraph is obtained after discarding these ``redundant'' edges, the canonical solution to the full system $\F_{\Circ,\Len}^k$ can be computed more efficiently.  Theorem~\ref{thm:tucker} below is the generalization of Soulignac's characterization for $k \geq 1$ and is the theoretical foundation for the algorithm that we develop in this section.

By definition, $\Bal$ is a weighting of $\Syn^k$ where, for every edge $A \to B$, $\Bal(A \to B) = -1$ if $A \to B$ is a $1$-hollow and $\Bal(A \to B) = i$ if $A \to B$ is an $i$-nose.  Similarly, $\Ext$ is a weighting of $\Syn^k$ where $\Ext(A \to B) = \Boolean{B \geq A}$ if $A \to B$ is a hollow and $\Ext(A \to B) = -\Boolean{A > B}$ if $A \to B$ is a nose.  Let $\Length^k$ and $\Lex^k$ be the weightings of $\Syn^k$ such that
\begin{align*}
 \Length^k(A \to B) =& \Bal(A \to B) + \Ratio^k\Ext(A \to B)\text{, and}\\
 \Lex^k(A \to B) =& (\Length^k(A \to B), \Ext(A \to B)).
\end{align*}
By \eqref{eq:sep}, if $\W$ is a walk of $\Syn^k$ and $\Circ = \Len\Ratio^k + \Extra$, then 
\begin{align}
 \Sep_{\Circ,\Len}(\W) &= \Len\Bal(\W) + \Circ\Ext(\W) + 2\nose(\W) \notag \\
 &= \Len\Bal(\W) + \Len\Ratio^k\Ext(\W) + \Extra\Ext(\W) + 2\nose(\W) \notag \\
 &= \Len\Length^k(\W) + \Extra\Ext(\W) + 2\nose(\W) = (\Len, \Extra) \cdot \Lex^k(\W) + 2\nose(\W). \label{eq:sep bis}
\end{align}

Let $A_0$ be the initial arc of $\M$.  Say that an edge $A \to B$ of $\Syn^k$ is \emph{redundant} when 
\begin{displaymath}
 \Dist{\Lex^k}(A_0, B) > \Dist{\Lex^k}(A_0, A) + \Lex^k(A \to B),
\end{displaymath}
where $>$ denotes the lexicographically greater relation.  Let $\Red^k(\M)$ be the spanning subgraph of $\Syn^k$ obtained by removing all the redundant edges; as usual, we omit $\M$ from $\Red^k$.  Our final characterization yields an alternative algorithm that provides a $k$-multiplicative model equivalent to $\M$ at the cost of having longer arcs.

\begin{theorem}\label{thm:tucker}
 Let $A_0$ be the initial arc of a connected PCA model\/ $\M$, $k \in \ORange{\wrap}$, $\Circ = \Len\Ratio^k+\Extra$, $\Len = \Extra^3$, and $\Extra = 4n$. The following statements are equivalent:
 \begin{enumerate}
   \item $\M$ is equivalent to a $k$-multiplicative UCA model.\label{thm:tucker:equivalence}
   \item $\Ratio^k < \RATIO^k$.\label{thm:tucker:bound}
   \item $\Lex^k(\W) < (0,0)$ for every cycle\/ $\W$ of\/ $\Syn^k$.\label{thm:tucker:length}
   \item $\Red^k$ is acyclic.\label{thm:tucker:negative-witness}
   \item $\Dist{\Sep_{\Circ, \Len}}(\Syn^k,A_0, A) = \IDist{\Sep_{\Circ, \Len}}(\Red^k, A_0, A)$ for every $A \in \A(\M)$.\label{thm:tucker:positive-witness}
 \end{enumerate}
\end{theorem}

\begin{proof}
 $\ref{thm:tucker:equivalence} \Rightarrow \ref{thm:tucker:bound}$ follows by Theorem~\ref{thm:no_positive_cycles} and \eqref{eq:sep}; see above.
 
 $\ref{thm:tucker:bound} \Rightarrow \ref{thm:tucker:length}$.  If $\Ext(\W) = 0$ and $\M$ is PIG, then $\W$ is internal because $\Syn^k$ has no external hollows.  Hence, by Lemma~\ref{lem:acyclic} and Corollary~\ref{cor:jmp and bal}, $\Length^k(\W) = \Bal(\W) < 0$.  Similarly, if $\Ext(\W) = 0$ and $\M$ is not PIG, then $\W$ has an internal copy $\T$ in $|\W|\Unroll\Syn^k$ that is a circuit and has $\Bal(\T) = \Bal(\W)$, thus $\Length^k(\W) = \Bal(\W) = \Bal(\T) < 0$ by Lemma~\ref{lem:acyclic} and Corollary~\ref{cor:jmp and bal}.  If $\Ext(\W) > 0$, then $-\Bal(\W)(\Ext(\W))^{-1} = \Ratio(\W) \geq \RATIO^k > \Ratio^k$, thus $\Length^k(\W) < 0$.  Finally, if $\Ext(\W) < 0$, then $-\Bal(\W)(\Ext(\W))^{-1} = \Ratio(\W) \leq \Ratio^k$, thus $\Length^k(\W) \leq 0$. 
 
 $\ref{thm:tucker:length} \Rightarrow \ref{thm:tucker:negative-witness}$.  If $\Red^k$ has some circuit $\W = B_0, \ldots, B_j$ ($B_0 = B_j$), then $B_i \to B_{i+1}$ is not redundant for $i \in \Range{j}$ and, consequently, 
  \begin{align*}
    \Dist{\Lex^k}(A_0, B_j) &\leq \Dist{\Lex^k}(A_0, B_{j-1}) + \Lex^k(B_{j-1} \to B_j) \leq \ldots \\
                            &\leq \Dist{\Lex^k}(A_0, B_0) + \Lex^k(\W) = \Dist{\Lex^k}(A_0, B_j) + \Lex^k(\W).
  \end{align*}
 
 $\ref{thm:tucker:negative-witness} \Rightarrow \ref{thm:tucker:positive-witness}$.  Note that $\Dist{\Sep_{\Circ,\Len}}(\Syn^k, A_0, A) \geq \IDist{\Sep_{\Circ,\Len}}(\Red^k,A_0, A)$ for $A \in \A(\M)$ because every path of $\Red^k$ is a walk of $\Syn^k$.  For the other inequality, it suffices to prove that $\Sep(\W) \leq \IDist{\Sep_{\Circ,\Len}}(\Red^k, A_0, A)$ for every walk $\W = A_0, \ldots, A_j$ with $A = A_j$ and $j \leq n$.  We prove this fact by induction on $j$.  The base case $j = 0$ is trivial.  In the inductive step $j > 0$, let
  \begin{itemize}
    \item for $i \in \CRange{j}$, $\W_i$ be a walk of $\Red^k$ from $A_0$ to $A_i$ with $\Lex^k(\W_i) = \IDist{\Lex^k}(\Red^k, A_0, A_i)$, and
    \item $\W_{\Syn}$ be the walk obtained by traversing $A_{j-1} \to A_j$ after $\W_{j-1}$.
  \end{itemize}
  By the inductive hypothesis, $\Sep(\W) \leq \Sep(\W_{j-1}) + \Sep(A_{j-1} \to A_{j}) = \Sep(\W_{\Syn})$, therefore $\Sep(\W) \leq \Sep(\W_j)$ when $A_{j-1} \to A_j$ is an edge of $\Red^k$.  Suppose, then, that $A_{j-1} \to A_{j}$ is redundant in $\Syn^k$.  In this case, taking into account that no edge of $\W_j$ is redundant, it follows by induction that $\Lex^k(\W_j) = \Dist{\Lex^k}(\Syn^k, A_0, A_j)$.  Consequently, there are only two possibilities for $\Lex^k(\W_j)$ and $\Lex^k(\W_\Syn)$ because
  \begin{displaymath}
    \Lex^k(\W_j) = \Dist{\Lex^k}(\Syn^k, A_0, A_j) > \Dist{\Lex^k}(\Syn^k, A_0, A_{j-1}) + \Lex^k(A_{j-1} \to A_{j}) = \Lex^k(\W_{\Syn}).
  \end{displaymath}

  \begin{discription}[label={\textbf{Case~$\mathbf{\arabic*}$:}}]
   \item $\Length^k(\W_j) > \Length^k(\W_\Syn)$.  If $\W$ is a cycle with $\Ext(\W) < 0$ and $\Ratio(\W) = \Ratio^k$, then \begin{displaymath}|\Ext(\W)|\Ratio^k = -|\Ext(\W)|\Bal(\W)(\Ext(\W))^{-1} = \Bal(\W)\end{displaymath} is integer.  Therefore, 
   \begin{align*}
   |\Ext(\W)|(\Length^k(\W_j) - \Length^k(\W_\Syn)) = & |\Ext(\W)|(\Bal(\W_j) - \Bal(\W_\Syn)) \\ 
   &+ \Bal(\W)(\Ext(\W_j) - \Ext(\W_\Syn))
   \end{align*}
   is also integer.  Then, as $|\Ext(\W)| \leq |\W| \leq n$, we gat that $\Length^k(\W_j) - \Length^k(\W_\Syn) \geq n^{-1}$. On the other hand, by definition, $\nose(\W_{\Syn}) \leq n$, $\Ext(\W_j) \geq -|\W_j| \geq -n$ and $\Ext(\W_\Syn) \leq |\W_\Syn| \leq n$.  Then, by \eqref{eq:sep bis},
  \begin{align*}
     \Sep_{\Circ,\Len}(\W_j) - &\Sep_{\Circ,\Len}(\W_{\Syn}) = \Len(\Length^k(\W_j) - \Length^k(\W_\Syn)) + \Extra(\Ext(\W_j) - \Ext(\W_\Syn)) \\
     & + 2(\nose(\W_j) - \nose(\W_\Syn)) \geq (4n)^3n^{-1} + 4n(-2n) - 2n > 0.
  \end{align*}

  \item $\Length^k(\W_j) = \Length^k(\W_\Syn)$ and $\Ext(\W_j) > \Ext(\W_\Syn)$.  In this case, \eqref{eq:sep bis} implies that
  \begin{align*}
     \Sep_{\Circ,\Len}(\W_j) - \Sep_{\Circ,\Len}(\W_{\Syn})  &= \Len(\Length^k(\W_j) - \Length^k(\W_\Syn)) + \Extra(\Ext(\W_j) - \Ext(\W_\Syn)) \\
     &\hphantom{= }+ 2(\nose(\W_j) - \nose(\W_\Syn)) \geq  4n - 2n > 0.
  \end{align*}
 \end{discription}
 Summing up, $\IDist{\Sep_{\Circ,\Len}}(\Red^k, A_0, A_j) = \Sep_{\Circ,\Len}(\W_j) > \Sep_{\Circ,\Len}(\W_{\Syn}) \geq \Sep_{\Circ,\Len}(\W)$ in the case when $A_{j-1} \to A_j$ is redundant.
  
  $\ref{thm:tucker:positive-witness} \Rightarrow \ref{thm:tucker:equivalence}$.  If $\M$ is equivalent to no $k$-multiplicative $(\Circ,\Len+1)$-CA model, then $\Syn^k$ has a cycle $\W$ with $\Sep_{\Circ,\Len}(\W) > 0$ by Theorem~\ref{thm:no_positive_cycles}.  Then, $\infty = \Dist{\Sep_{\Circ,\Len}}(\Syn^k, A_0, A) > \IDist{\Sep_{\Circ,\Len}}(\Red^k, A_0, A)$ for every $A \in \W$.
\end{proof}

Theorem~\ref{thm:tucker} yields the following algorithm to compute a $k$-multiplicative model equivalent to an input PCA model $\M$ when \kMult answers yes; $A_0$ is the initial arc of $\M$:
\begin{enumerate}
 \item Insert an arc intersecting $\LL(A)$ and $A$ for every $A \in \A(\M) \setminus \{A_0\}$ such that $\LL(A) \cap A = \emptyset$. (After this step, $\M$ is a connected model.)\label{alg:insert}
 \item Compute $\Ratio^k$ to obtain the weighting $\Lex^k$ of $\Syn^k$.\label{alg:ratio}
 \item Determine $\Dist{\Lex^k}(\Syn^k, A_0, A)$ for every $A \in \A(\M)$.\label{alg:dlex}
 \item Obtain $\Red^k$ by removing each redundant edge $A \to B$ of $\Syn^k$.\label{alg:red}
 \item Set $s(A) = \Dist{\Sep_{\Circ,\Len}(\Red^k, A_0, A)}$ for every $A \in \A(\M)$, where $\Len = (4n)^3$ and $\Circ = \Len\Ratio^k + 4n$.\label{alg:sep}
 \item Remove all the arcs inserted at Step~\ref{alg:insert}.
 \item Output $(C, \{s(A), s(A) + \Len \bmod \Circ \mid A \in \A(\M)\})$ for a circle $C$ with $|C| = \Circ$.\label{alg:model}
\end{enumerate}
Steps \ref{alg:insert} and \ref{alg:red}--\ref{alg:model} can be easily implemented in $O(n)$ time; just recall that $\Red^k$ is acyclic by Theorem~\ref{thm:tucker}.  In the following sections we discuss how to implement Steps \ref{alg:ratio}~and~\ref{alg:dlex}.  

\subsection{Step~\ref{alg:ratio}: computation of the ratios}

To efficiently compute $\Ratio^k$, the key is to observe that every greedy nose cycle $\W$ of $\Syn^k$ with $\Ext(\W) < 0$ has $\Ratio(\W) = \Ratio^k$.  A weaker form of this result, stating that at least one greedy nose cycle $\W$ has $\Ratio(\W) = \Ratio^k$, is already known for $k = 1$ \cite[Lemma~2]{SoulignacJGAA2017}.

\begin{lemma}\label{lem:ratio greedy}
 If\/ $\M$ is a connected PCA model and $k \in \ORange{\wrap}$, then\/ $\Ratio^k = \Ratio(\G_N)$ for every greedy nose cycle\/ $\G_N$  of\/ $\Syn^k$ with\/ $\Ext(\G_N) < 0$.
\end{lemma}

\begin{proof}
 Let $\G_N$ be a greedy nose cycle of $\Syn^k$ with $\Ext(\G_N) < 0$ and $\W$ be a cycle of $\Syn^k$ with $\Ext(\W) < 0$.  The existence of $\G_N$ follows by Lemma~\ref{lem:greedy paths}.  We shall prove that $\Ratio(\G_N) \geq \Ratio(\W)$ to obtain that $\Ratio^k = \Ratio(\G_N)$.
  
 Suppose first that $\M$ is PIG and let $A_0$ be the initial arc of $\M$.  By hypothesis, $\G_N$ and $\W$ both contain the unique external nose $\LL(A_0) \to A_0$ of $\Syn^k$, thus $\Ext(\G_N) = \Ext(\W) = -1$ and $\Jmp(\G') = \Jmp(\W) = \Rows(\M)-1$.  Moreover, the subpaths $\G'$ of $\G_N$ and $\W'$ of $\W$ from $A$ to $\LL(A)$ are internal.  Then, as $\G'$ is greedy, Theorem~\ref{thm:planarity} implies that $\Gr_0(\G')$ is bounded below by $\Gr_0(\W')$, thus the number of backward edges of $\G'$ is not greater than the number of backward edges of $\W'$.  Consequently, $\Bal(\G') \geq \Bal(\W')$ by Corollary~\ref{cor:jmp and bal}, thus $\Bal(\G_N) \geq \Bal(\W)$ and, therefore, $\Ratio(\G_N) \geq \Ratio(\W)$.
 
 Suppose now that $\M$ is not PIG.   Fix a vertex $A$ of $\G_N$ and let $A_0 < \ldots < A_{\Copies}$ be the copies of $A$ in $\Copies\Unroll\Syn^k$ for $\Copies \gg n^6$.  Similarly, let $B_0 < \ldots < B_{\Copies}$ be the copies in $\Copies\Unroll\Syn^k$ of a vertex $B$ of $\W$.  Let $w = n^2$, $i = \Ext(\G_N)\Ext(\W)$, and $z \in \CRange{n}- \{0\}$ be such that row $w+zj$ is a copy of row $w$ for every $j \geq 0$ with $w+zj < \Rows(\lambda\Unroll\Syn^k)$.  The existence of $z$ follows by Lemma~\ref{lem:repetitive}.  Note that $i \leq n^2$ because every edge has $\Ext \in [-1,1]$.  Moreover, for $j \in \CRange{n}$, the copy $\T_N(j)$ of $-z\Ext(\W) \Unroll \G_N$ that starts at $A_{w+zij}$ in $\Copies\Unroll\Syn^k$ is internal and ends at $A_{w+zi(j+1)}$. Similarly, the copy $\T(j)$ of $-z\Ext(\G_N)\Unroll\W$ that starts at $B_{w+zij}$ in $\Copies\Unroll\Syn^k$ is internal and ends at $B_{w+zi(j+1)}$.  By definition, the rows of $\Copies\Unroll\Syn^k$ between $A_{w+zij}$ and $B_{w+zij}$ are copies of the rows between $A_{w}$ and $B_{w}$ and, consequently, $\Jmp(\T_N(j)) = \Jmp(\T(j))$.
  
 For $j \in \CRange{n}$, let $a(j)$ be the number of backward edges in $\T_N(j)$ and $b(j)$ be the number of backward noses in $\T(j)$.  Clearly, $\Q_N = \T_N(1) + \ldots + \T_N(n-1)$ is greedy and internal because $\G_N$ is greedy and $\T_N(j)$ is internal for $j \in \CRange{n}$.  Moreover, $\Q_N$ has $x = \sum_{j=1}^{n-1} a(j)$ backward edges.  Similarly, $\Q = \T(0) + \ldots + \T(n)$ is internal and has $y = \sum_{j=0}^n b(j)$ backward edges.  By Corollary~\ref{cor:walk drawing}, $\Gr_0(\Q_N)$ is a continuous curve from $\Pos_0(A_{w+z})$ to $\Pos_{x}(A_{w+zin})$, whereas $\Gr_1(\Q)$ is a continuous curve from $\Pos_1(B_{w})$ to $\Pos_{y+1}(B_{wzi(n+1)})$.  Since $\Q_N$ is greedy, Theorem~\ref{thm:planarity} implies that $\Gr_0(\Q_N)$ is bounded below by $\Gr_1(\Q)$.  Consequently, $x \leq y + 2$.  Then, since $a(j)$ and $b(j)$ are integer, there exists $j \in \Range{n} - \{0\}$ such that $a(j) \leq b(j)$. By definition, $\T_N(j)$ has exactly $-z\Ext(\W)$ copies of each edge in $\G_N$, thus $\Bal(\T_N(j)) = -z\Ext(\W)\Bal(\G_N)$.  Similarly, $\Bal(\T(j)) = -z\Ext(\G_N)\Bal(\W)$. Then, by Corollary~\ref{cor:jmp and bal}:
 \begin{align*}
  zi(\Ratio(\G_N) - \Ratio(\W)) &= -zi\Ext(\G_N)^{-1}\Bal(\G_N) + zi\Ext(\W)^{-1}\Bal(\W) \\
  &= -z\Ext(\W)\Bal(\G_N) + z\Ext(\G_N)\Bal(\W) \\
  &= \Bal(\T_N(j)) - \Bal(\T(j)) \\
  &= \Jmp(\T_N(j)) - a(j) - \Jmp(\T(j)) + b(j) \geq 0.
 \end{align*}
 As $zi > 0$, we conclude that $\Ratio(\G_N) \geq \Ratio(\W)$ and, therefore, $\Ratio(\G_N) = \Ratio^k$.
\end{proof}

The analogous of Lemma~\ref{lem:ratio greedy} for $\RATIO^k$ is stated below without proof.  When $\M$ is SPCA, $\Ratio^k(\M)$ and $\RATIO^k(\M)$ can be obtained in $O(n)$ time by considering a greedy nose cycle $\G_N$ with $\Ext(\G_N) < 0$ and a greedy hollow cycle $\G_H$ with $\Ext(\G_H) > 0$.  The cycles $\G_N$ and $\G_H$ exist by Lemma~\ref{lem:greedy paths} and can be computed in $O(n)$ time as in Corollary~\ref{cor:kRep}.  This yields another algorithm for \kMult that is just a restatement of the one discussed in Corollary~\ref{cor:kRep}: instead of looking for the intersection of $\G_N$ and $\G_H$, compare their ratios.  This algorithm is a simplification of the one designed by Kaplan and Nussbaum~\cite{KaplanNussbaumDAM2009} %refchange-R4
in which all the greedy cycles are traversed.  

\begin{lemma}
 If\/ $\M$ is an SPCA model and $k \in \ORange{\wrap}$, then\/ $\RATIO^k = \Ratio(\G_H)$ for every greedy hollow cycle\/ $\G_H$ of\/ $\Syn^k$ with\/ $\Ext(\G_H) > 0$.
\end{lemma}

\begin{corollary}\label{cor:ratio computation}
 Given a connected PCA model\/ $\M$ and $k \in \ORange{\wrap}$, it takes $O(n)$ time to compute\/ $\Ratio^k$ and\/ $\RATIO^k$.
\end{corollary}

\subsection{\texorpdfstring{Step~\ref{alg:dlex}: determining the distances according to $\Lex$}{Step \ref{alg:dlex}: determining the distances according to Lex}}

Let $A_0$ be the initial arc of $\M$.  The key to efficiently determine $\Dist{\Lex^k}(A_0, A)$ is to observe that some path $\W$ of $\Syn^k$ from $A_0$ to $A$ with $\Lex^k(\W) = \Dist{\Lex^k}(A_0, A)$ is ``dually greedy''; we remark that a restricted version of this fact is already known for $k=1$ \cite[Lemma~4]{SoulignacJGAA2017}.   A walk $\W = B_0, \ldots, B_j$ of $\Syn^k$ is \emph{greedy anti-hollow} (resp.\ \emph{anti-nose}) when $\Syn^k$ has no hollows (resp.\ noses) reaching $B_{i+1}$ when $B_i \to B_{i+1}$ is a nose (resp.\ hollow), for $i \in \Range{j}$.  In other words, $\W$ is greedy anti-hollow (resp.\ anti-nose) when hollows (resp.\ noses) are preferred to noses (resp.\ hollows) in a backward traversal of $\W$.  The walk $\W$ is a \emph{dually greedy hollow} (resp.\ \emph{nose}) when there exists $i \in \CRange{j}$ such that:
\begin{itemize}
 \item $B_0, \ldots, B_i$ is a greedy hollow (resp.\ nose) walk of $\Syn^k$, and
 \item $B_i, \ldots, B_j$ is a greedy anti-nose (resp.\ anti-hollow) walk of $\Syn^k$.
\end{itemize}

\begin{lemma}\label{lem:semi-greedy}
 Let\/ $\M$ be a connected PCA model that is equivalent to a $k$-multiplicative model for $k \in \ORange{\wrap}$, and $A_0$ be the initial arc of\/ $\M$.  If\/ $\W$ is a path of\/ $\Syn^k$ from $A_0$ to a vertex $A$, then\/ $\Syn^k$ has a dually greedy nose (resp.\ hollow) walk\/ $\W_D$ from $A_0$ to $A$ with\/ $\Lex^k(\W_D) = \Lex^k(\W)$.
\end{lemma}

\begin{proof}
 We only prove the existence of the dually greedy nose walk as the existence of the dually greedy hollow walk is analogous.  Suppose first that $\M$ is not PIG.  By Theorem~\ref{thm:fast equivalence}, $\Syn^k$ has a twisting bound $\Spiral$.  %refchange-R12
 Let $\Copies \gg h^9$ for $h = \max\{10+\Spiral, n(k+1)\}$, and consider a copy $B$ of $A_0$ in $\Copies\Unroll\Syn^k$ with $\Row(B) \in (h^4, h^4+h]$.  Let $\T$ be the copy of $\W$ in $\Copies\Unroll\Syn^k$ that starts at $B$ and ends at a copy $X$ of $A$, $\T_N$ be the greedy nose path of $\Copies\Unroll\Syn^k$ with $|\T_N| = 3h^2$ that starts at $B$ and ends at a vertex $Y$, and $\T_H$ be the greedy anti-hollow walk of $\Copies\Unroll\Syn^k$ with $|\T_H| = 7h^4$ that starts at a vertex $Z$ and ends at $X$.  Moreover, let $x$, $y$, and $z$ be the number of backward edges in $\T$, $\T_N$, and $\T_H$, respectively.
 
 \begin{figure}
  \centering 
  \begin{tabular}{c@{\hspace{3mm}}c}
    \includegraphics{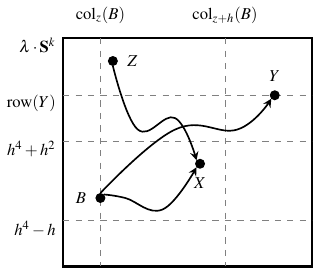} & \includegraphics{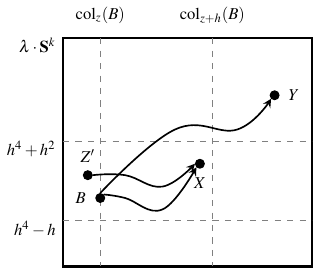} \\ 
    (a) & (b) 
  \end{tabular}

	\caption[]{Scheme for the proof of Lemma~\ref{lem:semi-greedy}}\label{fig:dually}%refchange-R16
 \end{figure}  

 By Lemma~\ref{lem:greedy paths}, $\Syn^k$ has a greedy nose cycle $\G$ with $\Ext(\G) < 0$.  If some vertex of $\G$ has a copy in $\T_N$, then a copy $\T_G$ of $\G$ is included in $\T_N$ because $\T_N$ is greedy.  Otherwise, $\T_N$ contains the copy $\T_G$ of a greedy nose cycle disjoint from $\G$ that also has $\Ext < 0$ by Theorem~\ref{thm:equivalence}.  Thus, whichever the case, $\Jmp(\T_N) > 0$. Moreover, $\T_N$ visits at least $2h$ copies of $\T_G$ because $|\T_G| \leq h$ and $|\T_N| = 3h^2$.  Hence, $y > h$ because $h \geq \Spiral$ and $\Spiral$ is a twisting bound.  %refchange-R12
 As usual, Lemma~\ref{lem:jmp of noses and hollows} implies that $\T$ and $\T_N$ are internal, $\T$ visits vertices with $\Row \in [h^4-h,h^4+2h^2]$, and $\T_N$ visits vertices with $\Row \in [h^4-2h^2, h^4+4h^3]$.  Moreover, by Theorem~\ref{thm:planarity}, $\Gr_z(\T)$ is bounded above by $\Gr_z(\T_N)$ because both start at $\Pos_{z}(B)$ and $\T_N$ is greedy.  
 
 Similarly, $\T_H$ is internal and visits vertices with $\Row \in [h^3, h^5]$ by Lemma~\ref{lem:jmp of noses and hollows}, thus $\Gr_{z}(\T)$ is bounded above by $\Gr_x(\T_H)$ because both end at $\Pos_{z+x}(X)$ and $\T_H$ is greedy anti-hollow.  Consider the following alternatives to prove that $\Gr_z(\T_N)$ and $\Gr_x(\T_H)$ share a point $p$.
 \begin{discription}[label={\textbf{Case~\arabic*: }}]    
   \item $\Jmp(\T_H) < 0$.   By Corollary~\ref{cor:jmp and bal}, $|\Jmp(\T_H)| \geq 7h^3$ because $\T_H$ traverses at least $7h^3$ copies of $Z$, thus $\Row(X) \leq \Row(Y) \leq \Row(Z)$.  Then, as $x \leq |\T| \leq h < y$ it follows that $\Gr_z(\T_N)$ and $\Gr_x(\T_H)$ have a point in common (\cref{fig:dually}(a)).
    
   \item $\Jmp(\T_H) \geq 0$.  Let $q \in [z, x+z]$.  Recall that $\Gr_{z}(\T)$ is bounded above by $\Gr_x(\T_H)$.  Then,  Lemma~\ref{lem:jmp of noses and hollows} implies that $\Row(Z') \in [h^4-h, h^4+4h^2]$ for every $Z'$ such that $\Gr_x(\T_H)$ traverses $\Pos_{q}(Z')$.  As there are less than $6h^2$ such possible $Z'$ for each $q$ and $x \leq h$, it follows that $\Gr_x(\T_H$) passes through $\Pos_{z-1}(Z')$ for some vertex $Z'$, i.e., $x < z$.  Then, $\Gr_z(\T_N)$ and $\Gr_x(\T_H)$ share some point (\cref{fig:dually}(b)).
 \end{discription}
    
 Summing up, $\Copies\Unroll\Syn^k$ has a walk $\T_D$ such that $\Gr_z(\T_D)$ starts at $\Pos_z(B_0)$, takes the arrows of $\Gr_z(\T_N)$ until reaching $p$, and then it takes the arrows of $\Gr_x(\T_H)$ until reaching $\Pos_{x+z}(B)$.  The path $\T_D$ is dually greedy by construction, and so is the walk $\W_D$ of $\Syn^k$ whose copy is $\T_D$.  Moreover, $\T_D$ has $x$ backward edges and $\Jmp(\T_D) = \Jmp(\T)$.  Then, by Corollary~\ref{cor:jmp and bal}, $\Bal(\W_D) = \Bal(\T_D) = \Jmp(\T_D) - x = \Jmp(\T) - x = \Bal(\T) = \Bal(\W)$, whereas $\Ext(\W_D) = \Ext(\W)$ is the number of copies of $A_0$ between $B$ and $X$ in $\Copies\Unroll\Syn$.  Therefore, $\Lex^k(\W_D) = \Lex^k(\W)$ as desired.
 
 The proof for the case in which $\M$ is PIG is analogous, although loop unrolling is avoided.  We succinctly describe it here for the sake of completeness.  By Lemma~\ref{lem:cycle ext < 0}, $\Syn^k$ has a greedy nose path $\G$ that is internal, starts at $A_0$, and ends at $\LL(A_0)$.  The path $\W$ is also internal because it starts at $A_0$ and, thus, it cannot take the unique external nose $\LL(A_0) \to A_0$.  By Theorem~\ref{thm:planarity}, $\Gr_1(\G)$ is bounded below by $\Gr_1(\W)$.  Let $\W_H$ be the maximal greedy anti-hollow walk that is internal and has a drawing $\Gr_p(\W_H)$, $p\geq0$, that ends at the same point as $\Gr_1(\W)$.  By Theorem~\ref{thm:planarity}, $\Gr_p(\W_H)$ is bounded below by $\Gr_1(\W)$.  Moreover, $\Gr_p(\W_H)$ is bounded above by $x \mapsto \Rows(\M)-1$ because $\Ext(\W_H) \leq 0$ as $\M$ is PIG.  Then, $\Gr_1(\G)$ and $\Gr_p(\W_H)$ share a point $q$ because there is a finite number of vertex positions that $\Gr(\W_H)$ can traverse without either reaching the column 0 or taking the external nose.  Moreover, as above, the path $\W_D$ whose drawing takes $\Gr_1(\G)$ until $q$ and then takes $\Gr_1(\W_H)$ is dually greedy and has $\Lex^k(\W_D) = \Lex^k(\W)$.
\end{proof}

\begin{lemma}\label{lem:lex computation}
 If\/ $\M$ is a connected PCA model that is equivalent to a $k$-multiplicative model for $k \in \ORange{\wrap}$, then\/ $\Dist{\Lex^k}$ can be computed in $O(n)$ time.
\end{lemma}

\begin{proof}
 We prove that the next algorithm is linear and computes a function $\psi = \Dist{\Lex^k}$.
 \begin{enumerate}
  \item Let $\G = A_0, \ldots, A_p$ be the maximal greedy nose path from $A_0$.  For $i \in \CRange{p}$, let $\phi(A_i) = i$ and $\alpha(A_i) = \Lex^k(A_0, \ldots, A_i)$.  For $A \not\in \G$, let $\phi(A) = n+1$ and $\alpha(A) = (-\infty, -\infty)$.\label{alg:dlex:W}
  \item For $A \in \A(\M)$, let $x(A)$ be the unique vertex that precedes $A$ in every greedy anti-hollow path of $\Syn^k$ that traverses $A$.  Let $D$ be the digraph with a vertex $v(A)$ and an edge $v(x(A)) \to v(A)$ for every $A\in\A(\M)$.\label{alg:dlex:D}
  \item For each cycle $W$ of $D$, let $A_W$ be arc of $\A(\M)$ with minimum $\phi$ among those arcs $A$ such that $v(A) \in W$. The digraph $F$ that is obtained after removing $v(x(A_W)) \to A_W$ for every cycle $W$ of $D$ is a forest: each root has in-degree $0$, whereas $v(x(A))$ is the parent of $A$ for each edge $v(x(A)) \to v(A)$. \label{alg:dlex:F}
  \item The algorithm outputs the function $\psi$ below, that is well defined because $F$ is a forest:\label{alg:dlex:psi}
  \begin{equation*}
   \psi(A) = \begin{cases}
                   \alpha(A) & \text{if $A$ is a root of $F$} \\
                   \max\{\psi(x(A)) + \Lex^k(x(A) \to A), \alpha(A)\} & \text{otherwise} \\
                  \end{cases}
  \end{equation*}
 \end{enumerate}
 
 To see that the algorithm is correct, we prove that $\psi(A) = \Dist{\Lex^k}(A_0, A)$ for every $A \in \A(\M)$.  By Theorem~\ref{thm:tucker}, $\Lex^k(\W) < (0,0)$ for every cycle $\W$ of $\Syn^k$, thus $\Dist{\Lex^k}(A_0, A) = \IDist{\Lex^k}(A_0, A)$ is well defined.  By Lemma~\ref{lem:semi-greedy},
\begin{displaymath}
  \Dist{\Lex^k}(A_0, A) = \max\{\Lex^k(\W) \mid \W \text{ is a dually greedy hollow path from $A_0$ to $A$}\}.
\end{displaymath}
 In other words, there exists a dually greedy nose path $\G_D$ from $A_0$ to $A$ with $\Lex^k(\G_D) = \Dist{\Lex^k}(A_0, A)$.  By definition, $\G_D = \G' + \G_A'$ where $\G' = A_0, \ldots, A_q$, $q\in\CRange{p}$, is a subpath of the greedy nose path $\G$ computed at Step~\ref{alg:dlex:W} and $\G_A'$ corresponds to the path from $x(A_q)$ to $x(A)$ in the digraph $D$ computed at Step~\ref{alg:dlex:D}.  The proof that $\psi(A) = \Lex^k(\G_D)$ is by induction on the length of the unique path of $F$ from a root to $A$.  
 
 In the base case, $v(A)$ is a root of $F$.  Clearly, $\Lex^k(\G_D) \geq \alpha(A)$ by Step~\ref{alg:dlex:W}, whereas $\psi(A) = \alpha(A)$ by Step~\ref{alg:dlex:psi}.  Suppose, to obtain a contradiction, that $\Lex^k(\G_D) > \alpha(A)$.   Then, $A \neq A_q$ and, moreover, $B \neq A_i$ for every $B \in \G_A'$ and every $i \in \Range{q}$ because $\G_D$ is a path.  By Step~\ref{alg:dlex:W}, it follows that $\phi(B) > \phi(A_q)$, thus $\G_A'$ corresponds also to a path of the digraph $F$ computed at Step~\ref{alg:dlex:F}.  But this is impossible because $v(A)$ is a root of $F$.  Hence, $\Lex^k(\G_D) = \alpha(A) = \psi(A)$ when $v(A)$ is a root of $F$.
 
 In the inductive step, $v(x(A))$ is the parent of $v(A)$ in $F$.  Clearly, $\Lex^k(\G_D) \geq \alpha(A)$ by Step~\ref{alg:dlex:W}, whereas $\Lex^k(\G_D) \geq \Lex^k(\W) + \Lex^k(x(A) \to A)$ for every dually greedy nose path $\W$ from $A_0$ to $x(A)$.  Hence, $\Lex^k(\G_D) \geq \psi(A)$ follows by Step~\ref{alg:dlex:psi} and the inductive hypothesis.  Conversely, if $A = A_q$, then $\Lex^k(\G_D) = \alpha(A)$, whereas if $A \neq A_q$, then $\G_D = \W, A$ for a dually greedy path $\W$ from $A_0$ to $x(A)$.  Then,  $\Lex^k(\G_D) \leq \psi(A)$ also follows by Step~\ref{alg:dlex:psi} and the inductive hypothesis.

 Regarding the time complexity, $\Syn^k$ is computed in $O(n)$ time via Theorem~\ref{thm:Syn^k algorithm} before the algorithm is invoked.  Then, the greedy nose path $\G$ of Step~\ref{alg:dlex:W} can be obtained in $O(n)$ time, while $\phi$ and $\alpha$ are calculated in $O(n)$ time with a single traversal of $\G$.  Similarly, Step~\ref{alg:dlex:D} is implemented in $O(n)$ time by traversing the edges of $\Syn^k$ in a backward direction.  Step~\ref{alg:dlex:F} consumes $O(n)$ time as well as each vertex of $D$ has at most one in-neighbor.  Finally, $\psi$ is calculated in $O(n)$ time at Step~\ref{alg:dlex:psi} with a traversal of $F$ from the roots to its leaves.
\end{proof}

\begin{theorem}
 The problem $\kMult$ can be solved in $O(n)$ time for every PCA model\/ $\M$ and every $k \in \Range{\wrap}$.  The algorithm outputs either a $k$-multiplicative model equivalent to\/ $\M$ or a negative certificate that can be authenticated in $O(n)$ time. 
\end{theorem}

\begin{proof}
 If $k = 0$, the algorithm returns $\M$ in $O(1)$ time.  Otherwise, $O(n)$ time spent by Corollary~\ref{cor:kRep} to decide if $\M$ is equivalent to some $k$-multiplicative model.  If the answer is no, then a negative certificate is obtained as a by-product.  If the answer is yes, then the algorithm implied by Theorem~\ref{thm:tucker} (Section~\ref{sec:representation}) is executed to build the $k$-multiplicative model $\Unit$ equivalent to $\M$.  By Corollary \ref{cor:ratio computation} and Lemma~\ref{lem:lex computation}, this last step requires $O(n)$ time as well.
\end{proof}

\section{Concluding remarks}

In this article we designed a certifying and linear time algorithm to solve \kMult.  As a by-product, we obtained a certifying and $O(n^2)$-time algorithm %refchange-R5
for \klcMult.  From a theoretical point of view, we provided a new characterization of those PCA models that are equivalent to a $k$-multiplicative UCA model, for every $k < \wrap$.  The proof of this characterization exploits a powerful geometric framework given by Mitas' drawings and the loop unrolling technique.  Mitas' drawings allow us to treat the internal cycles of the synthetic graphs as if they were curves in $\mathbb{R}^2$.  The intersection of two curves corresponds to the intersection of the paths.  The loop unrolling technique, on the other hand, provides an internal copy of every cycle of the synthetic graph.  In a forthcoming article \cite[see the preprint]{SoulignacTerliskyC2017} we combine Mitas' drawings of the synthetic graphs with the loop unrolling technique to solve the minimal representation problem.  

In some sense, the algorithm that we provide is a generalization of the one given by Soluignac \cite{SoulignacJGAA2017,SoulignacJGAA2017a} %refchange-R4
that, in turn, generalizes the algorithm by Mitas~\cite{Mitas1994} %refchange-R3
for UIG graphs.  Even though many properties of UIG models hold naturally in UCA models, this is not always the case, as UCA models have a much richer structure than UIG models.  This is the case for many of algorithms that solve \RepUIG.  The fact that the algorithm for \RepUIG based on synthetic graphs generalizes to \Rep is a plus for this tool.  

As discussed in Section~\ref{sec:introduction}, every PIG model $\M$ is equivalent to an $\infty$-multiplicative UIG model.  This fact can be easily proven by looking at the algorithms by Corneil et al.~\cite{CorneilKimNatarajanOlariuSpragueIPL1995} and Lin et al.~\cite{LinSoulignacSzwarcfiter2009}, %refchange-R4
and it also follows by Theorem~\ref{thm:equivalence}.  We note that there exist PCA models that are not PIG and are equivalent to $\infty$-multiplicative models as well.  For instance, $\{(2i, 2(i+k)+1 \mod 2n) \mid 1 \leq i \leq n\}$ is an $\infty$-multiplicative model representing $C_n^k$, $n \gg k$, where $C_n$ is the cycle with $n$ vertices.

 \begin{figure}
  \mbox{}\hfill\includegraphics{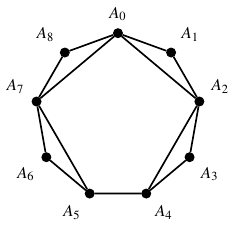}\hfill\raisebox{-5mm}{\includegraphics{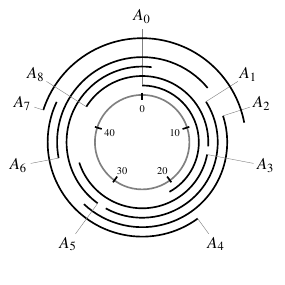}}\hfill\includegraphics{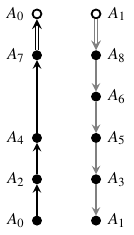}\mbox{}
  \caption{The graph $G$ in the left is $2$-UCA because the $(50,21)$-CA model in the center represents $G^2$.  However, the greedy cycles $A_0,A_2,A_4,A_7,A_0$ and $A_1,A_8,A_6,A_5,A_3,A_1$ of $\Syn^2(\M)$ do not intersect, for the PCA model $\M$ representing $G$ (that is unique up to full reversal and movement of $0$).}\label{fig:example-non-mult-power}
\end{figure}

Say that a PCA (resp.\ PIG) graph $G$ is $k$-UCA (resp.\ $k$-UIG), for $k \geq 0$, when $G^i$ is UCA (resp.\ UIG) for every $1 \leq i \leq k$.  By definition, if a PCA model $\M$ is equivalent to a $k$-multiplicative UCA model, then $G(\M)$ is $k$-UCA.  Theorem~\ref{thm:multiplicative PIG} implies that the converse is also true for PIG graphs: if $G$ is $k$-PIG, then $G$ is PIG and, thus, it admits an $\infty$-multiplicative UIG model.  One is tempted to think that the converse is also true for UCA graphs: if $G$ is $k$-UCA and $k < \wrap$, then $G$ is $k$-multiplicative.  Unfortunately, this is false (\cref{fig:example-non-mult-power}).  Note that $\infty$-UCA graphs is the subclass of UCA graphs closed under taking powers; its graphs can be recognized in $O(n^2)$ time.  Several classes of graphs that are closed under taking powers were studied, including PIG, interval, PCA, and circular-arc graphs~\cite{RaychaudhuriCN1987,FlotowDAM1996}.  Computing models representing powers of circular-arc graphs is an important problem with different applications~\cite{AgnarssonDamaschkeHalldorssonDAM2003}.  We leave open the problem of recognizing these graphs in $o(n^2)$ time.  

\acknowledgements{A preliminary version of this article was presented at CLAIO 2018, and some of the results presented here were made available on a preprint server~\cite{SoulignacTerliskyC2022}. This work was developed as part of the PhD thesis of the second author. During the review process of this manuscript, the thesis was successfully defended in 2024~\cite{Terlisky2024}. The thesis also incorporates findings from an earlier preprint~\cite{SoulignacTerliskyC2017}, which are planned for future publication.

The authors were supported by PICT ANPCyT grant 2015--2419.}

\end{document}